\newif\iffull
\fulltrue
\title{Constructions of High-Rate Minimum Storage\\Regenerating Codes over Small Fields}
\author{Netanel Raviv \and Natalia Silberstein \and Tuvi Etzion}
\date{\today}
\documentclass[12pt]{article}

\usepackage[margin=0.7in]{geometry}
\usepackage{amssymb}
\usepackage{multirow}
\usepackage{amsmath}
\usepackage{amsthm}
\usepackage{units}
\usepackage{color}
\usepackage{filecontents}
\usepackage{hyperref}
\usepackage{mathtools}
\usepackage[sort,numbers]{natbib}

\usepackage{tikz,pgf}
\usepackage{mathrsfs}
\usetikzlibrary{arrows}

\DeclarePairedDelimiter{\floor}{\lfloor}{\rfloor}

\newtheorem{theorem}{Theorem}
\newtheorem{definition}{Definition}

\newtheorem{lemma}{Lemma}
\newtheorem{corollary}{Corollary}

\newtheorem{example}{Example}

\newtheorem{remark}{Remark}

\newcommand{\AS}{(\cA,\cS)}

\newcommand{\bC}{\mathbb{C}}

\newcommand{\bF}{\mathbb{F}}

\newcommand{\bN}{\mathbb{N}}

\newcommand{\cA}{\mathcal{A}}
\newcommand{\cB}{\mathcal{B}}

\newcommand{\cS}{\mathcal{S}}

\newcommand{\cX}{\mathcal{X}}
\newcommand{\cY}{\mathcal{Y}}
\newcommand{\cZ}{\mathcal{Z}}

\newcommand{\Span}[1]{{\left\langle {#1} \right\rangle}}

\DeclareMathOperator{\rank}{rank}

\DeclareMathOperator{\image}{Im}

\makeatletter
\def\blfootnote{\xdef\@thefnmark{}\@footnotetext}
\makeatother

\begin{document}
\maketitle

\begin{abstract}
A novel technique for construction of minimum storage regenerating (MSR) codes is presented. Based on this technique, three explicit constructions of MSR codes are given. The first two constructions provide access-optimal MSR codes, with two and three parities, respectively, which attain the sub-packetization bound for access-optimal codes. The third construction provides longer MSR codes with three parities, which are not access-optimal, and do not necessarily attain the sub-packetization bound.

In addition to a minimum storage in a node, all three constructions allow the entire data to be recovered from a minimal number of storage nodes. That is, given storage $\ell$ in each node, the entire stored data can be recovered from any $2\log_2 \ell$ for 2 parity nodes, and either $3\log_3\ell$ or $4\log _3\ell$ for 3 parity nodes. Second, in the first two constructions, a helper node accesses the minimum number of its symbols for repair of a failed node (access-optimality). The generator matrix of these codes is based on perfect matchings of complete graphs and hypergraphs, and on a rational canonical form of matrices. The goal of this paper is to provide a construction of such optimal codes over the smallest possible finite fields. 
For two parities, the field size is reduced by a factor of two for access-optimal codes compared to previous constructions. For three parities, in the first construction the field size is $6\log_3 \ell+1$ (or $3\log_3 \ell+1$ for fields with characteristic~2), and in the second construction the field size is larger, yet linear in $\log_3\ell$. Both constructions with 3 parities provide a significant improvement over existing previous works, since only non-explicit constructions with exponential field size (in $\log_3\ell$) were known so far.
\end{abstract}

\footnotetext[1]{
\noindent		This research was supported in part by the Israeli
	Science Foundation (ISF), Jerusalem, Israel, under
	Grant~10/12.
	
\noindent	The work of Netanel Raviv was supported in part by the IBM fellowship. It is a part of his Ph.D.
	thesis performed at the Technion.
	
\noindent	The work of Natalia Silberstein was supported in part by the Fine fellowship.
	
\noindent	The authors are with the Department of Computer Science, Technion,
	Haifa 32000, Israel.
	
\noindent	e-mail: \{netanel,natalys,etzion\}@cs.technion.ac.il.}

\section{Introduction}\label{section:introduction}
\emph{Regenerating codes} are a family of erasure codes proposed by Dimakis et al.~\cite{dimakis} to store data in \emph{distributed storage systems} (DSSs) in order to reduce the amount of data (\emph{repair bandwidth}) downloaded during a repair of a failed node.
An $(n,k,\ell, d,\beta, \cB)_q$ regenerating code $C$, for $k\leq d\leq n-1$, $\beta\leq \ell$, is used to store a file of size $\cB$ in a DSS across a network of $n$ nodes, where each node of the system stores $\ell$ symbols from $\bF_q$, a finite field with $q$ elements, such that the stored file can be recovered by downloading  the data from any set of $k$ nodes.
When a single node fails, a \emph{newcomer} node which substitutes the failed node, contacts any set of $d$ nodes, called \emph{helper nodes}, and downloads $\beta$ symbols from  each helper node to reconstruct the data stored in the failed node. This process is called a \emph{node repair process} and the parameter $d$ is called the \emph{repair degree}. 
There are two general methods of node repairs: \emph{functional repair} and \emph{exact repair}. Functional repair ensures that when a node repair process is completed, the system is \emph{equivalent} to the original one, i.e., the stored file can be recovered from any $k$ nodes. However, the newcomer node may contain a different data from what was stored in the failed node. 
Exact repair requires that the newcomer node will store exactly the same data as was stored in the failed node. Usually, exact repair is required for systematic nodes (nodes that contain the actual data), while the parity nodes can be functionally repaired.

Based on a min-cut analysis of the information flow graph which represents a DSS, Dimakis et al.~\cite{dimakis} presented an upper bound on the size of a file that can be stored using a regenerating code under functional repairs,
\begin{equation*}
	\label{inf_flow_bound}
	\cB\leq \sum_{i=1}^k\min\{(d-i+1)\beta,\ell\}.
\end{equation*}
Given the values of $\cB,n,k,d$, this bound provides a tradeoff between the number $\ell$ of stored symbols in a node and the repair bandwidth $\beta d$. The extremal point on this tradeoff, where $\ell$ is minimized is referred to as \emph{minimum storage regenerating} (MSR) point, and a code that attains it, namely, a minimum storage regenerating (MSR) code satisfies~\cite{dimakis}
\begin{equation}
	\label{MSR code parameters}
	\left(\ell,\beta d\right)=\left(\frac{\cB}{k},\frac{\cB d}{k(d-k+1)}\right).
\end{equation}

The other extremal point, where $\beta\ell$ is minimized is referred to as \emph{minimum bandwidth regenerating} (MBR) point, and a code that attains it, namely, a minimum bandwidth regenerating (MBR) code satisfies~\cite{dimakis}
\begin{equation*}
	\label{MBR code parameters}
	\left(\ell,\beta d\right)=\left(\frac{2\cB d}{2k d-k ^2+k},\frac{2\cB d}{2k d-k^2+k}\right).
\end{equation*}
Constructions of MBR codes which support exact repair can be found for example in~\cite{rashmi,RSKR12,Rashmi09}.
Constructions of MSR codes which support exact repair can be found in~\cite{DimakisHadamard,rashmi,RSKR12,Rashmi09,suh1,LongMDS,Zigzag} and references therein.
Note that MSR codes are in particular MDS array codes~\cite{BlRo99,CaBr06}.

In this paper we focus on exact MSR codes which provide minimum repair bandwidth of systematic nodes and have additional properties listed below.

\begin{enumerate}
	\item \textbf{Maximum repair degree $d=n-1$}: this enables to minimize the repair bandwidth among all MSR codes. Such MSR codes satisfy $(\ell, \beta d)=\left(\frac{\cB}{k},\frac{\cB(n-1)}{k(n-k)}\right)$~\cite{dimakis}.
	\item \textbf{High rate $\frac{k}{n}$}: in particular  the number of parity nodes $r=n-k$ is  $r=2$ or $r=3$ (see e.g.~\cite{DimakisHadamard,Zigzag,LongMDS} for previously known constructions of such MSR codes).
	\item \textbf{Optimal access}:  the number of symbols  accessed in a helper node is minimal and equals to the number $\beta$ of symbols transmitted during node repair.
	(See~\cite{AccessKumar,AccessTWB} for bounds and constructions of access-optimal codes.)
	\item \textbf{Optimal sub-packetization factor}:  for an $(n,k,\ell, d,\beta, \cB)_q$ regenerating code, the number of stored symbols $\ell$ in a node is also  called  the \emph{sub-packetization} factor of the code. Low-rate ($\frac{k}{n}\leq \frac{1}{2}$, i.e., $\frac{r}{n}\geq \frac{1}{2}$) MSR codes with $d=n-1$,
	where $\ell$  is  linear in $r$ were constructed in~\cite{rashmi}. However, in the known  high-rate ($\frac{k}{n}>\frac{1}{2}$) MSR codes
	$\ell$ is exponential in $k$~\cite{DimakisHadamard, LongMDS}. Moreover, it was proved in~\cite{AccessTWB} that for an access-optimal code, given a fixed sub-packetization factor $\ell$ and $r$ parity nodes, the largest number $k$ of systematic nodes is
	\begin{equation}
		\label{eq:access-optimal-k}
		k=r\log_r\ell,
	\end{equation}
	i.e., the required sub-packetization factor~$\ell$ is~$r^{\frac{k}{r}}$.
	\item \textbf{Small finite field}: construction of access-optimal MSR codes with $r=2$  and optimal sub-packetization  $2^{\frac{k}{2}}$ over a finite field of size $1+2\log_2\ell$ is presented in~\cite{LongMDS}. More precisely, this construction provides a code with a larger number $k'=3\log_2\ell$ of systematic nodes out of which $k=2\log_2\ell$ have the optimal access property. Hence the shortened code with $k$ systematic nodes is an access-optimal code. For  general $r$, codes with sub-packetization factor $\ell=r^m$ and  $k=rm$, over $\bF_q$, where $q\geq k^{r-1}r^{m-1}+1$ were presented in~\cite{LongMDS} and codes for any $q\geq \binom{n}{k}r^{m+1}$ were presented in~\cite{AccessKumar}. 
	In addition, a construction for $r=2$ which achieves $k=2\log_2\ell$ and requires $q\ge \log_2\ell+1$ is presented in~\cite{framework}. This construction requires $q$ to be even, and while it does not provide the optimal access property, it provides a different property called \textit{optimal update}\footnote{DSS codes which satisfy the property
		that any change in the original data requires minimal updates at the storage nodes are called optimal update
		codes.}.
\end{enumerate}

We propose a construction of access-optimal MSR codes with optimal sub-packetization factor $\ell=r^m$, $k=rm$, for $r=2$ and $r=3$, over any finite field $\bF_q$ such that $q\geq m+1$ and $q\geq 6m+1$, respectively. Moreover, for $r=3$, if $q$ is a power of 2 then the field size can be reduced to $q\geq 3m+1$. In addition, we present a construction of a longer code which is not access-optimal, with $r=3$, $k=4m$, and $q=\Theta(m)$.

The comparison of the results presented in this paper with some previously known MSR codes can be found in the following tables.

\begin{table*}[ht]
	{\centering
		\begin{tabular}{|c|c|c|c|c|c|}
			\hline
			&$\bC$ (Theorem~\ref{theorem:main}) & \cite{framework} & \cite{LongMDS} (long MDS code) & \cite{AccessKumar} & \cite{Zigzag} (zigzag code) \\ \hline\hline
			$\ell$       &               $2^m$               &      $2^m$       &             $2^m$              &       $2^m$        &            $2^m$            \\ \hline
			$k$        &               $2m$                &       $2m$       &              $3m$              &        $2m$        &            $m+1$            \\ \hline
			\multicolumn{1}{|c|}{access} &\multirow{2}{*}{$\checkmark$} &  \multirow{2}{*}{for $m$ nodes}  & \multirow{2}{*}{for $2m$ nodes} & \multirow{2}{*}{$\checkmark$}  & \multirow{2}{*}{$\checkmark$} \\ 
			\multicolumn{1}{|c|}{optimality} &~&~&~&~&~\\\hline
			field size $q$   &               $m+1$               &    even $m+1$    &             $2m+1$             &       $2m+1$       &             $3$             \\ \hline
		\end{tabular} \vspace{0.1cm}\caption{Comparison of our codes with some previously known MSR codes with two parities.}\label{table1}
	}
\end{table*}

\begin{table*}[ht]
	{\centering
		\begin{tabular}{|c|c|c|c|c|c|c|}
			\hline
			
			\multirow{2}{*}{~}&\multicolumn{1}{|c|}{$\bC_1$} & \multicolumn{1}{|c|}{$\bC_2$}  & \multicolumn{1}{|c|}{\cite{LongMDS}} & \multirow{2}{*}{\cite{AccessKumar}} & \multicolumn{1}{|c|}{\cite{Zigzag}}  \\ 
			~&(Theorem~\ref{theorem:3ParitiesAccessOptimal})&
			  (Theorem~\ref{theorem:theorem:3ParitiesNotAccessOptimal})&
			  (long MDS code)&~&
			  (zigzag code)\\ \hline\hline			
			$\ell$              &                   $3^m$                   &                   $3^m$                   &              $3^m$              &                   $3^m$                   &             $3^m$               \\ \hline
			$k$                &                   $3m$                    &                   $4m$                    &              $4m$               &                   $3m$                    &             $m+1$               \\ \hline
			\multicolumn{1}{|c|}{access}   &       \multirow{2}{*}{$\checkmark$}       &      \multirow{2}{*}{for $3m$ nodes}      & \multirow{2}{*}{for $3m$ nodes} &       \multirow{2}{*}{$\checkmark$}       & \multirow{2}{*}{$\checkmark$}   \\
			\multicolumn{1}{|c|}{optimality} &                     ~                     &                     ~                     &                ~                &                     ~                     &  \\ \hline
			\multirow{2}{*}{field size $q$}  &    \multicolumn{1}{|c|}{even $3m+1$;}     &       \multirow{2}{*}{$\Theta(m)$}        & \multirow{2}{*}{$m^23^{m+1}+1$} & \multirow{2}{*}{$\binom{3m+3}{3}3^{m+1}$} &     \multirow{2}{*}{$4$ }       \\
			&     \multicolumn{1}{|c|}{odd $6m+1$}      &                     ~                     &                ~                &                     ~                    ~&\\ \hline
		\end{tabular}\vspace{0.1cm}\caption{Comparison of our codes with some previously known MSR codes with three parities.}\label{table2}
	}
\end{table*}

\subsection{Organization}
The construction for $r=2$ is given in Section~\ref{section:twoParities}
 and for $r=3$ in Section~\ref{section:construction3}. Additional construction for $r=3$, which does not have the access-optimal property, is given is Section~\ref{section:rPlus1}. 
Section~\ref{section:preliminaries} and Section~\ref{section:ourTechniques} present the techniques and underlying ideas used in Sections~\ref{section:twoParities},~\ref{section:construction3}, and~\ref{section:rPlus1}. Subsection~\ref{section:background} contains some necessary mathematical background, Subsection~\ref{section:subspaceCondition} describes the underlying algebraic problem, and Section~\ref{section:ourTechniques} explains the outline for all constructions.

\section{Preliminaries}\label{section:preliminaries}
\subsection{Algebra of Matrices - Background and Notations}\label{section:background}
The constructions in Section~\ref{section:twoParities} and Section~\ref{section:construction3} extensively use several standard linear-algebraic notions. For the sake of completeness, we include below a short introduction about these necessary notions. Some of the given background is not directly used in the constructions, but may assist the reader with understanding our techniques, and their underlying reasoning.

For a prime power $q$, $\bF_q^*$ is the set $\bF_q\setminus\{0\}$, $\bF_q^\ell$ is a vector space of dimension $\ell$ over $\bF_q$, which consists of vectors of length $\ell$, and $\bF_q^{\ell\times\ell}$ is the set of all $\ell\times\ell$ matrices with entries in $\bF_q$. It is widely known that a matrix $M\in\bF_q^{\ell\times\ell}$ admits \textit{(left\footnote{Unless otherwise stated, all multiplications of a vector $v$ by a matrix $M$ in this paper are from the left, i.e., $vM$.}) eigenvectors} and \textit{eigenvalues}~\cite[Section VII.7]{Algebra}. If $v\in\bF_q^\ell$ and $vM=\lambda v$ for some $\lambda\in\bF_q$, then $v$ is called a (left) eigenvector for the eigenvalue $\lambda$. The linear span  of all eigenvectors for a certain eigenvalue $\lambda$ is a subspace of $\bF_q^\ell$, and it is called a \textit{(left) eigenspace} of $M$.


For a subspace $S$ of $\bF_q^\ell$, let $S M\triangleq\{sM~\vert~s\in S\}$. The set $SM$ is obviously a subspace of $\bF_q^\ell$, and if $M$ is invertible then $\dim S=\dim (SM)$. A subspace $S$ which satisfies $SM=S$ is called \textit{a (left) invariant subspace of $M$}~\cite[Section XI.4]{Algebra} (in short, $S$ is $M$ invariant). Clearly, an eigenspace of $M$ is also an invariant subspace of $M$, but not necessarily vice versa.

For a polynomial $p(x)\in\bF_q[x]$ such that $p(x)=\sum_{i=0}^{d}p_ix^i$, let $p(M)\triangleq \sum_{i=0}^{d}p_i M^i$. The \textit{characteristic polynomial} $c(x)$ of $M$ is the determinant of $M-xI$~\cite[Section IX.5]{Algebra}, where $I$ is the $\ell\times\ell$ identity matrix, i.e. $c(M)=0$. Furthermore, there exists a unique monic polynomial $m(x)\in\bF_q[x]$, of minimum degree, such that $m(M)=0$. The polynomial $m(x)$, called the \textit{minimal polynomial of~$M$}, divides the characteristic polynomial $c(x)$ of $M$, and its roots are the eigenvalues of $M$.

If $P\in\bF_q^{\ell\times\ell}$ is an invertible matrix, then the matrices $P^{-1}MP$ and $M$ are called \textit{similar matrices}, and the matrix $P$ is called a \textit{change matrix}, (or a \textit{change-of-basis} matrix)~\cite[Section VII.7]{Algebra}. It is easily verified that if $e_0,\ldots,e_{\ell-1}$ is the standard basis of $\bF_q^\ell$, and $p_0,\ldots,p_{\ell-1}$ are the rows of $P$, then $P^{-1}MP$ acts on $p_0,\ldots,p_{\ell-1}$ exactly as $M$ acts on $e_0,\ldots,e_{\ell-1}$. That is, if
\begin{eqnarray}\label{eqn:algebraicFact}
	\left(\sum_{i=0}^{\ell-1}\mu_i e_i\right)M=\sum_{i=0}^{\ell-1}\delta_i e_i
\end{eqnarray}
for some coefficients $(\mu_i)_{i=0}^{\ell-1}$ and $(\delta_i)_{i=1}^{\ell-1}$, then
\begin{eqnarray*}
	\left(\sum_{i=0}^{\ell-1}\mu_i p_i\right)\left(P^{-1}MP\right)=\sum_{i=0}^{\ell-1}\delta_i p_i.
\end{eqnarray*}
As a result of this fact, we have that similar matrices share the same eigenvalues, but not necessarily the same eigenvectors. In addition, similar matrices also share the same minimal polynomial~\cite[Section IX.7]{Algebra}.

Determining matrix similarity is possible by converting given matrices to one of several canonical forms. One such canonical form, which does not always exist, is the \textit{diagonal} form. If a matrix $M$ is similar to a diagonal matrix then $M$ is called \textit{a diagonalizable} matrix. It is well known that $M$ is diagonalizable if and only if there exists a basis of $\bF_q^\ell$ in which all vectors are eigenvectors of $M$~\cite[Section~VII.8, Theorem 19]{Algebra}, and two matrices with the same diagonal form are similar. Two diagonalizable matrices $A$ and $B$ are called \textit{simultaneously diagonalizable} if there exists an invertible matrix $P$ such that $P^{-1}AP$ and $P^{-1}BP$ are diagonal matrices. Equivalently, $A$ and $B$ are simultaneously diagonalizable if there exists a basis of $\bF_q^\ell$ whose vectors are eigenvectors of $A$ and of~$B$. Since diagonal matrices commute, it is easy to prove that simultaneously diagonalizable matrices commute as well.

Determining matrix similarity for matrices which are not necessarily diagonalizable is a corollary of the so-called \textit{decomposition theorem}~\cite[Section XI.4, Theorem 8]{Algebra}, one of the profoundest results in linear algebra. In order to state this theorem, we require the notions of a \textit{block diagonal matrix} and a \textit{companion matrix}. A block diagonal matrix is a block matrix (that is, a matrix which is interpreted as being partitioned to submatrices) in which the only non-zero sub-matrices are on the main diagonal. A companion matrix of a polynomial $p(x)$ is a $\deg p\times \deg p$ matrix consisting of $1$'s in the main sub-diagonal, the additive inverses of the coefficients of $p$ in the rightmost column, and $0$ elsewhere.

The decomposition theorem states that any matrix $M$ is similar to a block diagonal matrix, whose blocks are companion matrices of certain factors of the characteristic polynomial. The polynomials corresponding to these companion matrices may be ordered such that any polynomial is a multiple of the next, and the first one is the minimal polynomial of $M$. This block diagonal matrix is called \textit{the rational canonical form} (rational form, in short) of $M$, any matrix $M'$ is similar to $M$ if and only if the share the same rational form, which exists over any field.

In what follows, for a set of row vectors $T$ we denote its $\bF_q$-linear span by $\Span{T}$ and for subspaces $U$ and $V$, let $U+V\triangleq\{u+v~\vert~u\in U,v\in V\}$. For a matrix $M$ we denote its (left) image by $\image(M)\triangleq\{vM~\vert~v\in\bF_q^\ell\}$ and its row span by $\Span{M}$.
\subsection{The Subspace Condition}\label{section:subspaceCondition}
Usually, a distributed storage system has a \textit{systematic part}, i.e. certain nodes in the system should store an uncoded part of the data. Such nodes are called \textit{systematic nodes}, and they allow instant access to their stored data. An efficient repair algorithm for a failed systematic node is vital. In this paper, we devise an MSR code which allows a minimum repair bandwidth for a failed systematic node.

This problem was previously studied by~\cite{LongMDS,SubpacketizationBound_,ImprovedSubpacketizationBound,AccessTWB}, where it was shown to be equivalent to a purely algebraic condition called \textit{the subspace condition}\footnote{\cite{LongMDS,AccessTWB} used the term ``subspace property''.}. In this subsection we describe this condition, and explain why codes which satisfy it provide minimum repair bandwidth for a failed systematic node. We refer the interested reader to~\cite{LongMDS} for a proof that the subspace condition is also necessary. A more general formulation of the subspace condition, which is irrelevant in our context, may also be found in~\cite{LongMDS}.

In an MSR code with $k$ systematic nodes, $r$ parity nodes, sub-packetization $\ell$, and maximum repair degree $d=n-1$, a file $\textbf{f}\in\bF_q^{k\ell}$ is partitioned into $k$ parts of length $\ell$ each, denoted by
$\textbf{f} = (C_1,\ldots,C_k)$. The file $\textbf{f}$ is multiplied by a $k\ell\times (k+r)\ell$ generator block matrix of the form
	
	\begin{eqnarray}\label{equation:generatorMatrixGeneral}
		\begin{pmatrix}
			I &~     &~     &A_{1,0}&A_{1,1}&\cdots&A_{1,r-1}\\
			~ &\ddots&~&\vdots&\vdots&\ddots&\vdots\\
			~ &~     &I     &A_{k,0}&A_{k,1}&\cdots&A_{k,r-1}\\
		\end{pmatrix}
	\end{eqnarray}
	where $I$ is the $\ell\times\ell$ identity matrix, and the $A_{i,j}$'s are invertible matrices, which satisfy a certain set of properties~\cite{LongMDS}. In this paper $A_{i,j}=A_i^j$ for some $A_1,\ldots,A_k$ that will be defined in the sequel. That is, \eqref{equation:generatorMatrixGeneral} simplifies to 
	
	\begin{eqnarray}\label{equation:generatorMatrix}
		\begin{pmatrix}
			I &~     &~     &I&A_1   &\cdots&A_1^{r-1}\\
			~ &\ddots&~&I&\vdots&\ddots&\vdots\\
			~ &~     &I     &I&A_k   &\cdots&A_k^{r-1}\\
		\end{pmatrix}.
	\end{eqnarray}

The resulting codeword is partitioned into $k+r$ columns of length $\ell$ each, denoted $(C_1,\ldots,C_k,C_{k+1}, $ $\ldots,C_{k+r})$, where for all $j\in[r]\triangleq\{1,\ldots,r\}$,
\[
C_{k+j}=\sum_{i=1}^k A_i^{j-1}C_i.
\]
Each column $C_i$ is stored in a different storage node, where the first $k$ nodes are the systematic ones and the remaining~$r$ nodes are called \textit{parity} nodes.

Upon a failure of a systematic node $m\in[k]$, storing $C_m$, it is required to repair it by downloading a minimal amount of data. According to~(\ref{MSR code parameters}), since $n=k+r$ and $d=n-1$, we have that the minimum bandwidth $\beta d$ in this scenario is
\begin{eqnarray}\label{equation:minimumBandwidth}
	\beta d=\frac{\cB}{k}\cdot \frac{d}{d-k+1}=\frac{k\ell}{k}\cdot\frac{k+r-1}{r}=\frac{\ell}{r}(k+r-1).
\end{eqnarray}
That is, each of the remaining $k+r-1$ nodes should contribute $1/r$ of its stored data~\cite{dimakis}. Sufficient conditions for this minimum repair bandwidth are as follows.

\begin{definition}\label{definition:theSubspaceCondition}(The Subspace Condition, \cite[Section II]{AccessTWB} Let $\ell$ and $r$ be integers such that $r$ divides $\ell$. A set of pairs $\{(A_i,S_i)\}_{i=1}^k$, where for all $i$, $A_i$ is an invertible $\ell\times\ell$ matrix and $S_i$ is an $\ell/r$-subspace of $\bF_q^\ell$, satisfies the subspace condition if the following properties hold.
	\begin{description}
		\item[The \textit{independence} property:]for each $i\in[k]$, \[S_i+S_iA_i+S_iA_i^2+\ldots+S_iA_i^{r-1}=\bF_q^\ell.\]
		\item[The \textit{invariance} property:] for all $i,j\in[k]$, $i\ne j$, $S_iA_j= S_i$.
		\item[The \textit{nonsingular} property:] Every square block submatrix of the following block matrix is invertible.
		\begin{eqnarray}\label{equation:nonSystematicPart}
			\begin{pmatrix}
				I & A_1 & \cdots & A_1^{r-1}\\
				\vdots & \vdots &  \ddots & \vdots\\
				I & A_k &  \cdots & A_k^{r-1}
			\end{pmatrix}
		\end{eqnarray}
	\end{description}
\end{definition}

If a subspace $S$ satisfies the invariance property for a matrix $A$, then $S$ is an \textit{invariant subspace} of $A$ (see Section~\ref{section:background}). 
If a subspace $S'$ satisfies the independence property for $A$, then $S'$ is an \textit{independent subspace} of $A$. Notice that the nonsingular property must hold for the code to be an MDS array code~\cite{BlRo99,CaBr06}, regardless of any applications in distributed storage.

\begin{theorem}\label{theorem:subspaceCondition}\cite{LongMDS}
	If the set $\{(A_i,S_i)\}_{i=1}^k$ satisfies the subspace condition for given $\ell$ and $r$, then the code whose generator matrix is given in~(\ref{equation:generatorMatrix}) is an MSR code which allows a minimum repair bandwidth for any systematic node.
\end{theorem}

The subspaces $\{S_i\}_{i=1}^k$ in this theorem are used in the repair process, and are often called \textit{repair subspaces}. To repair a systematic node~$j$, the remaining nodes project their data on $S_j$, i.e. multiply their data by some full rank matrix whose row span is $S_j$, and send it to the newcomer. For additional details see~\cite{LongMDS,SubpacketizationBound_,ImprovedSubpacketizationBound,AccessTWB}.

In order to compute the projections on the subspace $S_j$, each of the remaining nodes must access a certain amount of its stored symbols, and clearly, at least $\ell/r$ symbols must be accessed. A code in which this minimum is attained is called an access-optimal code~\cite{Zigzag,LongMDS,framework}. It can be shown that a code is access-optimal if and only if each subspace $S_j$ has a basis which consists of unit vectors only~\cite[Section~V]{AccessTWB}.

A set of the form $\{(A_i,S_i)\}_{i=1}^k$ is called \textit{an $\AS$-set}. Since the subspace condition is necessary and sufficient for construction of MSR codes, this paper will focus solely on the construction of $(\cA,\cS)$-sets which satisfy it.

\section{Our Techniques}\label{section:ourTechniques}
Our constructions rely on the properties of some matrix $A$, to which the matrices in our $\AS$-set are similar using certain change matrices. These change matrices are defined according to a set of matchings in the complete $r$-uniform hypergraph on $\ell$ vertices $K_\ell^r$. In this subsection the matrix $A$ is described, its properties are discussed, and the use of matchings for the definition of the change matrices is explained.

The matrix $A$ and the change matrices will be described with respect to a construction with $r$ parities, for a general $r$. In the following sections, the cases of $r=2$ and $r=3$ will be discussed in detail.

For a given number of parities $r$ and an integer $m$, the matrix $A$ is an $r^m\times r^m$ block diagonal matrix whose constituent blocks are the $r\times r$ companion matrix of $x^r-1$. That is, 

\begin{equation*}
	C\triangleq
	\begin{pmatrix}
		0 & \cdots & 0 & 1\\
		1 & \cdots & 0 & 0\\
		0 & \cdots & 0 & 0\\
		0 & \cdots & 1 & 0
	\end{pmatrix},
\end{equation*}
and the matrix $A$ is
\begin{equation} \label{equation:A}
	A\triangleq
	\begin{pmatrix}
		C & ~ & ~\\
		~ & \ddots & ~\\
		~ & ~ & C\\
	\end{pmatrix}.
\end{equation}

Since it is desirable that $A$ will have as many eigenspaces as possible, we operate over a field $\bF_q$, where $r|q-1$. This assumption about~$q$ provides the existence of all roots of unity $1,\gamma_1,\ldots,\gamma_{r-1}$ of order $r$ in the field $\bF_q$ (using the well-known Sylow theorems~\cite[Section XII.5]{Algebra}). It is readily verified that the eigenvalues of $A$ are $1,\gamma_1,\ldots,\gamma_{r-1}\in\bF_q$, since they are the roots of the minimal polynomial $x^r-1$ of $A$. We note that for the special case of $r=2$ (Section~\ref{section:twoParities}), we use an additional technique which allows to operate with any $q\ge m+1$, \textit{without} requiring that $2|q-1$.

In what follows we present the structure of the eigenspaces of $A$, and the eigenspaces of matrices which are similar to $A$.
	
	\begin{lemma}\label{lemma:matrixC}
		The matrix $C\in\bF_q^{r\times r}$ is a diagonalizable matrix whose set of linearly independent eigenvectors is $\{(1,\ldots,1)\}\cup\{(1,\gamma_i,\gamma_i^2,\ldots,\gamma_i^{r-1})\}_{i=1}^{r-1}$. Furthermore, the subspace $S=\Span{e_0}$ is an independent subspace of $C$, i.e. $S+SC+\ldots+SC^{r-1}=\bF_q^r$.
	\end{lemma}
	
	\begin{proof}
		It is readily verified that for all $i\in\{0,\ldots,r-1\}$, $e_iC=e_{i-1\bmod r}$. 
		Hence, $(1,\ldots,1)$ is an eigenvector for the eigenvalue 1. In addition, for any $i\in[r-1]$, we have that
		\begin{eqnarray*}
			(1,\gamma_i,\gamma_i^2,\ldots,\gamma_i^{r-1})C&=&\phantom{\gamma_i}(\gamma_i,\gamma_i^2,\ldots,\gamma_i^{r-1},1)\\
			&=&\gamma_i(1,\gamma_i,\gamma_i^2,\ldots,\gamma_i^{r-1}),
		\end{eqnarray*}
		and thus the vector $(1,\gamma_i,\gamma_i^2,\ldots,\gamma_i^{r-1})$ is an eigenvector which corresponds to the eigenvalue $\gamma_i$. These eigenvectors form an $r\times r$ Vandermonde matrix~\cite[p.~270]{LN97-FiniteFields}, and hence they are linearly independent and $C$ is diagonalizable. Since for all $i\in\{0,\ldots,r-1\}$, $e_iC=e_{i-1\bmod r}$, we also have that $S+SC+\ldots+SC^{r-1}=\bF_q^r$.
	\end{proof}
	
	The structure of the eigenspaces and eigenvalues of $A$ is a simple corollary of Lemma~\ref{lemma:matrixC}.
	
	\begin{corollary}\label{corollary:matrixA}
		The matrix $A$ is diagonalizable, and the linearly independent eigenvectors of $A$ are as follows.
		\begin{enumerate}
			\item \label{item:eigenspace1} For the eigenvalue $1$, the $\ell/r$ linearly independent eigenvectors are
			\[			\begin{array}{llcl}
			(1,1,\ldots,1,&0,0,\ldots,0,&\cdots&,0,0,\ldots,0)\\
			(0,0,\ldots,0,&1,1,\ldots,1,&\cdots&,0,0,\ldots,0)\\
			~&~&\ddots&\\
			(0,0,\ldots,0,&0,0,\ldots,0,&\cdots&,1,1,\ldots,1)
			\end{array}\]
			\item \label{item:eigenspacer1}For the eigenvalue $\gamma_i,i\in[r-1]$, the $\ell/r$ linearly independent eigenvectors are
\[			\begin{array}{llcl}
(1,\gamma_i,\ldots,\gamma_i^{r-1},&0,0,\ldots,0,&\cdots&,0,0,\ldots,0)\\
(0,0,\ldots,0,&1,\gamma_i,\ldots,\gamma_i^{r-1},&\cdots&,0,0,\ldots,0)\\
~&~&\ddots&\\
(0,0,\ldots,0,&0,0,\ldots,0,&\cdots&,1,\gamma_i,\ldots,\gamma_i^{r-1})
\end{array}\]
		\end{enumerate}
		In addition, if $S\triangleq\Span{e_0,e_r,e_{2r},\ldots,e_{\ell-r}}$ then $S+SA+SA^2+\ldots+SA^{r-1}=\bF_q^\ell$.
	\end{corollary}
	
	The matrices in our construction are similar to the matrix $A$. The following lemma, that is based on~\eqref{eqn:algebraicFact}, presents the eigenvalues and eigenvectors of a matrix which is similar to~$A$.
	
	\begin{lemma}\label{lemma:eigenspacesSimilar}
		If $P\in\bF_q^{\ell\times\ell}$ is an invertible matrix whose rows are $p_0,\ldots,p_{\ell-1}$, and $B\triangleq P^{-1}A P$, then $B$ is diagonalizable, with the following eigenspaces,
		\begin{itemize}
			\item[\emph{A1.}] For the eigenvalue $1$, a basis of the eigenspace is $\left\{\sum_{j=0}^{r-1}p_{ir+j}~\big\vert~i\in\{0,\ldots,\ell/r-1\}\right\}$.
			\item[\emph{A2.}] For the eigenvalue $\gamma_i$, a basis of the eigenspace is $\left\{\sum_{j=0}^{r-1}\gamma_i^{j}p_{ir+j}~\big\vert~i\in\{0,\ldots,\ell/r-1\}\right\}$.
		\end{itemize}
		In addition, the subspace $T\triangleq \Span{p_0,p_r,p_{2r},\ldots,p_{\ell-r}}$ satisfies $T+TB+TB^2+\ldots+TB^{r-1}=\bF_q^\ell$.
	\end{lemma}
	
	\begin{proof}
		Notice that the vectors in A1 and A2 are given by multiplying the eigenvectors of $A$ (see Corollary~\ref{corollary:matrixA}) by the matrix $P$. If $v$ is an eigenvector of $A$ which corresponds to an eigenvalue $\lambda$ then
		\begin{eqnarray*}
			(vP)B=(vP)P^{-1}AP=vAP=\lambda vP,
		\end{eqnarray*}
		and hence $vP$ is an eigenvector of $B$.
		
		Notice that the subspace $T$ may be written as $T=\Span{e_0P,e_rP,\ldots,e_{\ell-r}P}=\Span{e_0,e_r,\ldots,e_{\ell-r}}P=SP$, where $S=\Span{e_0,e_r,e_{2r},\ldots,e_{\ell-r}}$. Hence, it follows from Corollary~\ref{corollary:matrixA} that
		\begin{eqnarray*}
			T+TB+TB^2+\ldots+TB^{r-1}&=&SP+SPB+SPB^2+\ldots+SPB^{r-1}\\
			&=&SP+SAP+SA^2P+\ldots+SA^{r-1}P\\
			&=&(S+SA+SA^2+\ldots+SA^{r-1})P=\bF_q^\ell P=\bF_q^\ell.
		\end{eqnarray*}
	\end{proof}

The matrices in our construction are similar to the matrix $A$. The change matrices which induce the similarity are defined using \textit{perfect colored matchings} in the \textit{complete $r$-uniform hypergraph}. Although the specific choice of these change matrices varies from one construction to another, the general idea behind the use of matchings is roughly identical, and will be explained in the remainder of this subsection.

\begin{definition}\label{definition:matching}
	A perfect colored matching (matching, in short) is a perfect matching in the $r$-uniform hypergraph, whose vertices are colored in $r$ colors, such that no edge contains two nodes of the same color (see Figure~\ref{figure:oneMatching}).
\end{definition}

We denote a matching by $\cZ=(Z^{(0)},\ldots,Z^{(r-1)})$, where each $Z^{(i)}$ is an ordered color set (i.e. a subset of all vertices which are colored in the same color), and if $Z^{(i)}=(z_0^{(i)},\ldots,z_{\ell/r-1}^{(i)})$ for each $i\in\{0,\ldots,r-1\}$, then the edges of $\cZ$ are \[\left\{ \{z_j^{(0)},z_j^{(1)},\ldots,z_j^{(r-1)}\}\right\}_{j=0}^{\ell/r-1}.\]For example, for $r=2$, a matching is denoted by $\cZ=(Z,Z')$ (we use $Z$ and $Z'$ instead of $Z^{(0)}$ and $Z^{(1)}$ for convenience), where $Z=(z_0,\ldots,z_{\ell/2-1})$, $Z'=(z_0',\ldots,z_{\ell/2-1}')$, and the edges of $\cZ$ are $\left\{ \{ z_i,z_i'\}\right\}_{i=0}^{\ell/2-1}$.

\begin{figure}[here]
	\begin{center}
		\begin{tikzpicture}[scale=0.5][line cap=round,line join=round,>=triangle 45,x=1.0cm,y=1.0cm]
		\clip(11.,2.) rectangle (22.,13.);
		\draw(13.,11.) circle (1.cm)node[anchor=center] {$\heartsuit$};
		\draw(20.,11.) circle (1.cm) node[anchor=center] {$\diamondsuit$};
		\draw(13.,4.) circle (1.cm)node[anchor=center] {$\diamondsuit$};
		\draw(20.,4.) circle (1.cm)node[anchor=center] {$\heartsuit$};
		\draw [line width=4.4pt] (14.,11.)-- (19.,11.);
		\draw [dotted] (20.,10.)-- (20.,5.);
		\draw [dotted] (13.,10.)-- (13.,5.);
		\draw [line width=4.4pt] (14.,4.)-- (19.,4.);
		\draw [dotted] (19.292893218813454,10.292893218813452)-- (13.707106781186548,4.707106781186548);
		\draw [dotted] (13.707106781186548,10.292893218813452)-- (19.292893218813454,4.707106781186548);
		\draw (12.733998968070003,11.611143997505334);
		\draw (19.78068891530595,4.591873077301424);
		\draw (19.78068891530595,11.55630594344124) ;
		\draw (12.788837022134095,4.591873077301424) ;
		\end{tikzpicture}\caption{A perfect colored matching on the complete (2-uniform hyper)graph with four vertices. The edges of the matching appear in bold, and the nodes are colored in $\heartsuit,\diamondsuit$ such that there is no monochromatic edge.}\label{figure:oneMatching}
	\end{center}
\end{figure}
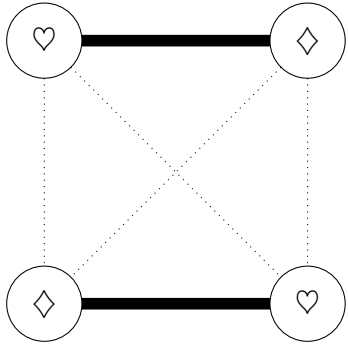

Each matching will be used to construct $r$ (or $r+1$ in Section~\ref{section:rPlus1}) change matrices for the $\AS$-set. Each $\ell\times\ell$ change matrix is constructed using constituent $r\times\ell$ sub-matrices~\eqref{equation:frameworkP}. Each such submatrix is a function of a single edge in the matching. That is, if the matching is $\cZ=(Z^{(0)},\ldots,Z^{(r-1)})$, then $r$ matrices in the $\AS$-set are constructed as $A_i=P_i^{-1}AP_i$, where
\begin{eqnarray}\label{equation:frameworkP}
	P_i=
	\begin{pmatrix}
		~\\\boxed{
			\begin{matrix}
				\mbox{An } r\times\ell \mbox{ submatrix based on }~\\
				\{z_0^{(0)},z_0^{(1)},\ldots,z_0^{(r-1)}\}
			\end{matrix}
		}\\
		~\\
		\boxed{
			\begin{matrix}
				\mbox{An } r\times\ell \mbox{ submatrix based on }~\\
				\{z_1^{(0)},z_1^{(1)},\ldots,z_1^{(r-1)}\}
			\end{matrix}
		}\\~
		\vdots\\
		\boxed{
			\begin{matrix}
				\mbox{An } r\times\ell \mbox{ submatrix based on }~\\
				\{z_{\ell/r-1}^{(0)},z_{\ell/r-1}^{(1)},\ldots,z_{\ell/r-1}^{(r-1)}\}
			\end{matrix}
		}\\~
	\end{pmatrix}.
\end{eqnarray}

The vertices of the $r$-uniform hypergraph $K_\ell^r$ are identified with the $\ell$ unit vectors $e_0,\ldots,e_{\ell-1}$. In all subsequent constructions, the subspaces in the $\AS$-set are defined using the color sets from the matchings, i.e., if $\cZ=(Z^{(0)},\ldots,Z^{(r-1)})$ is a matching, then we define $r$ subspaces ($r+1$ in Section~\ref{section:rPlus1}) of dimension $\ell/r$ as follows
\begin{eqnarray*}
	\forall i\in\{0,\ldots,r-1\},~S_{Z^{(i)}}&\triangleq& \Span{Z^{(i)}}\\
	S_{Z^*} &\triangleq& \Span{ \left\{ z_i^{(0)}+z_i^{(1)}+\cdots+z_i^{(r-1)}\right\}_{i=0}^{\ell/r-1}}.
\end{eqnarray*}

That is, each subspace $S_{Z^{(i)}}$ is the span of the color set $Z^{(i)}$, and the additional subspace $S_{Z^*}$ is the span of the sums of each edge in $\cZ$. To enlarge the $\AS$-set, different matchings can be used, as long as they satisfy the following simple condition.

\begin{definition}\label{definition:pairingCondition}
	Two matchings $\cX=(X^{(0)},\ldots,X^{(r-1)})$ and $\cY=(Y^{(0)},\ldots,Y^{(r-1)})$ satisfy the pairing condition if any edge in $\cX$ is monochromatic in $\cY$, and vice versa (see Figure~\ref{figure:twoMatchings}).
\end{definition}

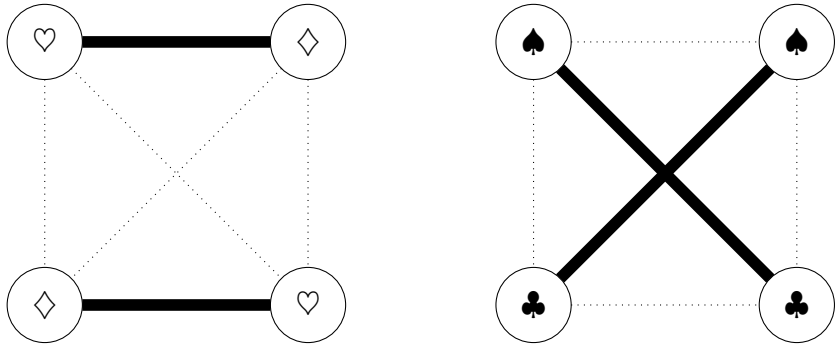
\begin{figure}[here]
	\begin{center}
		\begin{tikzpicture}[scale=0.5][line cap=round,line join=round,>=triangle 45,x=1.0cm,y=1.0cm]
		\clip(11.,2.) rectangle (35.,13.);
		\draw(13.,11.) circle (1.cm)node[anchor=center] {$\heartsuit$};
		\draw(20.,11.) circle (1.cm)node[anchor=center] {$\diamondsuit$};
		\draw(13.,4.) circle (1.cm)node[anchor=center] {$\diamondsuit$};
		\draw(20.,4.) circle (1.cm)node[anchor=center] {$\heartsuit$};
		\draw [line width=4.4pt] (14.,11.)-- (19.,11.);
		\draw [dotted] (20.,10.)-- (20.,5.);
		\draw [dotted] (13.,10.)-- (13.,5.);
		\draw [line width=4.4pt] (14.,4.)-- (19.,4.);
		\draw [dotted] (19.292893218813454,10.292893218813452)-- (13.707106781186548,4.707106781186548);
		\draw [dotted] (13.707106781186548,10.292893218813452)-- (19.292893218813454,4.707106781186548);
		\draw (12.73399896807,11.611143997505334) ;
		\draw (19.78068891530594,4.591873077301424) ;
		\draw (19.78068891530594,11.55630594344124) ;
		\draw (12.788837022134093,4.591873077301424) ;
		\draw(25.996618813190036,11.) circle (1.cm)node[anchor=center] {$\spadesuit$};
		\draw(32.99661881319001,11.) circle (1.cm)node[anchor=center] {$\spadesuit$};
		\draw(25.996618813190036,4.) circle (1.cm)node[anchor=center] {$\clubsuit$};
		\draw(32.99661881319001,4.) circle (1.cm)node[anchor=center] {$\clubsuit$};
		\draw [dotted] (26.996618813190036,11.)-- (31.99661881319001,11.);
		\draw [dotted] (32.99661881319001,10.)-- (32.99661881319001,5.);
		\draw [dotted] (25.996618813190036,10.)-- (25.996618813190036,5.);
		\draw [dotted] (26.996618813190036,4.)-- (31.99661881319001,4.);
		\draw [line width=4.4pt] (32.289512032003465,10.292893218813452)-- (26.703725594376582,4.707106781186548);
		\draw [line width=4.4pt] (26.703725594376582,10.292893218813452)-- (32.289512032003465,4.707106781186548);
		\draw (25.730617781260026,11.611143997505334) ;
		\draw (32.77730772849597,4.591873077301424) ;
		\draw (32.77730772849597,11.55630594344124) ;
		\draw (25.78545583532412,4.591873077301424) ;
		\end{tikzpicture}
		\caption{Two matchings in the complete (2-uniform hyper)graph which satisfy the pairing condition. Each edge in one is monochromatic in the other, and vice versa.}\label{figure:twoMatchings}
	\end{center}
\end{figure}

Subspaces which correspond to distinct matchings $\cX=(X^{(0)},\ldots,X^{(r-1)})$ and $\cY=(Y^{(0)},\ldots,Y^{(r-1)})$, that satisfy the pairing condition, have a useful property. This property is a corollary of the following lemma.

\begin{lemma}\label{lemma:lOver9intersection}
	If $D\in\{X^{(0)},\ldots,X^{(r-1)}\}$ and $E\in\{Y^{(0)},\ldots,Y^{(r-1)}\}$ then $|D\cap E| = \ell/r^2$.
\end{lemma}
\begin{proof}
	Since $\cX$ and $\cY$ satisfy the pairing condition, $D$ can be written as a union of edges from $\cY$, that is, $D=\cup_{i=0}^{\ell/r^2-1}\{y_i^{(0)},y_i^{(1)},\ldots,y_i^{(r-1)}\}$. Hence, $D$ contains exactly $\ell/r^2$ elements of $E$.
\end{proof}

By~\cite[Lemma~11]{AccessTWB}, in any MSR code which attains minimum repair bandwidth, the dimension of the intersection between any two repair subspaces is at most $\ell/r^2$. As a result of Lemma~\ref{lemma:lOver9intersection}, we have that repair subspaces which correspond to different matchings attain this bound with equality.

\begin{corollary}\label{corollary:eigenspaceIntesections}
	If $U\in\{S_{X^{(0)}},S_{X^{(1)}},\ldots,S_{X^{(r-1)}},S_{X^*}\}$ and $V\in\{S_{Y^{(0)}},S_{Y^{(1)}},\ldots,S_{Y^{(r-1)}},S_{Y^*}\}$, then $\dim(U\cap V)=\ell/r^2$.
\end{corollary}
\iffull
\begin{proof}
		We distinguish between the following cases.
		\begin{description}
			\item[Case 1.] $U\in \{S_{X^{(0)}},S_{X^{(1)}},\ldots,S_{X^{(r-1)}}\}$ and $V\in \{S_{Y^{(0)}},S_{Y^{(1)}},\ldots,S_{Y^{(r-1)}}\}$. Let $B_U$ and $B_V$ be bases of $U$ and $V$, respectively, which consists of unit vectors only. Any $w\in U\cap V$ corresponds to a unique solution to the following equation, whose variables are $\{\alpha_i\}_{i=0}^{\ell/r-1}$ and $\{\beta_i\}_{i=0}^{\ell/r-1}$.
			\begin{eqnarray*}
				\sum_{u_i\in B_U}\alpha_i u_i = \sum_{v_i\in B_V}\beta_i v_i.
			\end{eqnarray*}
			Since all vectors involved in this equation are unit vectors, and since by Lemma~\ref{lemma:lOver9intersection} we have that $|B_U\cap B_V|=\ell/r^2$, it follows that \textit{exactly} $\ell/r^2$ of the coefficients in the left-hand side are equal to \textit{exactly} $\ell/r^2$ in the right-hand side, and the rest of the coefficients are zero. Therefore, this equation has \textit{exactly} $\ell/r^2$ degrees of freedom, and thus $\dim(U\cap V)=\ell/r^2$.
			\item[Case 2.] $U\in \{S_{X^{(0)}},S_{X^{(1)}},\ldots,S_{X^{(r-1)}}\}$ and $V=S_{Y^*}$. Let $B_U$ be a basis of $U$ which consists of unit vectors only. As in the previous case, any $w\in U\cap V$ corresponds to a solution of
			\begin{eqnarray*}
				\sum_{u_i\in B_U}\alpha_i u_i=\sum_{i=0}^{\ell/r-1}\beta_i(y_i^{(0)}+y_i^{(1)}+\ldots+y_i^{(r-1)}).
			\end{eqnarray*}
			By Lemma~\ref{lemma:lOver9intersection}, exactly $\ell/r^2$ of the edges $\{y_i^{(0)},y_i^{(1)},\ldots,y_i^{(r-1)}\}$ are in $B_U$, and hence, $(r-1)\ell/r^2$ of the coefficients in the right-hand side must be zero. The remaining $\ell/r^2$ coefficients may be chosen arbitrarily. Thus there are $\ell/r^2$ degrees of freedom for this equation, and hence $\dim(U\cap V)=\ell/r^2$.
			\item[Case 3.] $U\in \{S_{Y^{(0)}},S_{Y^{(1)}},\ldots,S_{Y^{(r-1)}}\}$ and $V=S_{X^*}$. This case is symmetric to Case 2.
			\item[Case 4.] $U=S_{X^*}$, $V=S_{Y^*}$. Any vector $w\in U\cap V$ corresponds to a solution to
			\begin{eqnarray*}
				\sum_{i=0}^{\ell/r-1}\alpha_i(x_i^{(0)}+x_i^{(1)}+\ldots+x_i^{(r-1)})=\sum_{i=0}^{\ell/r-1}\beta_i(y_i^{(0)}+y_i^{(1)}+\ldots+y_i^{(r-1)}).
			\end{eqnarray*}
			Any edge $\{x_i^{(0)},x_i^{(1)},\ldots,x_i^{(r-1)}\}$ is contained in either of $Y^{(0)},Y^{(1)},\ldots,Y^{(r-1)}$. Hence, for each value chosen for $\alpha_i$, exactly $r$ distinct values for the different $\beta$-s immediately follow. Hence, this equation also has $\ell/r^2$ degrees of freedom, and $\dim(U\cap V)=\ell/r^2$.
		\end{description}
	\end{proof}
	
	\begin{remark}\label{remark:missingCase}
		In the proof of Corollary~\ref{corollary:eigenspaceIntesections}, the subspace $S_{X^*}$ could equally be replaced with any subspace of the form 
		\[\Span{\left\{c_0 x_i^{(0)}+c_1 x_i^{(1)}+\cdots+c_{r-1}x_i^{(r-1)}\right\}_{i=0}^{\ell/r-1}}\]
		for any constants $c_i$ such that $c_i\ne 0$ for all~$i$, and the subspace $S_{Y^*}$ could equally be replaced with any subspace of the form 
		\[\Span{\left\{d_0 y_i^{(0)}+d_1 y_i^{(1)}+\cdots+d_{r-1}y_i^{(r-1)}\right\}_{i=0}^{\ell/r-1}}\]
		for any constants $d_i$ such that $d_i\ne 0$ for all~$i$. This is since cases 2-4 of the proof do not use the specific choice of $c_i=1$ and $d_i=1$. This fact will be used in the proof of Lemma~\ref{lemma:simDig} to follow.
	\end{remark}
	\fi
	In the sequel we use a large set of matchings in which every two matchings satisfy the pairing condition. To satisfy the nonsingular property (Definition~\ref{definition:theSubspaceCondition}), each matrix of the $\AS$-set is multiplied by a properly chosen field constant, without compromising the invariance property and the independence property. The constructions of the $\AS$-sets, which follow the general outline described in this subsection, are discussed in detail in the following sections.
\section{Construction of an MSR Code with Two Parities} \label{section:twoParities}

\subsection{Two Parities Code from One Matching}
Recall that the vertices of the complete graph $K_\ell$ are identified by all unit vectors $e_0,\ldots,e_{\ell-1}$ of length~$\ell$, $\ell=2^m$ for some integer $m$, and a matching $\cZ=(Z,Z')$ is a set of $\ell/2$ vertex-disjoint edges of $K_\ell$.
Such a matching will provide an $\AS$-set of size 2, satisfying the subspace condition. The construction of this $\AS$-set also relies on the following $\ell\times\ell$ matrices, which resemble the matrix in~\eqref{equation:A}. For $\lambda \in \bF_q^*$, consider the following two $\ell/2\times\ell/2$ matrices
\begin{eqnarray*}
A^+(\lambda)\triangleq
\begin{pmatrix}
0      & \lambda& 0 	 & 0 	  & \cdots & 0 		& 0\\
\lambda& 0      & 0 	 & 0 	  & \cdots & 0 		& 0\\
0      & 0      & 0 	 & \lambda& \cdots & 0 		& 0\\
0      & 0      & \lambda& 0 	  & \cdots & 0 		& 0\\
\vdots & \vdots & \vdots & \vdots & \ddots & \vdots & \vdots\\
0      & 0      & 0 	 & 0 	  & \cdots & 0 		& \lambda\\
0      & 0      & 0 	 & 0 	  & \cdots & \lambda& 0\\
\end{pmatrix},
A^-(\lambda)\triangleq -A^+(\lambda),
\end{eqnarray*}
and let $A(\lambda)$ be the following $\ell\times\ell$ block diagonal matrix
\begin{eqnarray}\label{equation:A_lambda}
A(\lambda)\triangleq
\begin{pmatrix}
A^+(\lambda)& 0 \\
0	 	  & A^-(\lambda)\\
\end{pmatrix}.
\end{eqnarray}
The matrix $A(\lambda)$ possesses several useful properties, which are essential in our construction. These useful properties follow from the fact that the minimal polynomial of $A(\lambda)$ is $x^2-\lambda^2$. This form of the minimal polynomial shows that the matrix $A(\lambda)$ acts as a \textit{transposition} on the vectors of $\bF_{q}^\ell$ which are not eigenvectors, up to a multiplication by $\lambda$. That is, all vectors which are not eigenvectors may be partitioned to pairs $(u,v)$ such that $uA(\lambda)=v$ and $vA(\lambda)=\lambda^2 u$, as proved in Lemma~\ref{lemma:PAP} which follows. In addition, for field with even characteristic, the matrix $A(\lambda)$ is non-diagonalizable. To the best of our knowledge, this constitutes the first construction of an $\AS$-set satisfying the subspace condition whose matrices are non-diagonalizable. Notice that the multiplication of a vector $v$ by the matrix $A(\lambda)$ switches between entries $2t$ and $2t+1$ of $v$ for all $t\in\{0,\ldots,\ell/2-1\}$, and multiplies all entries by either $\lambda$ or $-\lambda$ according to $t\le \ell/4-1$ or $t>\ell/4-1$. This is demonstrated in the following lemma.

\begin{lemma}\label{lemma:PAP}
If $P\in\bF_q^{\ell\times\ell}$ is an invertible matrix whose rows are $p_0,\ldots,p_{\ell-1}$, and $B\triangleq P^{-1}A(\lambda) P$ for some $\lambda\in\bF_q^*$, then for all $t\in\{0,\ldots,\ell/2-1\}$
\begin{eqnarray*}
p_{2t}B=
\begin{cases}
\phantom{-}\lambda p_{2t+1} &\mbox{if } t\le \ell/4-1\\
-\lambda p_{2t+1} &\mbox{if } t> \ell/4-1
\end{cases},~~~
p_{2t+1}B=
\begin{cases}
\phantom{-}\lambda p_{2t} &\mbox{if } t\le \ell/4-1\\
-\lambda p_{2t} &\mbox{if } t> \ell/4-1
\end{cases}.
\end{eqnarray*}
Furthermore, the vectors $p_{2t+1}+p_{2t}$ and $p_{2t+1}-p_{2t}$ are eigenvectors of $B$.
\end{lemma}

\begin{proof}
By (\ref{equation:A_lambda}), for all $t\in\{0,\ldots,\ell/2-1\}$ we have that
\begin{eqnarray*}
e_{2t}A(\lambda)=
\begin{cases}
\phantom{-}\lambda e_{2t+1} &\mbox{if } t\le \ell/4-1\\
-\lambda e_{2t+1} &\mbox{if } t> \ell/4-1
\end{cases},~~~
e_{2t+1}A(\lambda)=
\begin{cases}
\phantom{-}\lambda e_{2t} &\mbox{if } t\le \ell/4-1\\
-\lambda e_{2t} &\mbox{if } t> \ell/4-1
\end{cases}.
\end{eqnarray*}
In addition, since $PP^{-1}=I$, it follows that $p_iP^{-1}=e_i$ for all $i\in\{0,\ldots,\ell-1\}$. Therefore, for all $t\in\{0,\ldots,\ell/2-1\}$
\begin{eqnarray}
\nonumber p_{2t}B&=&p_{2t}P^{-1}A(\lambda)P\\
\nonumber &=&e_{2t}A(\lambda)P\\
&=&\pm\lambda e_{2t+1}P=\pm \lambda p_{2t+1}\label{equation:Beven},\\
\nonumber p_{2t+1}B&=&p_{2t+1}P^{-1}A(\lambda)P\\
\nonumber &=&e_{2t+1}A(\lambda)P\\
&=&\pm\lambda e_{2t}P=\pm \lambda p_{2t}\label{equation:Bodd},
\end{eqnarray}
where the $\pm$ sign distinguishes between the cases $t\le \ell/4-1$ and $t>\ell/4-1$. To see that $p_{2t+1}+p_{2t}$ and $p_{2t+1}-p_{2t}$ are eigenvectors of $B$, notice that by adding and substracting (\ref{equation:Beven}) and (\ref{equation:Bodd}), we have that
\begin{eqnarray}
(p_{2t+1}+p_{2t})B&=&\pm\lambda(p_{2t+1}+p_{2t})\label{equation:sumIsEigenVector}\\
(p_{2t+1}-p_{2t})B&=&\mp\lambda(p_{2t+1}-p_{2t})\label{equation:diffIsEigenVector}.
\end{eqnarray}
\end{proof}

Given a matching $\cZ=(Z,Z')$, it is easily verified that the following two matrices are invertible, where $z_i,z_i'$ are vertices in the complete graph, which are identified by unit vectors of length~$\ell$.
\begin{eqnarray*}
P_Z\triangleq\begin{pmatrix}
z_0\\z_0'-z_0\\z_1\\z_1'-z_1\\ \vdots \\z_{\ell/2-1}\\z_{\ell/2-1}'-z_{\ell/2-1}
\end{pmatrix},
P_{Z'}\triangleq\begin{pmatrix}
z_0'\\z_0+z_0'\\z_1'\\z_1+z_1'\\ \vdots \\z_{\ell/2-1}'\\z_{\ell/2-1}+z_{\ell/2-1}'
\end{pmatrix}
\end{eqnarray*}

\begin{definition}\label{definition:oneMatchingCode}
Given a matching $\cZ=(Z,Z')$, let
\begin{align}
\label{equation:AZ}  A_Z   (\lambda)&\triangleq   P_Z   ^{-1}\cdot A(\lambda) \cdot P_Z   ,& S_Z   &\triangleq \Span{Z}=\Span{\{z_i\}_{i=0}^{\ell/2-1}}\\
\label{equation:AZ'} A_{Z'}(\lambda)&\triangleq   P_{Z'}^{-1}\cdot A(\lambda) \cdot P_{Z'},& S_{Z'}&\triangleq \Span{Z'}=\Span{\{z_i'\}_{i=0}^{\ell/2-1}}.
\end{align}
\end{definition}

As an immediate consequence of Lemma~\ref{lemma:PAP} and Definition~\ref{definition:oneMatchingCode}, we have the following.

\begin{corollary}\label{corollary:fromLemmaPAP}
For every $i\in\{0,\ldots,\ell/4-1\}$,
\begin{eqnarray*}
z_iA_Z(\lambda)=
\begin{cases}
\phantom{-}\lambda (z_i'-z_i) &\mbox{if } i\le \ell/4-1\\
-\lambda (z_i'-z_i) &\mbox{if } i> \ell/4-1
\end{cases},~~~
z_i'A_{Z'}(\lambda)=
\begin{cases}
\phantom{-}\lambda (z_i+z_i') &\mbox{if } i\le \ell/4-1\\
-\lambda (z_i+z_i') &\mbox{if } i> \ell/4-1
\end{cases},
\end{eqnarray*}
and,
\begin{itemize}
\item For $i\le \ell/4-1$,
\begin{itemize}
\item $z_i'$ is an eigenvector of $A_Z(\lambda)$ which corresponds to the eigenvalue $\lambda$.
\item $z_i$ is an eigenvector of $A_{Z'}(\lambda)$ which corresponds to the eigenvalue $-\lambda$.
\end{itemize}
\item For $i>\ell/4-1$,
\begin{itemize}
\item $z_i'$ is an eigenvector of $A_Z(\lambda)$ which corresponds to the eigenvalue $-\lambda$.
\item $z_i$ is an eigenvector of $A_{Z'}(\lambda)$ which corresponds to the eigenvalue $\lambda$.
\end{itemize}
\end{itemize}
\end{corollary}

A matching $\cZ$ provides an $\AS$-set of size two as follows.

\begin{lemma}\label{lemma:onePair}
If $\cZ=(Z,Z')$ is a matching, then $\{(A_Z(\lambda),S_Z),(A_{Z'}(\lambda),S_{Z'})\}$ satisfies the subspace condition.
\end{lemma}

\begin{proof}
For convenience of notation, and since the proof which follows holds for each $\lambda\ne 0$, let $A_Z$ and $A_{Z'}$ denote $A_Z(\lambda)$ and $A_{Z'}(\lambda)$, respectively. We show that all properties of the subspace condition are satisfied.

To prove the independence property, notice that by Corollary~\ref{corollary:fromLemmaPAP},
\begin{eqnarray*}
S_ZA_Z&=&\Span{\{z_i'-z_i\}_{i=0}^{\ell/2-1}},\\
S_{Z'}A_{Z'}&=&\Span{\{z_i+z_i'\}_{i=0}^{\ell/2-1}},
\end{eqnarray*}
and thus, $S_ZA_Z+S_Z = S_{Z'}A_{Z'}+S_{Z'} = \bF_q^\ell$.

To prove the invariance property, notice that by Corollary~\ref{corollary:fromLemmaPAP}, $S_Z$ (resp. $S_{Z'}$) is a span of eigenvectors\footnote{Note that it does not comply with the definition of an eigenspace, since it contains vectors that correspond to distinct eigenvalues.} of $A_{Z'}$ (resp. $A_{Z}$) and hence it is $A_{Z'}$ (resp. $A_{Z}$) invariant.


To prove the nonsingular property, first notice that $A_Z,A_{Z'}$ are invertible since they are defined as a product of invertible matrices, and thus every $1\times 1$ block submatrix is invertible. Second, notice that\[\begin{pmatrix} I & I \\ A_Z & A_{Z'} \end{pmatrix}\]is invertible if and only if $A_Z-A_{Z'}$ is invertible. Since $Z\cup Z'$ is a basis of $\bF_q^\ell$, to show that $A_Z-A_{Z'}$ is invertible it suffices to show that its image contains $Z\cup Z'$.

Let $i\in\{0,\ldots,\ell/2-1\}$, and notice that by Corollary~\ref{corollary:fromLemmaPAP}, if $ i\le\ell/4-1$ then
\begin{eqnarray*}
\lambda ^{-1}z_i(A_Z-A_{Z'})&=&\phantom{-}\lambda^{-1}\left(z_iA_Z-z_iA_{Z'}\right)\\
&=&\phantom{-}\lambda^{-1}\left(\lambda(z_i'-z_i)+\lambda z_i\right)=z_i'\\
-\lambda^{-1} z_i'(A_Z-A_{Z'})&=&-\lambda^{-1}\left( z_i'A_Z-z_i'A_{Z'}\right)\\
&=&-\lambda^{-1}\left( \lambda z_i'-\lambda(z_i+z_i')\right)=z_i.\\
\end{eqnarray*}
On the other hand, if $i>\ell/4+1$, then
\begin{eqnarray*}
-\lambda^{-1} z_i(A_Z-A_{Z'})&=&-\lambda^{-1}\left( z_iA_{Z}-z_iA_{Z'}\right)\\
&=&-\lambda^{-1}\left(-\lambda(z_i'-z_i)-\lambda z_i\right)=z_i'\\
\lambda^{-1}z_i'(A_Z-A_{Z'})&=&\phantom{-}\lambda^{-1}\left( z_i'A_Z-z_i'A_{Z'}\right)\\
&=&\phantom{-}\lambda^{-1}\left(-\lambda m_i'+\lambda(z_i'+z_i)\right)=z_i.
\end{eqnarray*}
Therefore, for all $i\in\{0,\ldots,\ell/2-1\}$, the vectors $z_i$ and $z_i'$ are in the image $\image(A_Z-A_{Z'})$, which implies that $A_Z-A_{Z'}$ is of full rank.
\end{proof}

From Lemma~\ref{lemma:onePair} it is evident that any pair $(\cZ,\lambda)$ of a matching $\cZ=(Z,Z')$ and a nonzero field element $\lambda$ provides an $\AS$-set of size two. In Section~\ref{section:twoMatchings2} which follows we discuss the required relation between \textit{two} such pairs $(\cX,\lambda_x)$, $(\cY,\lambda_y)$ that allow the corresponding $\AS$-sets to be united without compromising the subspace condition.

\subsection{Two Parities Code from Two Matchings}\label{section:twoMatchings2}
To construct larger $\AS$-sets, we analyse the required relations between two \textit{distinct} pairs $(\cX,\lambda_x)$, $(\cY,\lambda_y)$ of matchings $\cX=(X,X'),~\cY=(Y,Y')$ and field elements $\lambda_x,~\lambda_y$, that allow the construction of an $\AS$-set of size four. In Lemma~\ref{lemma:twoMatchings}, which follows, we show that there exist three sufficient conditions that $(\cX,\lambda_x),(\cY,\lambda_y)$ should satisfy for this purpose. The first condition states that $\lambda_x$ and $\lambda_y$ must be distinct. The second condition, called the \textit{pairing condition}, appears in Definition~\ref{definition:pairingCondition}. The third condition, which is a more subtle one and will only be relevant in fields with odd characteristic, is that the vertices of certain edges from $\cX$ fall into distinct halves defined by the order of $\cY$, and vice versa.

Clearly, a set $\{(\cX_i,\lambda_i)\}_{i=1}^t$ such that any two pairs satisfy all of the above conditions, will provide an $\AS$-set of size $2t$. In the sequel we provide such a set of size $m$ over $\bF_q$, for any $m\in\bN$ and any $q\ge m+1$. This set will yield an $\AS$-set of size $2m$ for $q\ge m+1$, which consists of matrices of size~$2^m\times 2^m$.
%
%
%
%

\begin{lemma}\label{lemma:twoMatchings}
If $\cX=(X,X'),~\cY=(Y,Y')$ are matchings and $\lambda_x,\lambda_y$ are nonzero field elements such that
\begin{itemize}
\item[\emph{B1.}] $\lambda_x\ne\lambda y$.
\item[\emph{B2.}] $\cX$ and $\cY$ satisfy the pairing condition (see Definition~\ref{definition:pairingCondition}).
\item[\emph{B3.}] If $\lambda_x=-\lambda_y$, then for all $i\in\{0,\ldots,\ell/2-1\}$,
\begin{eqnarray*}
\mbox{if }(x_i,x_i')=(y_j,y_t)   &\mbox{ then } & i\le \ell/4-1,~ j\le \ell/4-1,\mbox{ and } t>\ell/4-1, \mbox{ and}\\
\mbox{if }(x_i,x_i')=(y_j',y_t') &\mbox{ then } & i> \ell/4-1,~j\le \ell/4-1, \mbox{ and } t>\ell/4-1.
\end{eqnarray*}
\end{itemize}
then the $\AS$-set
\begin{eqnarray*}
\left\{(A_X(\lambda_x),S_X),(A_{X'}(\lambda_x),S_{X'}),(A_Y(\lambda_y),S_Y),(A_{Y'}(\lambda_y),S_{Y'})\right\}
\end{eqnarray*}
satisfies the subspace condition.
\end{lemma}
\begin{proof}
For convenience, we omit the notations of $\lambda_x,\lambda_y$ from $A_X(\lambda_x),A_{X'}(\lambda_x),A_Y(\lambda_y)$, and $A_{Y'}(\lambda_y)$ (even so $\lambda_x$ and $\lambda_y$ are crucial for this proof). The independence property follows directly from Lemma~\ref{lemma:onePair}, as well as the non-singularity of any $1\times 1$ submatrix in the nonsingular property. To prove the invariance property, notice that the cases
\begin{align*}
S_XA_{X'}&=S_X      &S_YA_{Y'}&=S_Y \\
S_{X'}A_{X}&=S_{X'} &S_{Y'}A_{Y}&=S_{Y'}
\end{align*}
follow from Lemma~\ref{lemma:onePair} as well. We prove now that $S_XA_Y=S_X$, and the rest of the cases follow by symmetry.

Since $S_X=\Span{x_0,\ldots,x_{\ell/2-1}}$, a necessary and sufficient condition for $S_XA_Y=S_X$ is that  $x_iA_Y\in S_X$ for each $i\in\{0,\ldots,\ell/2-1\}$. Let $x_i\in S_X$ for some $i\in\{0,\ldots,\ell/2-1\}$. Since $\cX$ and $\cY$ are matchings over the same vertex set, we have that either $x_i\in Y$ or $x_i\in Y'$. If $x_i\in Y'$, i.e. $x_i=y_j'$ for some $j\in\{0,\ldots,\ell/2-1\}$, then by Corollary~\ref{corollary:fromLemmaPAP} and by the definition of $A_Y$ (\ref{equation:AZ}), we have that $y_j'$ is an eigenvector of $A_Y$. Therefore,
\begin{eqnarray*}
x_iA_Y=y_j'A_Y=\pm \lambda_y y_j'=\pm\lambda_y x_i\in S_X.
\end{eqnarray*}
On the other hand, if $x_i\in Y$, i.e. $x_i=y_j$ for some $j\in\{0,\ldots,\ell/2-1\}$, then by Corollary~\ref{corollary:fromLemmaPAP},
\begin{eqnarray}\label{equation:sumInSX}
x_iA_Y=y_jA_Y=\pm\lambda_y(y_j'-y_j)=\pm\lambda_yy_j'\mp\lambda_y y_j=\pm\lambda_yy_j'\mp\lambda_y x_i.
\end{eqnarray}
According to B2 (the pairing condition), we have that if $y_j\in X$, then $y_j'\in X$ as well. Therefore~(\ref{equation:sumInSX}) is a sum of two vectors in $S_X$, which implies that $x_iA_Y\in S_X$.

To prove the nonsingular property, we show that $X\cup X'\subseteq \image\left(A_X-A_Y\right)$, and the rest of the cases follow by symmetry. Since $X\cup X'$ is a basis of $\bF_q^\ell$, it will follow that $\rank(A_X-A_Y)=\ell$ as required. We split the proof to two cases as follows.


\begin{description}
\item[Case 1.]$\lambda_x\ne-\lambda_y$ (and thus $\lambda_x\ne\pm \lambda_y$ by B1). If $i\in\{0,\ldots,\ell/2-1\}$, then by A2, we have that either $(x_i,x_i')=(y_j,y_t)$ or $(x_i,x_i')=(y_j',y_t')$ for some distinct $j,t\in\{0,\ldots,\ell/2-1\}$. If $(x_i,x_i')=(y_j',y_t')$, then simple calculations that follow from Corollary~\ref{corollary:fromLemmaPAP} show that
\begin{eqnarray}
\label{equation:xIn}
x_i(A_X-A_Y)&=&
\begin{cases}
\phantom{-}\lambda_x x_i'-(\lambda_x+\lambda_y)x_i & \mbox{if } i\le \ell/4-1,j\le \ell/4-1\\
\phantom{-}\lambda_x x_i'-(\lambda_x-\lambda_y)x_i & \mbox{if } i\le \ell/4-1,j> \ell/4-1\\
-\lambda_x x_i'+(\lambda_x-\lambda_y)x_i & \mbox{if } i> \ell/4-1,j\le \ell/4-1\\
-\lambda_x x_i'+(\lambda_x+\lambda_y)x_i & \mbox{if } i> \ell/4-1,j> \ell/4-1
\end{cases}\\
\label{equation:xpIn}
x_i'(A_X-A_Y)&=&
\begin{cases}
\phantom{-}(\lambda_x -\lambda_y) x_i' & \mbox{if } i\le \ell/4-1,t\le \ell/4-1\\
\phantom{-}(\lambda_x +\lambda_y) x_i' & \mbox{if } i\le \ell/4-1,t> \ell/4-1\\
-(\lambda_x +\lambda_y) x_i' & \mbox{if } i> \ell/4-1,t\le \ell/4-1\\
-(\lambda_x -\lambda_y) x_i' & \mbox{if } i> \ell/4-1,t> \ell/4-1.
\end{cases}
\end{eqnarray}
Since $\lambda_x\ne \pm \lambda_y$, it follows by (\ref{equation:xpIn}) that $x_i'\in\image(A_X-A_Y)$, which also implies by~(\ref{equation:xIn}) that $x_i\in\image(A_X-A_Y)$.
If $(x_i,x_i')=(y_j,y_t)$, then similar calculations show that
\begin{eqnarray}
\label{equation:xIn2}
x_i(A_X-A_Y)&=&
\begin{cases}
\phantom{-}\lambda_x x_i'-(\lambda_x-\lambda_y)x_i-\lambda_y y_j' & \mbox{if } i\le \ell/4-1,j\le \ell/4-1\\
\phantom{-}\lambda_x x_i'-(\lambda_x+\lambda_y)x_i+\lambda_yy_j' & \mbox{if } i\le \ell/4-1,j> \ell/4-1\\
-\lambda_x x_i'+(\lambda_x+\lambda_y)x_i-\lambda_y y_j' & \mbox{if } i> \ell/4-1,j\le \ell/4-1\\
-\lambda_x x_i'+(\lambda_x-\lambda_y)x_i+\lambda_y y_j' & \mbox{if } i> \ell/4-1,j> \ell/4-1
\end{cases}\\
\label{equation:xpIn2}
x_i'(A_X-A_Y)&=&
\begin{cases}
\phantom{-}(\lambda_x +\lambda_y) x_i'-\lambda_y y_t' & \mbox{if } i\le \ell/4-1,t\le \ell/4-1\\
\phantom{-}(\lambda_x -\lambda_y) x_i'+\lambda_y y_t'& \mbox{if } i\le \ell/4-1,t> \ell/4-1\\
-(\lambda_x-\lambda_y) x_i'-\lambda_y y_t' & \mbox{if } i> \ell/4-1,t\le \ell/4-1\\
-(\lambda_x+\lambda_y) x_i'+\lambda_y y_t' & \mbox{if } i> \ell/4-1,t> \ell/4-1
\end{cases}.
\end{eqnarray}
Now, notice that since $x_i=y_j$ we have that $y_j\in X$. By B2,  we also have that $y_j'\in X$, and hence $y_j'=x_s$ for some $s\in\{0,\ldots,\ell/2-1\}$. We have shown earlier that if $x_s=y_j'$ then $x_s=y_j'\in\image (A_X-A_Y)$. Similarly, since $x_i'=y_t$, we have that $y_t'\in X'$, i.e. $y_t'=x_r'$ for some $r\in\{0,\ldots,\ell/2-1\}$. This implies that $x_r,x_r'\in Y'$ by the pairing condition, and thus, $x_r'=y_t'\in \image (A_X-A_Y)$.

Since $y_t'\in \image (A_X-A_Y)$, and since $\lambda_x\ne\pm\lambda_y$, it follows from~(\ref{equation:xpIn2}) that $x_i'\in \image (A_X-A_Y)$. Therefore, by~(\ref{equation:xIn2}), and since $y_j'\in \image (A_X-A_Y)$ and $\lambda_x\ne\pm\lambda_y$, it follows that $x_i\in \image (A_X-A_Y)$ as well.
\item[Case 2.] $\lambda_x=-\lambda_y$ (and thus we have to consider B3). 
The pairing condition implies that either $(x_i,x_i')=(y_j,y_t)$ or $(x_i,x_i')=(y_j',y_t')$ for some distinct $j,t\in\{0,\ldots,\ell/2-1\}$. However, by B3, most of the cases in~(\ref{equation:xIn}),(\ref{equation:xpIn}),(\ref{equation:xIn2}),(\ref{equation:xpIn2}) are impossible. Hence, if $(x_i,x_i')=(y_j',y_t')$ then  $i> \ell/4-1,j\le \ell/4-1$, and $t>\ell/4-1$, and thus
\begin{eqnarray*}
x_i(A_X-A_Y)&=&
-\lambda_x x_i'+(\lambda_x-\lambda_y)x_i \\
x_i'(A_X-A_Y)&=&(\lambda_y -\lambda_x) x_i'.
\end{eqnarray*}
Since $\lambda_x\ne \lambda_y$, we have that $x_i,x_i'\in\image(A_X-A_Y)$. If $(x_i,x_i')=(y_j,y_t)$ then $i\le \ell/4-1,j\le \ell/4-1$, and $t>\ell/4-1$, and thus
\begin{eqnarray*}
x_i(A_X-A_Y)&=&\phantom{-}\lambda_x x_i'+(\lambda_y-\lambda_x)x_i-\lambda_y y_j' \\
x_i'(A_X-A_Y)&=&
\phantom{-}(\lambda_x -\lambda_y) x_i'+\lambda_y y_t'.
\end{eqnarray*}
As in Case~1, we can prove that $y_j',y_t'\in\image(A_X-A_Y)$, and get that since $\lambda_x\ne\lambda_y$, we have that  $x_i,x_i'\in\image(A_X-A_Y)$.
\end{description}
\end{proof}
By Lemma~\ref{lemma:twoMatchings} we have that two matchings $\cX,\cY$ and two corresponding field elements $\lambda_x,\lambda_y$ that meet the requirements B1-B3, provide an $\AS$-set of size four. Therefore, a construction of a large set of pairs $(\cX_i,\lambda_i)$, such that any two pairs satisfy B1-B3, is required for a construction of a large $\AS$-set which satisfies the subspace condition.

\subsection{Construction of Matchings for Two Parities}

In the sequel we construct a set $\{(\cX_i,\lambda_i)\}_{i=0}^{m-1}$ whose elements satisfy the requirements of Lemma~\ref{lemma:twoMatchings} in pairs. We identify vertex $e_i$ of $K_\ell$ with the binary $m$-bit representation of $i$. We will use the following standard notion of a boolean cube.

\begin{definition}\label{definition:booleanCube}
Given a sequence of distinct indices $i_1,\ldots,i_k\in \{0,\ldots,m-1\}$ and a sequence of binary values $b_1,\ldots,b_k\in\{0,1\}$, the \emph{boolean cube $C(\{(i_j,b_j)\}_{j=1}^k)$}, of size $2^{m-k}$, is the set of all $m$-bit vectors over $\{0,1\}$ that have $b_j$ in entry $i_j$ for all $j=1,\ldots,k$. That is,
\begin{eqnarray*}
C(\{(i_j,b_j)\}_{j=1}^k)\triangleq \left\{x\in\{0,1\}^m~\vert~\mbox{for all } j\in[k],~x_{i_j}=b_j\right\}.
\end{eqnarray*}
We consider the elements in such a boolean cube as \emph{ordered} according to the lexicographic order (see Example~\ref{example:booleanCube} below), that is, we consider a boolean cube as a \emph{sequence} rather than a~\emph{set}.
\end{definition}

\begin{example}\label{example:booleanCube}
If $m=4$ then the boolean cube $C(\{(1,1),(2,1)\})$ is the set $\{v_1,v_2,v_3,v_4\}$ such that $(v_1,v_2,v_3,v_4)=(0110,0111,1110,1111)$.
\end{example}

We begin by defining a set of matchings that meets the pairing condition.

\begin{definition}\label{definition:matchings}
For any $m\in\bN$, define $m$ matchings $\{\cX_i=(X_i,X_i')\}_{i=0}^{m-1}$ as follows
\begin{eqnarray*}
\cX_{2t}:~~~
\begin{cases}
X_{2t} = C(\{(2t,0),(2t+1,0)\})\circ C(\{(2t,0),(2t+1,1)\})\\
X_{2t}' = C(\{(2t,1),(2t+1,0)\})\circ C(\{(2t,1),(2t+1,1)\}),
\end{cases}\\
\cX_{2t+1}:~~~
\begin{cases}
X_{2t+1} = C(\{(2t,0),(2t+1,0)\})\circ C(\{(2t,1),(2t+1,0)\})\\
X_{2t+1}' = C(\{(2t,0),(2t+1,1)\})\circ C(\{(2t,1),(2t+1,1)\})
\end{cases}
\end{eqnarray*}
where $t\in\{0,\ldots,\floor{\frac{m}{2}}-1\}$, and $\circ$ indicates the concatenation of sequences. If $m$ is odd, we add the matching
\begin{eqnarray*}
\cX_{m-1}:~~~
\begin{cases}
X_{m-1} = C(\{(m-1,0)\})\\
X_{m-1}' = C(\{(m-1,1)\})
\end{cases}.
\end{eqnarray*}
\end{definition}

\begin{example}
If $m=4$, then
\begin{eqnarray*}
\cX_{0}:~~~
\begin{cases}
X_{0} = C(\{(0,0),(1,0)\})\circ C(\{(0,0),(1,1)\})\\
X_{0}' = C(\{(0,1),(1,0)\})\circ C(\{(0,1),(1,1)\})
\end{cases}\\
\cX_{1}:~~~
\begin{cases}
X_{1} = C(\{(0,0),(1,0)\})\circ C(\{(0,1),(1,0)\})\\
X_{1}' = C(\{(0,0),(1,1)\})\circ C(\{(0,1),(1,1)\})
\end{cases}\\
\cX_{2}:~~~
\begin{cases}
X_{2} = C(\{(2,0),(3,0)\})\circ C(\{(2,0),(3,1)\})\\
X_{2}' = C(\{(2,1),(3,0)\})\circ C(\{(2,1),(3,1)\})
\end{cases}\\
\cX_{3}:~~~
\begin{cases}
X_{3} = C(\{(2,0),(3,0)\})\circ C(\{(2,1),(3,0)\})\\
X_{3}' = C(\{(2,0),(3,1)\})\circ C(\{(2,1),(3,1)\}),
\end{cases}
\end{eqnarray*}
which implies that
\begin{eqnarray*}
\cX_{0}:~~~
\begin{cases}
X_{0} = (\mathbf{00}00,\mathbf{00}01,\mathbf{00}10,\mathbf{00}11,\mathbf{01}00,\mathbf{01}01,\mathbf{01}10,\mathbf{01}11)\\
X_{0}' = (\mathbf{10}00,\mathbf{10}01,\mathbf{10}10,\mathbf{10}11,\mathbf{11}00,\mathbf{11}01,\mathbf{11}10,\mathbf{11}11)
\end{cases}\\
\cX_{1}:~~~
\begin{cases}
X_{1} = (\mathbf{00}00,\mathbf{00}01,\mathbf{00}10,\mathbf{00}11,\mathbf{10}00,\mathbf{10}01,\mathbf{10}10,\mathbf{10}11)\\
X_{1}' = (\mathbf{01}00,\mathbf{01}01,\mathbf{01}10,\mathbf{01}11,\mathbf{11}00,\mathbf{11}01,\mathbf{11}10,\mathbf{11}11)
\end{cases}\\
\cX_{2}:~~~
\begin{cases}
X_{2} = (00\mathbf{00},01\mathbf{00},10\mathbf{00},11\mathbf{00},00\mathbf{01},01\mathbf{01},10\mathbf{01},11\mathbf{01})\\
X_{2}' = (00\mathbf{10},01\mathbf{10},10\mathbf{10},11\mathbf{10},00\mathbf{11},01\mathbf{11},10\mathbf{11},11\mathbf{11})
\end{cases}\\
\cX_{3}:~~~
\begin{cases}
X_{3} = (00\mathbf{00},01\mathbf{00},10\mathbf{00},11\mathbf{00},00\mathbf{10},01\mathbf{10},10\mathbf{10},11\mathbf{10}
)\\
X_{3}' = (00\mathbf{01},01\mathbf{01},10\mathbf{01},11\mathbf{01},00\mathbf{11},01\mathbf{11},10\mathbf{11},11\mathbf{11}),
\end{cases}
\end{eqnarray*}
where the values in bold indicate the fixed entries in each boolean cube.
\end{example}

Before we choose a field element for each matching, which satisfies B1-B3 of Lemma~\ref{lemma:twoMatchings}, we prove that the matchings from Definition~\ref{definition:matchings} satisfy the pairing condition.
%

\begin{lemma}\label{lemma:pairingCondition}
Each two distinct matchings $\cX_i,\cX_j$ from Definition~\ref{definition:matchings} satisfy the pairing condition.
\end{lemma}

\begin{proof}
Denote the elements of the matchings $\cX_i,\cX_j$ as
\begin{eqnarray*}
X_i&=&(x_{i,0},\ldots,x_{i,\ell/2-1})\\
X'_i&=&(x_{i,0}',\ldots,x_{i,\ell/2-1}')\\
X_j&=&(x_{j,0},\ldots,x_{j,\ell/2-1})\\
X_j'&=&(x_{j,0}',\ldots,x_{j,\ell/2-1}').
\end{eqnarray*}
By Definition~\ref{definition:matchings}, it is evident that in every edge $(x_{i,t},x_{i,t}')\in\cX_i$, the $i$-th entry of $x_{i,t}$ is 0, the $i$-th entry of $x_{i,t}'$ is 1, and the rest of the entries are identical. Similarly, in every edge $(x_{j,t},x_{j,t}')\in\cX_j$, the $j$-th entry of $x_{j,t}$ is 0, the $j$-th entry of $x_{j,t}'$ is 1, and the rest of the entries are identical. Therefore, for every edge $(x_{i,t},x_{i,t}')\in\cX_i$, if the $j$-th entry of both $x_{i,t}$ and $x_{i,t}'$ is 0, then $x_{i,t},x_{i,t}'\in X_j$, and if it is 1, then $x_{i,t},x_{i,t}'\in X_j'$. Therefore, $X_j$ is a union of edges from $\cX_i$. The proof that $X_i$ is a union of edges from $\cX_j$ is similar.
\end{proof}

We now turn to choose a proper nonzero field element for every matching from Definition~\ref{definition:matchings}. This choice must comply with requirements B1 and B3 of Lemma~\ref{lemma:twoMatchings}. Note that if $q$ is even, then B3 follows from B1 (vacuously). Hence, if the field characteristics is $2$, the choice of field elements is straightforward.

\begin{lemma}
If $q\ge m+1$ is a power of two, then by any arbitrary choice of pairwise distinct elements from $\bF_q^*$ for the $m$ matchings from Definition~\ref{definition:matchings}, the resulting $\AS$-set satisfies \emph{B1-B3} from Lemma~\ref{lemma:twoMatchings}.
\end{lemma}

\begin{proof}
Since the assigned elements are distinct, every two matchings satisfy property B1 of Lemma~\ref{lemma:twoMatchings}. According to Lemma~\ref{lemma:pairingCondition}, every two matchings satisfy the pairing condition (B2) as well. Since $q$ is even, property B3 is implied by property B1.
\end{proof}

If $q$ is odd, more care is needed for the mapping of nonzero field elements to the matchings. We do this by choosing field elements $\lambda$ and $-\lambda$ for two adjacent matchings $\cX_{2t},\cX_{2t+1}$.

\begin{lemma}\label{lemma:oddChar}
Let $q\ge m+1$ be a power of an odd prime. Assume we are given an arbitrary choice of pairwise distinct elements from $\bF_q^*$ to the $m$ matchings from Definition~\ref{definition:matchings}, such that for $\cX_{2t},\cX_{2t+1}$, some elements $\lambda,-\lambda$ are chosen for every $t\in\{0,\ldots,\floor{\frac{m}{2}}-1\}$. Then, the resulting $\AS$-set satisfies \emph{B1-B3} from Lemma~\ref{lemma:twoMatchings}.
\end{lemma}

\begin{proof}
For every two distinct matchings, requirement B1 of Lemma~\ref{lemma:twoMatchings} is trivially satisfied, and requirement B2 is satisfied by Lemma~\ref{lemma:pairingCondition}. To prove B3, let $\lambda_i=-\lambda_j$ be two field elements which are chosen for two matchings $\cX_i,\cX_j$. Without loss of generality (w.l.o.g.) assume that $i=2t$ and $j=2t+1$ for some $t\in\{0,\ldots,\floor{\frac{m}{2}}-1\}$.

Let $(x_{2t,s},x_{2t,s}')$ be an edge in $\cX_{2t}$, which implies that the $(2t)$-th bit of $x_{2t,s}$ is~0 and the $(2t)$-th bit of $x_{2t,s}'$ is~1. To prove B3, we must show that
\begin{eqnarray*}
\mbox{if }(x_{2t,s},x_{2t,s}')=(x_{2t+1,u},x_{2t+1,r})   &\mbox{ then } & s\le \ell/4-1,~ u\le \ell/4-1,\mbox{ and } r>\ell/4-1\mbox{, and}\\
\mbox{if }(x_{2t,s},x_{2t,s}')=(x_{2t+1,u}',x_{2t+1,r}') &\mbox{ then } & s> \ell/4-1,~ u\le \ell/4-1, \mbox{ and } r>\ell/4-1.
\end{eqnarray*}
If $(x_{2t,s},x_{2t,s}')=(x_{2t+1,u},x_{2t+1,r})$ for some $u,r\in\{0,\ldots,\ell/2-1\}$, it follows that the $(2t)$-th bit of $x_{2t+1,s}$ and $x_{2t+1,s}'$ is~0. Therefore
\begin{eqnarray*}
x_{2t,s}=x_{2t+1,u}&\in & C(\{(2t,0),(2t+1,0)\})\\
x_{2t,s}'=x_{2t+1,r}&\in & C(\{(2t,1),(2t+1,0)\}),
\end{eqnarray*}
and hence, by the definition of $\cX_{2t+1}$ (Definition~\ref{definition:matchings}), it follows that $u\le \ell/4-1$ and $r>\ell/4-1$. In addition, by the definition of $\cX_{2t}$ it follows that $s\le \ell/4-1$.

If $(x_{2t,s},x_{2t,s}')=(x_{2t+1,u}',x_{2t+1,r}')$ for some $u,r\in\{0,\ldots,\ell/2-1\}$, it follows that the $(2t)$-th bit of $x_{2t+1,s}$ and $x_{2t+1,s}'$ is~1. Therefore
\begin{eqnarray*}
x_{2t,s}=x_{2t+1,u}'&\in & C(\{(2t,0),(2t+1,1)\})\\
x_{2t,s}'=x_{2t+1,r}'&\in & C(\{(2t,1),(2t+1,1)\}),
\end{eqnarray*}
and hence, by the definition of $\cX_{2t+1}$, it follows that $u\le \ell/4-1$ and $r>\ell/4-1$. In addition, by the definition of $\cX_{2t}$ it follows that $s> \ell/4-1$.
\end{proof}

The main construction of this section is summarized in the following theorem.

\begin{theorem}\label{theorem:main}
If $m$ is a positive integer and $q\ge m+1$ is a prime power, then there exists an explicitly defined $\AS$-set $\bC$ of size $2m$ and $2^m\times 2^m$ matrices over $\bF_q$, which satisfies the subspace condition.
\end{theorem}

\begin{proof}
Let $\{\cX_i=(X_i,X_i')\}_{i=0}^{m-1}$ be the set of matchings from Definition~\ref{definition:matchings}, which  by Lemma~\ref{lemma:pairingCondition}, satisfies the pairing condition (Definition~\ref{definition:pairingCondition}). If $q$ is even, then let $\lambda_0,\ldots,\lambda_{m-1}$ be distinct elements in $\bF_q^*$, and let
\begin{eqnarray}\label{equation:code}
\bC\triangleq \bigcup_{i=0}^{m-1} \left\{ (A_{X_i}(\lambda_i),S_{X_i}),(A_{X_i'}(\lambda_i),S_{X_i'})\right\},
\end{eqnarray}
where $(A_{X_i}(\lambda_i),S_{X_i}),(A_{X_i'}(\lambda_i),S_{X_i'})$ were defined in Lemma~\ref{lemma:onePair}. Since conditions B1-B3 of Lemma~\ref{lemma:twoMatchings} are met with respect to every two matchings and their respective field elements, it follows that~$\bC$ satisfies the subspace condition.

If $q$ is odd, let $\lambda_0,\ldots,\lambda_{m-1}$ be distinct elements in $\bF_q^*$ such that $\lambda_{2t}=-\lambda_{2t+1}$ for every $t\in\{0,\ldots,\floor{\frac{m}{2}}-1\}$. Define~$\bC$ in a similar way to the one defined in (\ref{equation:code}). Conditions B1 and B2 are satisfied as in the case of an even $q$. Condition B3 is satisfied by Lemma~\ref{lemma:oddChar}, and therefore $\bC$ satisfies the subspace condition in this case as well.
\end{proof}

Notice that since each repair subspace in Theorem~\ref{theorem:main} contains a basis of unit vectors, it follows that the resulting code has the access-optimal property (see Section~\ref{section:subspaceCondition}). Moreover, it is readily verified that the resulting code attains the sub-packetization bound for access-optimal codes~\eqref{eq:access-optimal-k}.

\section{Construction of an MSR Code with Three Parities}\label{section:construction3}

In this section we construct MSR codes with three parities using the framework mentioned in Subsection~\ref{section:ourTechniques}.
The size of the matrices is $\ell\times\ell$, where $\ell=3^m$ for some integer $m$. This construction requires that all three roots of unity of order three lie in the base field (which implies the necessary condition $3|q-1$). If $q$ is odd we require that $q\ge 6m+1$ and if $q$ is even we require that $q\ge 3m+1$. As the roots of unity of order three play an important role in this section, recall the following properties of these roots, some of which can be generalized for every set of roots of unity of any order.

\begin{lemma}\label{lemma:rootsOfUnity}
If $q$ is a prime power such that $3|q-1$, then $\bF_q$ contains three distinct roots of unity of order three $1,\gamma_1,\gamma_2$, which satisfy $1+\gamma_1+\gamma_2=0$, $\gamma_1^2=\gamma_2$, and $\gamma_2^{-1}=\gamma_1$.
\end{lemma}
\iffull
\begin{proof}
The existence of all roots of unity in $\bF_q$ is a consequence of the Sylow Theorems~\cite[Section XII.5]{Algebra}. In addition, it is widely known that the sum of all roots of unity of any order is~0~\cite[Chapter~2, Ex.~2.49]{LN97-FiniteFields}. The other properties follow from the fact that $\{1,\gamma_1,\gamma_2\}$ is a multiplicative subgroup of~$\bF_q^*$.
\end{proof}
\fi

From now on we assume that $3|q-1$, and $1,\gamma_1,\gamma_2$ are the three roots of unity of order three. Notice that this necessary condition rules out the possibility of using fields with characteristic~3.

Proving the nonsingular property for three parities becomes more involved, since we must show that any $1\times 1$, $2\times 2$, and $3\times 3$ block submatrix of~\eqref{equation:nonSystematicPart} is invertible. Fortunately, showing that any $1\times 1$ block submatrix (that is, an entry in~\eqref{equation:nonSystematicPart}) is invertible is trivial in our construction. Moreover, assuming that any entry of~\eqref{equation:nonSystematicPart} is invertible, and that submatrices of the form
\begin{equation*}
\begin{pmatrix}
I & A_i \\
I & A_j
\end{pmatrix}
\end{equation*}
are invertible, by using block-row operations\footnote{The three standard block operations are interchanging two block rows (columns), multiplying a block row (column) from the left (right) by a non-singular matrix, and multiplying a block row (column) by a matrix from the left (right) and adding it to another row.} we have that matrices of the form 
\begin{equation*}
\begin{pmatrix}
A_i & A_i^2 \\
A_j & A_j^2
\end{pmatrix}
\end{equation*}
are invertible as well. That is, by multiplying the top row of the latter matrix by $A_i^{-1}$ and multiplying the bottom row by $A_j^{-1}$, we get the former matrix.
Hence, for three parities, to show that an $\AS$-set satisfies the nonsingular it suffices to show the following three conditions.
\vspace{0.5cm}

\large\noindent\textbf{Conditions for the nonsingular property:} (three parities)
\normalsize
\begin{enumerate}
\item For all $i,j\in[k],i\ne j$, the matrix $\begin{pmatrix}I & A_i\\ I& A_j\end{pmatrix}$ is invertible.\label{condition:rankMetric}
\item For all $i,j\in[k], i\ne j$, the matrix $\begin{pmatrix}I & A_i^2\\ I& A_j^2\end{pmatrix}$ is invertible.\label{condition:squareRankMetric}
\item For all distinct $i,j,t\in[k]$, the matrix $\begin{pmatrix}I & A_i & A_i^2\\ I& A_j&A_j^2 \\ I & A_t &A_t^2\end{pmatrix}$ is invertible.\label{condition:3x3}
\end{enumerate}



\subsection{Three Parities from One Matching}
Recall that in Section~\ref{section:twoParities}, every matching $\cZ$ (Definition~\ref{definition:matching}) provided an $\AS$-set $(A_Z,S_Z),(A_{Z'},S_{Z'})$, where $S_Z$ is an eigenspace of $A_{Z'}$ and $S_{Z'}$ is an eigenspace of $A_Z$. Later on, we added together $\AS$-sets which were defined by different matchings satisfying the pairing condition (Definition~\ref{definition:pairingCondition}). For three parities, we consider the natural generalization of matchings in the complete \textit{3-unifrom hypergraph}.

Similarly to Section~\ref{section:twoParities}, this construction will rely on $\ell\times\ell$ matrices whose minimal polynomial is $x^3-\lambda^3$ for some $\lambda\in\bF_q^*$. All the matrices in the $\AS$-set will be similar to the matrix $A$~\eqref{equation:A}, which for $r=3$ takes the form of

\begin{eqnarray}
A\triangleq
\begin{pmatrix}\label{equation:A3}
0       & 0		  & 1 & 0 	  & 0		& 0 	  & \cdots & 0 	  & 0 	    & 0\\
1 & 0       & 0 	    & 0 	  & 0		& 0 	  & \cdots & 0 	  & 0 	    & 0\\
0       & 1 & 0 	    & 0	  	  & 0		& 0 	  & \cdots & 0 	  & 0 	    & 0\\
0       & 0       & 0		& 0 	  & 0 	    & 1 & \cdots & 0 	  & 0 	    & 0\\
0	    & 0		  & 0		& 1 & 0	    & 0		  & \cdots & 0 	  & 0 	    & 0\\
0       & 0       & 0 	    & 0 	  & 1 & 0 	  & \cdots & 0 	  & 0 	    & 0\\
\vdots  & \vdots  & \vdots  & \vdots  & \vdots  & \vdots  & \ddots & \vdots  & \vdots  & \vdots\\
0       & 0       & 0 	    & 0 	  & 0	    & 0		  & \cdots & 0 	  & 0	    & 1\\
0       & 0       & 0 	    & 0 	  & 0	    & 0		  & \cdots & 1 & 0 	    & 0\\
0       & 0       & 0 	    & 0 	  & 0	    & 0		  & \cdots & 0	      & 1 & 0\\
\end{pmatrix}.
\end{eqnarray}

According to Corollary~\ref{corollary:matrixA} we have the following lemma.


\begin{lemma}\label{lemma:matrixA}
The matrix $A$ (\ref{equation:A3}) is diagonalizable, with the following eigenspaces,
\begin{enumerate}
\item  For the eigenvalue $1$, a basis of the eigenspace is
\[
\begin{array}{llcll}
\{&(1,1,1,0,0,0,&\ldots&,0,0,0),&\\
 ~&(0,0,0,1,1,1,&\ldots&,0,0,0),&\\
 ~&     ~       &\ddots&~       &\\
 ~&(0,0,0,0,0,0,&\ldots&,1,1,1) &\}.
\end{array}
\]
\item For the eigenvalue $\gamma_1$, a basis of the eigenspace is

\[
\begin{array}{lllcll}
\{&(1,\gamma_1,\gamma_2,&0,0,0,&\ldots&,0,0,0),&\\
~&(0,0,0,&1,\gamma_1,\gamma_2,&\ldots&,0,0,0),&\\
~&     ~  &~     &\ddots&~       &\\
~&(0,0,0,&0,0,0,&\ldots&,1,\gamma_1,\gamma_2) &\}.
\end{array}
\]
\item For the eigenvalue $\gamma_2$, a basis of the eigenspace is
\[
\begin{array}{lllcll}
\{&(1,\gamma_2,\gamma_1,&0,0,0,&\ldots&,0,0,0),&\\
~&(0,0,0,&1,\gamma_2,\gamma_1,&\ldots&,0,0,0),&\\
~&     ~  &~     &\ddots&~       &\\
~&(0,0,0,&0,0,0,&\ldots&,1,\gamma_2,\gamma_1) &\}.
\end{array}
\]
\end{enumerate}
In addition, the subspace $S\triangleq\Span{e_0,e_3,e_6,\ldots}$ is an independent subspace of $A$.
\end{lemma}

The matrices in our $\AS$-set are similar to a constant multiple of the matrix $A$, and thus they are also diagonalizable. The structure of their eigenspaces, which follows from Lemma~\ref{lemma:matrixA}, is as follows.
%

\begin{lemma}\label{lemma:eigenspaces3}
If $P\in\bF_q^{\ell\times\ell}$ is an invertible matrix whose rows are $p_0,\ldots,p_{\ell-1}$, and $M\triangleq \lambda P^{-1}A P$ for some $\lambda\in\bF_q^*$, then $M$ has the following eigenspaces,
\begin{enumerate}
\item For the eigenvalue $\lambda$, a basis of the eigenspace is $\{p_{3i}+p_{3i+1}+p_{3i+2}~\vert~i\in\{0,\ldots,\ell/3-1\}\}$.
\item For the eigenvalue $\gamma_1\lambda$, a basis of the eigenspace is $\{p_{3i}+\gamma_1p_{3i+1}+\gamma_2p_{3i+2}~\vert~i\in\{0,\ldots,\ell/3-1\}\}$.
\item For the eigenvalue $\gamma_2\lambda$, a basis of the eigenspace is $\{p_{3i}+\gamma_2p_{3i+1}+\gamma_1p_{3i+2}~\vert~i\in\{0,\ldots,\ell/3-1\}\}$.
\end{enumerate}
In addition, the subspace $S\triangleq \Span{p_0,p_3,p_6,\ldots}$ is an independent subspace of $M$.
\end{lemma}


We are now in a position to describe the $\AS$-set, of size three, that is given by a single-matching. $\AS$-sets that are given by a union of single-matching $\AS$-sets will be discussed in the sequel. As mentioned earlier, all three matrices of this $\AS$-set are similar to $A$. The $\ell\times\ell$ change matrices are defined using $3\times \ell$ constituent blocks (see~\eqref{equation:frameworkP}) as follows. For $\alpha,\beta\in\bF_q^*$ and $u,v,w\in\bF_q^\ell$,  let
\begin{eqnarray}\label{equation:AlphaBetaN}
N(\alpha,\beta,u,v,w)=\begin{pmatrix}
1 & 0 & 0\\
1 & -\frac{\alpha \gamma_1}{\gamma_1-1} & \frac{\beta }{\gamma_1-1}\\
1 & \frac{\alpha }{\gamma_1-1} & -\frac{\beta \gamma_1}{\gamma_1-1}\\
\end{pmatrix}\cdot\begin{pmatrix}
u\\
v\\
w
\end{pmatrix}.
\end{eqnarray}
The determinant of the $3\times 3$ matrix in (\ref{equation:AlphaBetaN}) equals $\alpha\beta\cdot\frac{\gamma_1+1}{\gamma_1-1}$, which is nonzero, and thus $N(\alpha,\beta,u,v,w)$ is row-equivalent to a matrix whose rows are $u,v,w$, for any choice of $\alpha,\beta\in\bF_q^*$. This fact gives rise to the following necessary lemma, which can be easily proved.

\begin{lemma}\label{lemma:changeOfBasisMatrices}
If $\cZ=(Z,Z',Z'')$ is a matching, then for any choice of $\alpha,\alpha',\alpha''$ and $\beta,\beta',\beta'',$ in $\bF_q^*$, the following matrices are invertible.
\begin{eqnarray*}
P_Z&\triangleq&
\begin{pmatrix}
N(\alpha,\beta,z_0,z_0',z_0'')\\
N(\alpha,\beta,z_1,z_1',z_1'')\\
\vdots \\
N(\alpha,\beta,z_{\ell/3-1},z_{\ell/3-1}',z_{\ell/3-1}'')
\end{pmatrix},
P_{Z'}\triangleq
\begin{pmatrix}
N(\alpha',\beta',z_0',z_0'',z_0)\\
N(\alpha',\beta', z_1',z_1'',z_1)\\
\vdots \\
N(\alpha',\beta' ,z_{\ell/3-1}',z_{\ell/3-1}'',z_{\ell/3-1})
\end{pmatrix},\\
P_{Z''}&\triangleq&
\begin{pmatrix}
N(\alpha'',\beta'',z_0'',z_0,z_0')\\
N(\alpha'',\beta'',z_1'',z_1,z_1')\\
\vdots \\
N(\alpha'',\beta'',z_{\ell/3-1}'',z_{\ell/3-1},z_{\ell/3-1}')
\end{pmatrix}
\end{eqnarray*}
\end{lemma}

\begin{lemma}\label{lemma:oneMatching3}
If $\cZ=(Z,Z',Z'')$ is a matching, then for any $\lambda\in\bF_q^*$, the following $\AS$-set satisfies the subspace condition.

\begin{align*}
A_Z(\lambda)&\triangleq
\lambda\cdot P_Z^{-1}A P_Z,& S_Z&\triangleq\Span{Z}\\
A_{Z'}(\lambda)&\triangleq\lambda\cdot
P_{Z'}^{-1}A
P_{Z'},& S_{Z'}&\triangleq\Span{Z'}\\
A_{Z''}(\lambda)&\triangleq
\lambda\cdot P_{Z''}^{-1}A
P_{Z''},& S_{Z''}&\triangleq\Span{Z''}\\
\end{align*}
where $\alpha,\alpha',\alpha'',\beta,\beta',\beta''$ are nonzero field elements that will be chosen according to the field characteristic.
\end{lemma}

\iffull
\begin{proof}
For convenience of notation, denote $A_Z(\lambda),A_{Z'}(\lambda)$, and $A_{Z''}(\lambda)$ by $A_Z,A_{Z'}$, and $A_{Z''}$, respectively. To prove the independence property, it follows from Lemma~\ref{lemma:eigenspaces3} that for any matrix of the form $P^{-1}AP$, where the rows of $P$ are $\{p_0,\ldots,p_{\ell-1}\}$, the subspace $\Span{p_0,p_3,\ldots}$ is an independent subspace of $P^{-1}AP$. Notice that the vectors in $P_Z$ that correspond to rows $p_i$ with $i\equiv 0 \bmod 3$ are $\{z_0,\ldots,z_{\ell/3-1}\}$. Hence, $S_Z$ is an independent subspace of~$A_Z$. Similarly, we have that $S_{Z'},S_{Z''}$ are independent subspaces of $A_{Z'},A_{Z''}$, respectively. Therefore, the independence property is satisfied.

To show the invariance property, the eigenspaces of $A_Z,A_{Z'}$, and $A_{Z''}$, are computed according to Lemma~\ref{lemma:eigenspaces3}. The eigenspace of $A_Z$ that corresponds to the eigenvalue $\gamma_1\lambda$ is the span of vectors of the form $p_{3i}+\gamma_1p_{3i+1}+\gamma_2p_{3i+2}$, for $i\in\{0,\ldots,\ell/3-1\}$, where $p_j$ is the $j$-th row of $P_Z$. Therefore, by the definition of $N$ (\ref{equation:AlphaBetaN}), we have that the eigenspace of $A_Z$ that corresponds to the eigenvlaue $\gamma_1\lambda$ is the span of
\begin{eqnarray*}
\left\{z_i+\gamma_1\cdot\left(z_i-\frac{\alpha \gamma_1}{\gamma_1-1}z_i'+\frac{\beta}{\gamma_1-1}z_i''\right)+\gamma_2\cdot\left(z_i+\frac{\alpha }{\gamma_1-1}z_i'-\frac{\beta \gamma_1}{\gamma_1-1}z_i''\right)~\Bigg\vert~0\le i\le\ell/3-1\right\}&=&\\
\left\{(1+\gamma_2+\gamma_1)z_i+\beta z_i''~\big\vert ~ 0\le i\le\ell/3-1\right\}&=&\\
\left\{\beta z_i''~\vert~0\le i\le\ell/3-1\right\}
\end{eqnarray*}
and clearly, this is $S_{Z''}$. Similarly, we have that the eigenspace of $A_Z$ which corresponds to the eigenvalue $\gamma_2\lambda$ is $S_{Z'}$. Hence, the subspaces $S_{Z'}$ and $S_{Z''}$ are invariant subspaces of $A_Z$. By identical arguments, it can be shown that for $A_{Z'}$, the eigenspace that corresponds to the eigenvalue $\gamma_2\lambda$ is $S_{Z''}$ and the eigenspace that corresponds to the eigenvalue $\gamma_1\lambda$ is $S_Z$. Furthermore, the eigenspace of $A_{Z''}$ that corresponds to the eigenvalue $\gamma_2\lambda$ is $S_Z$, and the eigenspace that corresponds to the eigenvalue $\gamma_1\lambda$ is $S_{Z'}$. Therefore, the invariance property holds.

To show the nonsingular property, we show that Conditions \ref{condition:rankMetric}-\ref{condition:3x3} are met. To prove Condition~\ref{condition:rankMetric} it should be shown that $\rank(A_Z-A_{Z''})=\ell$, $\rank(A_Z-A_{Z'})=\ell$, and $\rank(A_{Z'}-A_{Z''})=\ell$. To show that $\rank(A_Z-A_{Z''})=\ell$, it will be proved that for all $0\le i\le \ell/3-1$, $\{z_i,z_i',z_i''\}\subseteq\image(A_Z-A_{Z''})$. Since $Z\cup Z' \cup Z''$ is a basis of $\bF_q^\ell$, the claim will follow.


By the structure of the matrix $A$, it can be easily verified that if $P$ is an invertible matrix whose rows are $\{p_0,\ldots,p_{\ell-1}\}$, then
\begin{eqnarray*}
p_{3i}\lambda P^{-1}AP&=&\lambda p_{3i+2}\\
p_{3i+1}\lambda P^{-1}AP&=&\lambda p_{3i}\\
p_{3i+2}\lambda P^{-1}AP&=&\lambda p_{3i+1},
\end{eqnarray*}
and therefore,
\begin{eqnarray*}
z_iA_Z&=&\lambda\left(z_i+\frac{\alpha }{\gamma_1-1}z_i'-\frac{\beta \gamma_1}{\gamma_1-1}z_i''\right)
\\z_i''A_{Z''}&=&\lambda\left(z_i''+\frac{\alpha'' }{\gamma_1-1}z_i-\frac{\beta'' \gamma_1}{\gamma_1-1}z_i'\right).
\end{eqnarray*}
Hence, since $S_{Z'},S_{Z''}$ are eigenspaces of $A_Z$ and $S_{Z},S_{Z'}$ are eigenspaces of $A_{Z''}$, it follows that
\begin{eqnarray*}
\nonumber z_i(A_Z-A_{Z''})&=&z_iA_Z-z_iA_{Z''}\\\nonumber
&=&\lambda\left(z_i+\frac{\alpha
}{\gamma_1-1}z_i'-\frac{\beta \gamma_1}{\gamma_1-1}z_i'' \right)-\gamma_2\lambda z_i\\\nonumber
&=&\lambda\left((1-\gamma_2)z_i+\frac{\alpha
}{\gamma_1-1}z_i'-\frac{\beta \gamma_1}{\gamma_1-1}z_i''\right)\\
z_i'(A_Z-A_{Z''})&=&z_i'A_Z-z_i'A_{Z''}=\lambda\left(\gamma_2-\gamma_1\right)z_i'\\
\nonumber z_i''(A_Z-A_{Z''})&=&z_i''A_Z-z_i''A_{Z''}\\ \nonumber
&=&\gamma_1\lambda z_i''-\lambda\left(z_i''+\frac{\alpha''}{\gamma_1-1}z_i-\frac{\beta'' \gamma_1}{\gamma_1-1}z_i'\right)\\ \nonumber
&=&\lambda\left(-\frac{\alpha''}{\gamma_1-1}z_i+\frac{\beta'' \gamma_1}{\gamma_1-1}z_i'+(\gamma_1-1)z_i''\right),
\end{eqnarray*}
which in a more convenient matrix notation becomes
\begin{equation}\label{equation:rankMetric23}
\begin{pmatrix}
z_i\\z_i'\\z_i''
\end{pmatrix}\cdot(A_Z-A_{Z''})=\lambda
\begin{pmatrix}
1-\gamma_2 & \frac{\alpha}{\gamma_1-1} & -\frac{\beta\gamma_1}{\gamma_1-1}\\
0 & \gamma_2-\gamma_1 & 0 \\
-\frac{\alpha''}{\gamma_1-1} & \frac{\beta''\gamma_1}{\gamma_1-1}& \gamma_1-1
\end{pmatrix}
\cdot\begin{pmatrix}
z_i\\z_i'\\z_i''
\end{pmatrix}\triangleq \lambda\cdot \Phi \cdot\begin{pmatrix}
z_i\\z_i'\\z_i''
\end{pmatrix}.
\end{equation}
To show that $z_i,z_i',z_i''\in \image (A_Z-A_Z'')$, it suffices to show that $\Phi$ is invertible. Since $\lambda(\gamma_2-\gamma_1)\ne 0$, it is enough to prove that 
\begin{eqnarray}\label{equation:detOfRankDist}
\begin{pmatrix}
1-\gamma_2 & -\frac{\beta \gamma_1}{\gamma_1-1}\\
-\frac{\alpha''}{\gamma_1-1} & \gamma_1-1
\end{pmatrix}
\end{eqnarray}
is invertible. Simple calculations, which follow from the properties of $\gamma_1,\gamma_2$ (Lemma~\ref{lemma:rootsOfUnity}) show that this matrix has a nonzero determinant if and only if\footnote{Any $n\in \bN$ denotes the summation of $n$ copies of the unity of the field $\bF_q$.} $\alpha''\beta\ne 9$.
Similar arguments show that $\rank(A_Z-A_{Z'})=\ell$ and $\rank(A_{Z'}-A_{Z''})=\ell$ if and only if $\alpha'\beta''\ne 9$ and $\alpha\beta'\ne 9$, respectively.
Proper $\alpha,\alpha',\alpha'',\beta,\beta',\beta''$ will be chosen is the sequel.


To show Condition~\ref{condition:squareRankMetric}, it must proved that the difference between any two \textit{squares} of matrices in the $\AS$-set has full rank. Fortunately, the squares of the matrices in the $\AS$-set present a very similar behavior to the matrices themselves. That is, if we denote any matrix in the $\AS$-set by $\hat{A}=\lambda P^{-1}AP$, then $\hat{A}^2 = \lambda^2 P^{-1}A^2P$, and
\begin{align*}
\text{if } v\hat{A}&=\lambda \gamma_1 v~\text{, then }v\hat{A}^2=\lambda^2 \gamma_2 v\text{, }\\
\text{if } v\hat{A}_i&=\lambda \gamma_2 v~\text{, then }v\hat{A}^2=\lambda^2 \gamma_1v.\\
\end{align*}

Hence, if $S_1,S_2$ are the eigenspaces of $\hat{A}$ that correspond to the eigenvalues $\lambda\gamma_1,\lambda\gamma_2$, respectively, then the eigenspaces of $\hat{A}^2$ that correspond to the eigenvalues $\lambda^2\gamma_1,\lambda^2\gamma_2$ are $S_2,S_1$, respectively. Moreover, we have that
\begin{eqnarray*}
z_iA_Z^2&=&\lambda^2\left(z_i-\frac{\alpha \gamma_1}{\gamma_1-1}z_i'+\frac{\beta}{\gamma_1-1}z_i''\right)
\\z_i''A_{Z''}^2&=&\lambda^2\left(z_i''-\frac{\alpha'' \gamma_1}{\gamma_1-1}z_i+\frac{\beta''}{\gamma_1-1}z_i'\right),
\end{eqnarray*}
and hence,
\begin{eqnarray}\label{equation:squareRankMetric}
\nonumber z_i(A_Z^2-A_{Z''}^2)&=&\lambda^2\left((1-\gamma_1)z_i-\frac{\alpha \gamma_1}{\gamma_1-1}z_i'+\frac{\beta}{\gamma_1-1}z_i''\right)\\
z_i'(A_Z^2-A_{Z''}^2)&=&\lambda^2\left(\gamma_1-\gamma_2\right)z_i'\\
\nonumber z_i''(A_Z^2-A_{Z''}^2)&=&\lambda^2\left(\frac{\alpha'' \gamma_1}{\gamma_1-1}z_i-\frac{\beta'' }{\gamma_1-1}z_i'+(\gamma_2-1)z_i''\right),
\end{eqnarray}
which in matrix notation becomes
\begin{equation*}
\begin{pmatrix}
z_i\\z_i'\\z_i''
\end{pmatrix}\cdot (A_Z^2-A_{Z''}^2) = \lambda^2
\begin{pmatrix}
1-\gamma_1 & -\frac{\alpha\gamma_1}{\gamma_1-1} & \frac{\beta}{\gamma_1-1}\\
0 & \gamma_1-\gamma_2 & 0 \\
\frac{\alpha''\gamma_1}{\gamma_1-1}& -\frac{\beta''}{\gamma_1-1}&\gamma_2-1
\end{pmatrix}\cdot \begin{pmatrix}
z_i\\z_i'\\z_i''
\end{pmatrix}\triangleq\lambda^2\cdot\Psi\cdot\begin{pmatrix}
z_i\\z_i'\\z_i''
\end{pmatrix}.
\end{equation*}
As in Condition~\ref{condition:rankMetric}, we show that $\Psi$ is invertible. Surprisingly, we have that $\det \Phi=-\det \Psi$, and thus Condition~\ref{condition:rankMetric} and Condition~\ref{condition:squareRankMetric} are implied by the same requirements $\alpha''\beta\ne 9$, $\alpha'\beta''\ne 9$, and $\alpha \beta '\ne 9$.


To prove Condition~\ref{condition:3x3}, first notice that the following two matrices are row-equivalent, provided that $A_{Z'}-A_Z$ is invertible.
\begin{eqnarray}\label{equation:rowEquivalence3x3}
\begin{pmatrix}
I & I & I \\
A_Z & A_{Z'} & A_{Z''} \\
A_Z^2 & A_{Z'}^2 & A_{Z''}^2 \\
\end{pmatrix}
\vspace{2cm},
\begin{pmatrix}
I & I & I \\
0 & I & (A_{Z'}-A_Z)^{-1}(A_{Z''}-A_Z) \\
0 & 0 & \underline{(A_{Z}^2-A_{Z''}^2)-(A_{Z}^2-A_{Z'}^2)(A_{Z}-A_{Z'})^{-1}(A_{Z}-A_{Z''})} \\
\end{pmatrix}
\end{eqnarray}
Therefore, Condition~\ref{condition:3x3} is met if and only if  the underlined matrix in~(\ref{equation:rowEquivalence3x3}) is invertible. Since $A_{Z''}-A_Z$ is invertible (given a proper choice for $\alpha,\ldots,\beta''$), it follows that the underlined matrix is invertible if and only if
\begin{eqnarray*}
L\triangleq(A_{Z}^2-A_{Z''}^2)(A_{Z}-A_{Z''})^{-1}-(A_{Z}^2-A_{Z'}^2)(A_{Z}-A_{Z'})^{-1}
\end{eqnarray*}
is invertible. To show that $L$ is indeed invertible, we show that $\{z_i,z_i',z_i''\}\subseteq\image(L)$ for every $i\in\{0,\ldots,\ell/3-1\}$. By (\ref{equation:squareRankMetric}), and by the corresponding equations for $A_Z^2-A_{Z'}^2$, we have that
\begin{align}\label{equation:L}
\lambda^{-2}z_iL&=\phantom{-}(1-\gamma_1)z_i(A_Z-A_{Z''})^{-1}-\frac{\alpha \gamma_1}{\gamma_1-1}z_i'(A_Z-A_{Z''})^{-1}+\frac{\beta}{\gamma_1-1}z_i''(A_Z-A_{Z''})^{-1}\nonumber\\
~&\phantom{-}-(1-\gamma_2)z_i(A_Z-A_{Z'})^{-1}+\frac{\alpha \gamma_1}{\gamma_1-1}z_i'(A_Z-A_{Z'})^{-1}-\frac{\beta}{\gamma_1-1}z_i''(A_Z-A_{Z'})^{-1} \nonumber\\
\lambda^{-2} z_i'L&= (\gamma_1-\gamma_2)z_i'(A_Z-A_{Z''})^{-1}+\\
~&\phantom{==}\frac{\beta'}{\gamma_1-1}z_i(A_Z-A_{Z'})^{-1}-(\gamma_1-1)z_i'(A_Z-A_{Z'})^{-1}-\frac{\alpha' \gamma_1}{\gamma_1-1}z_i''(A_Z-A_{Z'})^{-1} \nonumber \\ 
\nonumber \lambda^{-2}z_i''L&=\phantom{-} \frac{\alpha'' \gamma_1}{\gamma_1-1}z_i(A_Z-A_{Z''})^{-1}-\frac{\beta''}{\gamma_1-1}z_i'(A_Z-A_{Z''})^{-1}+(\gamma_2-1)z_i''(A_Z-A_{Z''})^{-1}\\\nonumber
~ &\phantom{-}-(\gamma_2-\gamma_1)z_i''(A_Z-A_{Z'})^{-1}.\nonumber
\end{align}
In matrix notation, this turns to
\begin{align}\label{equation:Lshort}
\nonumber\lambda^{-2}\begin{pmatrix}
z_i\\z_i'\\z_i''
\end{pmatrix}L&=
\begin{pmatrix}
1-\gamma_1& -\frac{\alpha\gamma_1}{\gamma_1-1}&\frac{\beta}{\gamma_1-1}\\
0&\gamma_1-\gamma_2&0\\
\frac{\alpha''\gamma_1}{\gamma_1-1}&-\frac{\beta''}{\gamma_1-1}&\gamma_2-1
\end{pmatrix}\cdot \begin{pmatrix}
z_i\\z_i'\\z_i''
\end{pmatrix}(A_Z-A_{Z''})^{-1}
\\~&+\begin{pmatrix}
\gamma_2-1 & \frac{\alpha\gamma_1}{\gamma_1-1}&-\frac{\beta}{\gamma_1-1}\\
\frac{\beta'}{\gamma_1-1}&1-\gamma_1&-\frac{\alpha'\gamma_1}{\gamma_1-1}\\
0&0&\gamma_1-\gamma_2
\end{pmatrix}\cdot\begin{pmatrix}
z_i\\z_i'\\z_i''
\end{pmatrix}(A_Z-A_{Z'})^{-1}.
\end{align}

Now, using~\eqref{equation:rankMetric23}, and the similar equation for $A_{Z}-A_{Z'}$,
the expressions 
\begin{equation*}
\begin{pmatrix}
z_i\\z_i'\\z_i''
\end{pmatrix}(A_Z-A_{Z'})^{-1}, ~\begin{pmatrix}
z_i\\z_i'\\z_i''
\end{pmatrix}(A_Z-A_{Z''})^{-1}
\end{equation*}
can be given as functions of $z_i,z_i',z_i''$, $\lambda$, and the matrix $\Phi$, e.g., by multiplying~\eqref{equation:rankMetric23} from the right by $(A_Z-A_{Z''})^{-1}$, and from the left by $\lambda^{-1}\Phi^{-1}$. By performing this substitution, we have that~\eqref{equation:Lshort} may be written as 
\begin{equation*}
\lambda^{-2}\begin{pmatrix}
z_i\\z_i'\\z_i''
\end{pmatrix}L = \Upsilon \begin{pmatrix}
z_i\\z_i'\\z_i''
\end{pmatrix},
\end{equation*}
for some $3\times 3$ matrix $\Upsilon$ whose entries are functions of $1,\gamma_1,\gamma_2,\alpha,\ldots,\beta''$. After some tedious calculations, we have that

\begin{equation*}
\det \Upsilon = \frac{\gamma_2-1}{-3\lambda^3}\cdot\frac{\alpha\alpha'\alpha''\beta\beta'\beta'' + 27\alpha\alpha'\alpha'' + 27\beta\beta'\beta'' + 729}{(\gamma_2-1)(\alpha\alpha''\beta\beta') - 9(\gamma_2+1)\alpha''\beta -9(\gamma_2+1)\alpha\beta' +81(\gamma_2+1)}.
\end{equation*}

We show that every possible field has a simple corresponding choice of values from $\{1,\gamma_1,\gamma_2\}$ to $\alpha,\ldots,\beta''$ such that the conditions $\det\Delta \ne 0$, $a''\beta\ne 9$,$\alpha'\beta''\ne 9$, and $\alpha\beta'\ne 9$ are satisfied.
\begin{itemize}
\item[\textbf{Case 1.}] If the characteristic is 2, choose
$\alpha = 1,\alpha' = \gamma_1,\alpha'' = \gamma_1,\beta = 1,\beta' = \gamma_2,\beta'' = 1$, and then,
\begin{eqnarray*}
\alpha''\beta = \gamma_1,
\alpha'\beta'' = \gamma_1,
\alpha\beta' = \gamma_2,
\det\Upsilon = \frac{622\gamma_2 - 781}{\lambda^3(-273\gamma_2 - 273)} = \frac{1}{\lambda^3(\gamma_2+1)}
\end{eqnarray*}
and since $9=1\notin\{\gamma_1,\gamma_2\}$, it follows that all conditions are satisfied.
\item[\textbf{Case 2.}] If the characteristic is $7$, choose $\alpha=\gamma_2,\alpha'=1,\alpha''=1,\beta=1,\beta'=\gamma_1,\beta''=1$. Hence,
\begin{eqnarray*}
	\alpha''\beta = 1,
	\alpha'\beta'' = 1,
	\alpha\beta' = 1,
	\det \Upsilon = \frac{703(1-\gamma_2)}{192\lambda^3 (\gamma_2+1)}=\frac{1-\gamma_2}{\lambda^3 (\gamma_2+1)},
\end{eqnarray*}
and since $9\ne 1$, all conditions are met as well.
\item[\textbf{Case 3.}] If the characteristic neither 2 nor 7, choose $\alpha=\ldots=\beta''=1$, and then,
\begin{eqnarray*}
\alpha''\beta = 1,
\alpha'\beta'' = 1,
\alpha\beta' = 1,
\det\Upsilon = \frac{49(1-\gamma_2)}{12\lambda^3(\gamma_2 + 1)}.
\end{eqnarray*}
Notice that we may divide by $12=2^2\cdot 3$ since the characteristic is neither 2 nor 3. Since $9\ne 1$ and $7\ne 0$, all conditions are satisfied.
\end{itemize}
\end{proof}
\fi

\subsection{Three Parities from Two Matchings}
We are now in a position to describe a construction of an $\AS$-set for three parities from more than one matching.
\iffull
In what follows we show that $\AS$-sets which correspond to two matchings which satisfy the pairing condition, may be united to achieve a larger $\AS$-set, given a proper choice of $\lambda_x,\lambda_y$. In the following lemmas of this subsection,  $\cX=(X,X',X''),\cY=(Y,Y',Y'')$ are two matchings that satisfy the pairing condition (Definition~\ref{definition:pairingCondition}) and\footnote{We omit the notations of $\lambda_x,\lambda_y$ for convenience.}
\begin{eqnarray*}
	C_\cX&\triangleq& \{(A_X,S_X), (A_{X'},S_{X'}),(A_{X''},S_{X''})\},\\
	C_\cY&\triangleq& \{(A_Y,S_Y), (A_{Y'},S_{Y'}),(A_{Y''},S_{Y''})\}.
\end{eqnarray*}

\begin{lemma}\label{lemma:missingCase}
	If $C\in\{A_X,A_{X'},A_{X''}\}$ then the eigenspace of $C$ which corresponds to the eigenvalue $\lambda_x$ is of the form $\Span{\left\{c_0 x_i+c_1 x_i'+c_2x_i''\right\}_{i=0}^{\ell/r-1}}$ for some nonzero constants $c_i's$. A similar claim holds for $D\in\{A_Y,A_{Y'},A_{Y''}\}$.
\end{lemma}

\begin{proof}
	According to Lemma~\ref{lemma:eigenspaces3} and Lemma~\ref{lemma:oneMatching3}, the eigenspace of $A_X$ which corresponds to the eigenvalue $\lambda_x$ is $\Span{\{3x_i-\alpha x_i'-\beta x_i''\}_{i=0}^{\ell/3-1}}$, the eigenspace of $A_X'$ which corresponds to the eigenvalue $\lambda_x$ is $\Span{\{3x_i'-\alpha' x_i''-\beta' x_i\}_{i=0}^{\ell/3-1}}$, and the eigenspace of $A_X''$ which corresponds to the eigenvalue $\lambda_x$ is $\Span{\{3x_i''-\alpha'' x_i-\beta'' x_i'\}_{i=0}^{\ell/3-1}}$. The claim for $D$ is similar.
\end{proof}

Lemma~\ref{lemma:missingCase}, together with Remark~\ref{remark:missingCase} and Corollary~\ref{corollary:eigenspaceIntesections}, implies that any eigenspace of $C\in\{A_X,A_{X'},A_{X''}\}$ and an eigenspace of $D\in\{A_Y,A_{Y'},A_{Y''}\}$ intersect at a subspace of dimension $\ell/9$. This gives rise to the following lemma.

\begin{lemma}\label{lemma:simDig}
Every pair of a matrix $C$ from $C_\cX$ and a matrix $D$ from $C_\cY$ are simultaneously diagonalizable. Furthermore, $C^2$ and $D^2$ are simultaneously diagonalizable as well.
\end{lemma}
\begin{proof}
Let $S_1,S_2,$ and $S_3$ be the eigenspaces of $C$, let $S_4,S_5,$ and $S_6$ be the eigenspaces of $D$, and note since $C$ and $D$ are diagonalizable, we have that $S_1+S_2+S_3=S_4+S_5+S_6=\bF_q^\ell$. According to Corollary~\ref{corollary:eigenspaceIntesections}, for any $i\in[3]$, we have that $\dim(S_i\cap S_4)=\dim(S_i\cap S_5)=\dim(S_i\cap S_6)=\ell/9$. Therefore, since $S_4+S_5+S_6=\bF_q^\ell$, it follows that $S_i$ contributes exactly $\ell/3$ mutual linearly independent eigenvectors of $C$ and $D$. Since $S_1+S_2+S_3=\bF_q^\ell$, it follows that there exists $\ell$ mutual linearly independent eigenvectors of $C$ and $D$, and hence they are mutually diagonalizable.

Now, since $C$ and $D$ are simultaneously diagonalizable, it follows that there exists an invertible matrix $P$, and diagonal matrices $E$ and $F$, such that $C=P^{-1}EP$ and $D=P^{-1}FP$. This implies that $C^2=P^{-1}E^2P$ and $D=P^{-1}F^2P$, and therefore, since the square of a diagonal matrix is a diagonal matrix, it follows that $C^2$ and $D^2$ are simultaneously diagonalizable as well.
\end{proof}

\begin{lemma}\label{lemma:simDiagAreInv}
	If $C$ and $D$ are $\ell\times\ell$ simultaneously diagonalizable matrices with no mutual eigenvalues, then $C-D$ is invertible.
\end{lemma}
\iffull
\begin{proof}
	Let $p_0,\ldots,p_{\ell-1}$ be a basis of mutual eigenvectors of $C$ and $D$. Clearly, for all $i\in\{0,\ldots,\ell-1\}$ we have that $p_i(C-D)=p_iC-p_iD=\lambda_{i,1}p_i-\lambda_{i,2}p_i=(\lambda_{i,1}-\lambda_{i,2})p_i$, where $\lambda_{i,1},\lambda_{i,2}$ are the eigenvalues which correspond to $p_i$. Since $\lambda_{i,1}\ne\lambda_{i,2}$, we have that $p_i\in\image(C-D)$. Therefore, since $p_0,\ldots,p_{\ell-1}$ is a basis, we have that $C-D$ is invertible.
\end{proof}
\fi

\fi
The following lemma shows that it is possible to unite the $\AS$-sets $C_\cX, C_\cY$ which were constructed using different matchings that satisfy the pairing condition (Definition~\ref{definition:pairingCondition}), as long as a simple condition regarding the chosen constants $\lambda_x,\lambda_y$ is met. \iffull
\fi

\begin{lemma}\label{lemma:twoMatchings3}
If $\lambda_x^6\ne \lambda_y^6$, then $C_\cX\cup C_\cY$ satisfies the subspace condition.
\end{lemma}

\iffull
\begin{proof}
Note that the invariance, independence, and nonsingular properties which involve matrices and subspaces from one matching follow immediately from Lemma~\ref{lemma:oneMatching3}. It remains to prove the cases of the invariance property and the nonsingular property which involve matrices from different matchings.

To prove the invariance property, let $S\in\{S_X,S_{X'},S_{X''}\}$ and $D\in\{A_Y,A_{Y'},A_{Y''}\}$. According to Corollary~\ref{corollary:eigenspaceIntesections}, Remark~\ref{remark:missingCase}, and Lemma~\ref{lemma:missingCase}, all eigenspaces of $D$ intersect with $S$ at a subspace of dimension $\ell/9$. Since the eigenspaces of $D$ are disjoint and span the entire space, it follows that $S$ contains a basis of eigenvectors of $D$. Therefore, $SD=S$. A similar proof holds if $S\in\{S_Y,S_{Y'},S_{Y''}\}$ and $D\in\{A_X,A_{X'},A_{X''}\}$.

	It remains to prove Conditions~\ref{condition:rankMetric},\ref{condition:squareRankMetric} and~\ref{condition:3x3} of the nonsingular property. To prove Condition~\ref{condition:rankMetric}, let $C\in\{A_X,A_{X'},A_{X''}\}$ and $D\in\{A_Y,A_{Y'},A_{Y''}\}$. By Lemma~\ref{lemma:simDig}, we have that $C$ and $D$ are simultaneously diagonalizable. Moreover, according to Lemma~\ref{lemma:eigenspaces3}, the eigenvalues of $C$ are $\lambda_x,\gamma_1\lambda_x,\gamma_2\lambda_x$, and the eigenvalues of $D$ are $\lambda_y,\gamma_1\lambda_y,\gamma_2\lambda_y$. Since $\lambda_x^6\ne \lambda_y^6$, we have that  $\{\lambda_x,\gamma_1\lambda_x,\gamma_2\lambda_x\}\cap\{\lambda_y,\gamma_1\lambda_y,\gamma_2\lambda_y\}=\varnothing$. Hence, Lemma~\ref{lemma:simDiagAreInv} implies that $\rank (C-D)=\ell$, which implies Condition~\ref{condition:rankMetric} of the nonsingular property. Condition~\ref{condition:squareRankMetric} also follows similarly - by Lemma~\ref{lemma:eigenspaces3} we have that the eigenvalues of $C^2$ are $\lambda_x^2,\gamma_1\lambda_x^2,\gamma_2\lambda_x^2$, the eigenvalues of $D^2$ are $\lambda_y^2,\gamma_1\lambda_y^2,\gamma_2\lambda_y^2$, and since $\lambda_x^6\ne \lambda_y^6$, it follows that $\{\lambda_x^2,\gamma_1\lambda_x^2,\gamma_2\lambda_x^2\}\cap\{\lambda_y^2,\gamma_1\lambda_y^2,\gamma_2\lambda_y^2\}=\varnothing$ (see Lemma~\ref{lemma:toThe6} in~\nameref{appendixC}). Hence, Lemma~\ref{lemma:simDiagAreInv} implies that $\rank(C^2-D^2)=\ell$, which implies Condition~\ref{condition:squareRankMetric} of the nonsingular property.

As for Condition~\ref{condition:3x3}, let $A_i,A_j,A_t$ be matrices in the $\AS$-set, which correspond to at least two distinct matchings. Recall that
\begin{eqnarray*}
\begin{pmatrix}
I & I & I \\
A_i & A_j & A_t \\
A_i^2 & A_j^2 & A_t^2 \\
\end{pmatrix}
\end{eqnarray*}
is invertible if and only if
\begin{eqnarray*}
(A_{t}^2-A_{i}^2)-(A_{j}^2-A_{i}^2)(A_{j}-A_i)^{-1}(A_{t}-A_i)
\end{eqnarray*}
is invertible (see the proof of Lemma~\ref{lemma:oneMatching3}). W.l.o.g assume that $A_i$ and $A_j$ correspond to different matchings, and so do $A_i$ and $A_t$. According to Lemma~\ref{lemma:simDig}, we have that $A_i$ commutes with $A_j$ and $A_t$. Hence,
\begin{eqnarray*}
(A_{t}^2-A_{i}^2)-(A_{j}^2-A_{i}^2)(A_{j}-A_i)^{-1}(A_{t}-A_i)&=\\(A_t+A_i)(A_t-A_i)-(A_j+A_i)(A_j-A_i)(A_j-A_i)^{-1}(A_t-A_i)&=\\
(A_t+A_i)(A_t-A_i)-(A_j+A_i)(A_t-A_i)&
\end{eqnarray*}
Multiplying from the right by $(A_t-A_i)^{-1}$, which exists by Condition~\ref{condition:rankMetric}, yields,
\begin{eqnarray*}
(A_t+A_i)-(A_j+A_i)&=A_t-A_j
\end{eqnarray*}
which is invertible by Condition~\ref{condition:rankMetric}.
\end{proof}
\fi

\subsection{Construction of Matchings for Three Parities}
In this subsection we present a set of matchings $\{\cX_i\}_{i\in[m]}$ such that any two satisfy the pairing condition, and construct the resulting $\AS$-set. Recall that each vertex in the complete 3-unifrom hypergraph $K_\ell^3$ is represented by a unique unit vector of length $\ell$. For convenience, we describe this set of matchings by considering vertex $e_i$ as the integer $i$ in its \textit{ternary} representation. The construction of a proper set of matchings relies on the following definition, which is the three parity equivalent of Definition~\ref{definition:booleanCube}.

\begin{definition}
Given an integer $i\in[m]$ and a value $b\in\{0,1,2\}$, the \emph{ternary cube} $C(i,b)$ is the set of all length $m$ vectors over $\{0,1,2\}$ that have $b$ in entry $i$. That is,
\[
C(i,b)\triangleq\left\{x\in\{0,1,2\}^m~\vert~x_i=b\right\}.
\]
For convenience, we consider the elements in a ternary cube as \emph{ordered} according to the lexicographic order (see Example~\ref{example:ternaryCube} below).
\end{definition}

\begin{example}\label{example:ternaryCube}
If $m=3$ then the ternary cube $C(2,2)$ is the set $\{v_1,\ldots,v_9\}$ such that
\[
(v_1,\ldots,v_9)=(020,021,022,120,121,122,220,221,222).
\]
\end{example}

\begin{definition}\label{definition:matchings3}
For any $m\in\bN$, define $m$ matchings $\{\cX=(X_i,X_i',X_i'')\}_{i\in[m]}$ as follows
\begin{eqnarray*}
\cX_{i}:~~
\begin{cases}
X_i = C(i,0)\\
X_i' = C(i,1)\\
X_i'' = C(i,2).
\end{cases}
\end{eqnarray*}
\end{definition}

\begin{lemma}\label{lemma:pairingCondition3}
The matchings from Definition~\ref{definition:matchings3} satisfy the pairing condition.
\end{lemma}
\iffull
\begin{proof}
Let $\cX_i=(X_i,X_i',X_i''),\cX_j=(X_j,X_j',X_j'')$ be two distinct matchings. To show that these matchings satisfy the pairing condition, we must show that any edge from $\cX_i$ is contained in either of $X_j,X_j',X_j''$. Let $(x_t,x_t',x_t'')$ be an edge in $\cX_i$ for some $t\in[\ell/3]$. By Definition~\ref{definition:matchings3}, we have that $x_t,x_t'$, and $x_t''$ have the same value in every entry other than entry $i$, and $(x_t)_i=0,(x_t')_i=1$, and $(x_t'')_i=2$. Hence, the $j$-th entry of $x_t,x_t'$, and $x_t''$ is equal, and hence $\{x_t,x_t',x_t''\}$ is contained in either of $X_j,X_j',X_j''$ as required. The other direction is symmetric.
\end{proof}
\fi

We conclude with the following theorem.

\begin{theorem}\label{theorem:3ParitiesAccessOptimal}
If $m$ is a positive integer, and $q$ is a prime power such that \begin{enumerate}
\item if $q$ is odd, then $3|q-1$ and $q\ge 6m+1$,
\item if $q$ is even, then $3|q-1$ and $q\ge 3m+1$,
\end{enumerate}
then there exists an explicitly defined $\AS$-set $\bC_1$ of size $3m$ and $3^m\times 3^m$ matrices over $\bF_q$, which satisfies the subspace condition.
\end{theorem}
\iffull
\begin{proof}
Let $\{\cX_i=(X_i,X_i',X_i'')\}_{i\in[m]}$ be the set of matchings from Definition~\ref{definition:matchings3}, which by Lemma~\ref{lemma:pairingCondition3} satisfies the pairing condition (Definition~\ref{definition:pairingCondition}). Let $\{\lambda_i\}_{i\in[m]}$ be any set of distinct nonzero elements of $\bF_q$ such that $\lambda_i^6\ne \lambda_j^6$ for any $i\ne j$. Notice that the existence of such set is guaranteed in fields of either odd or even characteristic. The former is due to $q\ge 6m+1$, where the latter is due to the fact that $\lambda_i^6=\lambda_j^6$ if and only if $\lambda_i^3=\pm\lambda_j^3$, which implies $\lambda_i^3=\lambda_j^3$ in fields with even characteristic, and thus $q\ge 3m+1$ suffices.
For any $\cX_i,i\in[m]$, define the code
\[
C_{\cX_i}=\{(A_{X}(\lambda_i),S_X),(A_{X'}(\lambda_i),S_{X'}),(A_{X''}(\lambda_i),S_{X''})\}
\]
as defined in Lemma~\ref{lemma:oneMatching3}, and let $\bC_1\triangleq\cup_{i\in[m]}C_{\cX_i}$. Notice that in $\bC_1$, Condition~\ref{condition:3x3} of the nonsingular property might involve matrices from \textit{three} different matchings, rather than two, as was considered in the proof of Lemma~\ref{lemma:twoMatchings3}. However, the proof of Condition~\ref{condition:3x3} in Lemma~\ref{lemma:twoMatchings3} requires two pairs of matrices among $\{A_i,A_j,A_t\}$ to belong to distinct matchings, in order for the resulting $3\times 3$ matrix to be invertible. This requirement is trivially satisfied also when considering matrices from three different matchings. Hence, since the pairing condition is satisfied, and since $\lambda_i^6\ne\lambda_j^6$ for all $i\ne j$, it follows by Lemma~\ref{lemma:twoMatchings3} that $\bC$ satisfies the subspace condition.
\end{proof}
\fi

\section{An Improved Construction over a Larger Field}\label{section:rPlus1}
In this section, a construction with $r=3$ and $k=(r+1)m=4m$ is presented. This construction requires a field larger than the one in Section~\ref{section:construction3}, yet still linear in $m$. This construction is comparable to~\cite{LongMDS} in terms of the parameter $k$, but outperforms it in terms of explicitness and field size. As in~\cite{LongMDS}, the construction considered in this section is not access-optimal, and is not known to achieve the sub-packetization bound. The ideas behind the construction follow the outline described in Section~\ref{section:ourTechniques}.

\subsection{A Code from One Matching}
A matching $\cZ=(Z,Z',Z'')$ will provide an $\AS$-set of size $r+1=4$, denoted by
\[
(A_{Z},S_{Z}),(A_{Z'},S_{Z'}),(A_{Z''},S_{Z''}),(A_{Z^*},S_{Z^*}).
\]
As in Section~\ref{section:construction3}, we assume that $3|q-1$ in order to have three roots of unity of order 3, denoted by $1,\gamma_1,\gamma_2$. The matrices in this construction are of the form $P^{-1}AP$, where $A$ was defined in~\eqref{equation:A3}.
The matrices $P$, as in~\eqref{equation:frameworkP}, consists of constituent $3\times\ell$ matrices which are defined using the following operator $N$. For $u,v,w\in\bF_q^\ell$, let
\begin{eqnarray}
N(u,v,w)&\triangleq&
\begin{pmatrix}
1 & 1 & 1\\
1 & \gamma_2 & \gamma_1 \\
1 & \gamma_1 & \gamma_2
\end{pmatrix}
\begin{pmatrix}
u\\
v\\
w
\end{pmatrix}=
\begin{pmatrix}
u+v+w\\
u+\gamma_2 v+\gamma_1 w\\
u+\gamma_1 v+\gamma_2 w\\
\end{pmatrix}.
\label{equation:N}
\end{eqnarray}
Notice that the $3\times \ell$ matrix $N(u,v,w)$ is row equivalent to a matrix whose rows are $u,v,w$, since $N(u,v,w)$ is defined as the multiplication of a matrix whose rows are $u,v,w$ by a Vandermonde matrix (since $\gamma_1^2=\gamma_2$ and $\gamma_1^2=\gamma_2$).
\begin{lemma}\label{lemma:SpacesMatrixA}
If $\{u_i\}_{i=0}^{\ell/3-1}\cup\{v_i\}_{i=0}^{\ell/3-1}\cup\{w_i\}_{i=0}^{\ell/3-1}$ is a basis of $\bF_q^\ell$, then the matrix $P^{-1}AP$, where
\begin{eqnarray*}
P &\triangleq &
\begin{pmatrix}
N( u_0, &v_0, &w_0)\\
 \vdots &\vdots &\vdots\\
N(u_{\ell/3-1}, &v_{\ell/3-1}, &w_{\ell/3-1})
\end{pmatrix},
\end{eqnarray*}
has $\Span{\{u_i\}_{i=0}^{\ell/3-1}}$ as an eigenspace for the eigenvalue $1$, $\Span{\{v_i\}_{i=0}^{\ell/3-1}}$ as an eigenspace for the eigenvalue $\gamma_1$, and $\Span{\{w_i\}_{i=0}^{\ell/3-1}}$ as an eigenspace for the eigenvalue $\gamma_2$. In addition, the subspace $\Span{\{u_i+v_i+w_i\}_{i=0}^{\ell/3-1}}$ is an independent subspace.
\end{lemma}

\begin{proof}
According to Lemma~\ref{lemma:eigenspaces3}, the matrix $P^{-1}AP$, where the rows of $P$ are $p_0,\ldots,p_{\ell-1}$, has the following eigenspaces.
\begin{enumerate}
\item For the eigenvalue $1$, a basis of the eigenspace is
\begin{eqnarray*}
\left\{p_{3i}+p_{3i+1}+p_{3i+2}~\vert~i\in\{0,\ldots,\ell/3-1\}\right\}&=&\\
\left\{ (u_i+v_i+w_i)+(u_i+\gamma_2v_i+\gamma_1w_i)+(u_i+\gamma_1v_i+\gamma_2w_i)~\vert ~i\in\{0,\ldots,\ell/3-1\right\}&=&\\
\left\{ 3u_i+(1+\gamma_1+\gamma_2)v_i+(1+\gamma_1+\gamma_2)w_i~\vert ~i\in\{0,\ldots,\ell/3-1\right\}&=&\{3u_i\}_{i=0}^{\ell/3-1}
\end{eqnarray*}
\item For the eigenvalue $\gamma_1$, a basis of the eigenspace is
\begin{eqnarray*}
\{p_{3i}+\gamma_1p_{3i+1}+\gamma_2p_{3i+2}~\vert~i\in\{0,\ldots,\ell/3-1\}\}&=&\\
\left\{ (u_i+v_i+w_i)+\gamma_1(u_i+\gamma_2v_i+\gamma_1w_i)+\gamma_2(u_i+\gamma_1v_i+\gamma_2w_i)~\vert ~i\in\{0,\ldots,\ell/3-1\right\}&=&\\
\left\{ 3v_i+(1+\gamma_1+\gamma_2)u_i+(1+\gamma_1+\gamma_2)w_i~\vert ~i\in\{0,\ldots,\ell/3-1\right\}&=&\{3v_i\}_{i=0}^{\ell/3-1}
\end{eqnarray*}
\item For the eigenvalue $\gamma_2$, a basis of the eigenspace is
\begin{eqnarray*}
\{p_{3i}+\gamma_2p_{3i+1}+\gamma_1p_{3i+2}~\vert~i\in\{0,\ldots,\ell/3-1\}\}&=&\\
\left\{ (u_i+v_i+w_i)+\gamma_2(u_i+\gamma_2v_i+\gamma_1w_i)+\gamma_1(u_i+\gamma_1v_i+\gamma_2w_i)~\vert ~i\in\{0,\ldots,\ell/3-1\right\}&=&\\
\left\{ 3w_i+(1+\gamma_1+\gamma_2)u_i+(1+\gamma_1+\gamma_2)v_i~\vert ~i\in\{0,\ldots,\ell/3-1\right\}&=&\{3w_i\}_{i=0}^{\ell/3-1}
\end{eqnarray*}
\end{enumerate}
In addition, by Lemma~\ref{lemma:eigenspaces3} we have that $\Span{\{u_i+v_i+w_i\}_{i=0}^{\ell/3-1}}$ is an independent subspace.
\end{proof}

Similarly, we have the following claim about matrices of the form $P^{-1}A^2P$.
\begin{lemma}\label{lemma:matrixA^2}
If $\{u_i\}_{i=0}^{\ell/3-1}\cup\{v_i\}_{i=0}^{\ell/3-1}\cup\{w_i\}_{i=0}^{\ell/3-1}$ is a basis of $\bF_q^\ell$, then the matrix $P^{-1}A^2P$ (where $P$ was defined in Lemma~\ref{lemma:SpacesMatrixA}) has $\Span{\{u_i\}_{i=0}^{\ell/3-1}}$ as an eigenspace for the eigenvalue $1$, $\Span{\{w_i\}_{i=0}^{\ell/3-1}}$ as an eigenspace for the eigenvalue $\gamma_1$, and $\Span{\{v_i\}_{i=0}^{\ell/3-1}}$ as an eigenspace for the eigenvalue $\gamma_2$. In addition, the subspace $\Span{\{u_i+v_i+w_i\}_{i=0}^{\ell/3-1}}$ is an independent subspace.
\end{lemma}

\begin{proof}
According to Lemma~\ref{lemma:SpacesMatrixA}, for each $i,~0\le i\le \ell/3-1$, we have that
\begin{eqnarray*}
u_i\cdot P^{-1}A^2P&=&u_i\cdot(P^{-1}AP)\cdot(P^{-1}AP)\\
&=&u_i\cdot(P^{-1}AP)=u_i\\
w_i\cdot P^{-1}A^2P&=&w_i\cdot(P^{-1}AP)\cdot(P^{-1}AP)\\
&=&\gamma_2 w_i\cdot(P^{-1}AP)=\gamma_2^2w_i=\gamma_1w_i\\
v_i\cdot P^{-1}A^2P&=&v_i\cdot(P^{-1}AP)\cdot(P^{-1}AP)\\
&=&\gamma_1 v_i\cdot(P^{-1}AP)=\gamma_1^2v_i=\gamma_2v_i.
\end{eqnarray*}
To see that $S\triangleq\Span{\{u_i+v_i+w_i\}_{i=0}^{\ell/3-1}}$ is an independent subspace of $P^{-1}A^2P$, recall that by Lemma~\ref{lemma:SpacesMatrixA}, we have that $S$ is an independent subspace of $P^{-1}AP$, namely, $S+S(P^{-1}AP)+S(P^{-1}A^2P)=\bF_{q^\ell}$. Since the minimal polynomial of $A$ is $x^3-1$, we have that $P^{-1}A^4P=P^{-1}AP$. Hence,
\[S+S(P^{-1}A^2P)+S(P^{-1}A^4P)=S+S(P^{-1}A^2P)+S(P^{-1}AP)=\bF_q^\ell,\]
 and therefore $S$ is an independent subspace of $P^{-1}A^2P$ as well.
\end{proof}

Recall that the matching $\cZ$ consists of the edges $\left\{\{z_i,z_i',z_i''\}\right\}_{i=0}^{\ell/3-1}$. The following invertible matrices are used in the construction.
\begin{eqnarray*}
P_{Z} &\triangleq &
\begin{pmatrix}
N(z_0+z_0'+z_0'', &-z_0', &-z_0'')\\
 \vdots & \vdots &\vdots\\
N( z_{\ell/3-1}+z_{\ell/3-1}'+z_{\ell/3-1}'', &-z_{\ell/3-1}', &-z_{\ell/3-1}'')
\end{pmatrix}\\
~\\
P_{Z'} &\triangleq &
\begin{pmatrix}
N( -z_0'', &-z_0, &z_0+z_0'+z_0'')\\
 \vdots & \vdots &\vdots\\
N( -z_{\ell/3-1}'', &-z_{\ell/3-1}, &z_{\ell/3-1}+z_{\ell/3-1}'+z_{\ell/3-1}'')
\end{pmatrix}\\
~\\
P_{Z''} &\triangleq &
\begin{pmatrix}
N( -z_0', &z_0+z_0'+z_0'', &-z_0)\\
 \vdots & \vdots &\vdots\\
N( -z_{\ell/3-1}', &z_{\ell/3-1}+z_{\ell/3-1}'+z_{\ell/3-1}'', &-z_{\ell/3-1})
\end{pmatrix}\\
~\\
P_{Z^*} &\triangleq &
\begin{pmatrix}
N( z_0, &z_0'', &z_0')\\
 \vdots & \vdots &\vdots\\
N( z_{\ell/3-1}, &z_{\ell/3-1}'', &z_{\ell/3-1}')
\end{pmatrix}
\end{eqnarray*}

\begin{definition}\label{definition:oneMatching}
For a matching $\cZ=(Z,Z',Z'')$ and any distinct nonzero field elements $\lambda_{Z},\lambda_{Z'},\lambda_{Z''},$ and $\lambda_{Z^*}$, let
\begin{align*}
A_Z&\triangleq \lambda_{Z}\cdot
P_Z^{-1}A
P_Z,& S_Z&\triangleq\Span{Z}=\Span{\{z_i\}_{i=0}^{\ell/3-1}}\\
A_{Z'}&\triangleq\lambda_{Z'}\cdot
P_{Z'}^{-1}A
P_{Z'},& S_{Z'}&\triangleq\Span{Z'}=\Span{\{z_i'\}_{i=0}^{\ell/3-1}}\\
A_{Z''}&\triangleq\lambda_{Z''}\cdot
P_{Z''}^{-1}A
P_{Z''},& S_{Z''}&\triangleq\Span{Z''}=\Span{\{z_i''\}_{i=0}^{\ell/3-1}}\\
A_{Z^*}&\triangleq\lambda_{Z^*}\cdot
P_{Z^*}^{-1}A
P_{Z^*},& S_{Z^*}&\triangleq\Span{\{z_i+z_i'+z_i''\}_{i=0}^{\ell/3-1}}.
\end{align*}
\end{definition}

In the following we show that in a large enough field, there exists a choice of field elements $\lambda_{Z},\lambda_{Z'},\lambda_{Z''},\lambda_{Z^*}$ such that the $\AS$-set in  Definition~\ref{definition:oneMatching} satisfies the subspace property, and this choice can be done efficiently. Our choice will satisfy that $\lambda_{Z'}=\lambda_Z\cdot h, \lambda_{Z''}=\lambda_Z\cdot h^2$, and $\lambda_{Z^*}=\lambda_Z\cdot h^3$ for some $h\in\bF_q^*$. A suitable value of $h$ and $\lambda_Z$ will be chosen at the end of the proof of the following theorem.

\begin{theorem}\label{theorem:subspaceConditionOneMatching}
If $q$ is large enough (yet independent of $m$), then there exists a choice of values $\lambda_{Z},\lambda_{Z'},\lambda_{Z''},\lambda_{Z^*}$ such that the $\AS$-set from Definition~\ref{definition:oneMatching} satisfies the subspace condition.
\end{theorem}
\begin{proof}
By Lemma~\ref{lemma:SpacesMatrixA}, we have the following table.
\begin{center}
\begin{table}[h]
\caption{Eigenspaces of $A_Z,A_{Z'},A_{Z''},$ and $A_{Z^*}$.}\label{table:1}
\begin{tabular}{|c|c|c|c|c|}
\hline & Eigenspace for $\lambda$ & Eigenspace for $\lambda \gamma_1$ & Eigenspace for $\lambda \gamma_2$ & Independent subspace\\\hline
$A_Z, \lambda=\lambda_Z$     & $S_{Z^*}$ & $S_{Z'}$ & $S_{Z''}$   & $S_Z$\\\hline
$A_{Z'} , \lambda=\lambda_{Z'}$  & $S_{Z''}$ & $S_{Z}$  & $S_{Z^*}$   & $S_{Z'}$\\\hline
$A_{Z''}, \lambda=\lambda_{Z''}$ & $S_{Z'}$  & $S_{Z^*}$    & $S_{Z}$ & $S_{Z''}$\\\hline
$A_{Z^*}, \lambda=\lambda_{Z^*}$ & $S_{Z}$   & $S_{Z''}$  & $S_{Z'}$  & $S_{Z^*}$\\\hline
\end{tabular}
\end{table}
\end{center}
and hence, the independence and the invariance properties hold.

To show the nonsingular property, assume for now that $h$ is chosen such that every distinct $\lambda_1,\lambda_2$ in $\{\lambda_Z,\lambda_{Z'},\lambda_{Z''},\lambda_{Z^*}\}=\{\lambda_Z,\lambda_{Z}\cdot h,\lambda_{Z}\cdot h^2,\lambda_{Z}\cdot h^3\}$ satisfy $\lambda_1^6\ne \lambda_2^6$. Notice that this requirement implies that $h^6,h^{12},h^{18}\ne 1$. A specific choice of $h$ which satisfies this condition, as well as additional conditions that will emerge in the sequel, will be shown at the end of this proof.

We first show that the difference between any two matrices is of full rank. We show that $A_Z-A_{Z'}$ is of full rank, and the rest of the cases, which are similar, are given in~\nameref{appendix:rankDistance}. Notice that
\begin{eqnarray*}
\nonumber z_i''(A_Z-A_{Z'})&=&(\lambda_Z\gamma_2-\lambda_{Z'})z_i''\\
\nonumber (z_i+z_i'+z_i'')(A_Z-A_{Z'})&=&(\lambda_Z-\lambda_{Z'}\gamma_2)(z_i+z_i'+z_i'')\\
\nonumber z_i(A_Z-A_{Z'})&=&\lambda_Z\left(z_i+(1-\gamma_1)z_i'+(1-\gamma_2)z_i''\right)-\lambda_{Z'}\gamma_1z_i\\
 &=&(\lambda_Z-\lambda_{Z'}\gamma_1)z_i+\lambda_Z(1-\gamma_1)z_i'+\lambda_Z(1-\gamma_2)z_i'',
\end{eqnarray*}
which is equivalent to
\begin{align*}
\begin{pmatrix}
z_i''\\z_i+z_i'+z_i''\\z_i
\end{pmatrix}\cdot (A_Z-A_{Z'})&=\lambda_Z\cdot
\begin{pmatrix}
0 & 0&\gamma_2-h  \\
1-h\gamma_2 & 1-h\gamma_2&1-h\gamma_2\\
1-h\gamma_1 & 1-\gamma_1&1-\gamma_2
\end{pmatrix}\cdot \begin{pmatrix}
z_i\\z_i'\\z_i''
\end{pmatrix},
\end{align*}
and since 
\begin{align*}
\begin{pmatrix}
z_i''\\z_i+z_i'+z_i''\\z_i
\end{pmatrix}=\begin{pmatrix}
0&0&1\\1&1&1\\1&0&0
\end{pmatrix}\cdot\begin{pmatrix}
z_i\\z_i'\\z_i''
\end{pmatrix}
\end{align*}
we may write

\begin{align}\label{equation:AZAZ'matrix}
\begin{pmatrix}
z_i\\z_i'\\z_i''
\end{pmatrix}\cdot (A_Z-A_{Z'})&\triangleq\lambda_Z\cdot\begin{pmatrix}
0&0&1\\1&1&1\\1&0&0
\end{pmatrix}^{-1}
\cdot\hat{\Phi}(h)\cdot\begin{pmatrix}
z_i\\z_i'\\z_i''
\end{pmatrix}.
\end{align}

 To show that $z_i,z_i',z_i''\in\image(A_Z-A_{Z'})$ it suffices to show that $\det\hat{\Phi}(h)\ne 0$. Since $\gamma_2\ne h$ (otherwise, $h^6=1$), it follows that $\det\hat{\Phi}(h)$ is nonzero if and only if 
 \[
(1-h\gamma_2)\cdot (1-\gamma_1)\ne(1-h\gamma_2)\cdot (1-h\gamma_1),
 \]
which also follows easily from the fact that $h^6\ne 1$.
%

To show that the difference between any two squares of matrices is of full rank, notice that by Lemma~\ref{lemma:matrixA^2}, we have the following table.
\begin{center}
\begin{table}[h]
\caption{Eigenspaces of $A_Z^2,A_{Z'}^2,A_{Z''}^2,$ and $A_{Z^*}^2$.}\label{table:2}
\begin{tabular}{|c|c|c|c|c|}
\hline & Eigenspace for $\lambda$ & Eigenspace for $\lambda \gamma_1$ & Eigenspace for $\lambda \gamma_2$ & Independent subspace\\\hline
$A_Z^2, \lambda=\lambda_Z^2$     & $S_{Z^*}$ & $S_{Z''}$ & $S_{Z'}$   & $S_Z$\\\hline
$A_{Z'}^2 , \lambda=\lambda_{Z'}^2$  & $S_{Z''}$ & $S_{Z^*}$  & $S_{Z}$   & $S_{Z'}$\\\hline
$A_{Z''}^2, \lambda=\lambda_{Z''}^2$ & $S_{Z'}$  & $S_{Z}$    & $S_{Z^*}$ & $S_{Z''}$\\\hline
$A_{Z^*}^2, \lambda=\lambda_{Z^*}^2$ & $S_{Z}$   & $S_{Z'}$  & $S_{Z''}$  & $S_{Z^*}$\\\hline
\end{tabular}
\end{table}
\end{center}

We show that $A_Z^2-A_{Z'}^2$ is of full rank. The rest of the cases are similar, and are given in~\nameref{appendix:squareRankDistance}. Notice that
\begin{eqnarray*}
\nonumber z_i''(A_Z^2-A_{Z'}^2)&=&(\lambda_Z^2 \gamma_1-\lambda_{Z'}^2)z_i''\\
\nonumber (z_i+z_i'+z_i'')(A_Z^2-A_{Z'}^2)&=&(\lambda_Z^2-\lambda_{Z'}^2\gamma_1)(z_i+z_i'+z_i'')\\
\nonumber z_i(A_Z^2-A_{Z'}^2)&=&\lambda_Z^2(z_i+(1-\gamma_2)z_i'+(1-\gamma_1)z_i'')-\lambda_{Z'}^2\gamma_2 z_i\\
\label{equation:SquareRankDistance1}&=&\lambda_Z^2(z_i+z_i'+z_i'')-\lambda_Z^2 \gamma_2 z_i'-\lambda_Z^2 \gamma_1 z_i''-\lambda_{Z'}^2\gamma_2z_i,
\end{eqnarray*}
which may be written as

\begin{equation*}
\begin{pmatrix}
z_i''\\z_i+z_i'+z_i''\\z_i
\end{pmatrix}\cdot (A_Z^2-A_{Z'}^2)=\lambda_Z^2
\begin{pmatrix}
0 & 0&\gamma_1-h^2 \\
1-h^2\gamma_1 & 1-h^2\gamma_1&1-h^2\gamma_1\\
1-h^2\gamma_2 & 1-\gamma_2&1-\gamma_1
\end{pmatrix}\cdot \begin{pmatrix}
z_i\\z_i'\\z_i''
\end{pmatrix}\triangleq \lambda_Z^2\cdot\hat{\Psi}(h)\cdot\begin{pmatrix}
z_i\\z_i'\\z_i''
\end{pmatrix}.
\end{equation*}
To show that $z_i,z_i',z_i''\in\image(A_Z^2-A_{Z'}^2)$ it suffices to show that $\det \hat{\Psi}(h)\ne 0$. Since $\gamma_1\ne h^2$ (otherwise, $h^6=1$) it follows that $\det\hat{\Psi}(h)$ is nonzero if and only if
\[
(1-h^2\gamma_1)\cdot(1-h^2\gamma_2)\ne (1-h^2\gamma_1)\cdot(1-\gamma_2),
\]
which also follows easily from the fact that $h^6\ne 1$.

To show that Condition~\ref{condition:3x3} of the nonsingular property (i.e. that any $3\times 3$ block submatrix of the non systematic part of the generator matrix is invertible, as mentioned at the beginning of Section~\ref{section:construction3}), we must show that the following matrices are invertible
\begin{eqnarray*}\label{equation:rowEquivalence3x3r+1}
V(Z,Z',Z'')\triangleq\begin{pmatrix}
I & I & I \\
A_Z & A_{Z'} & A_{Z''} \\
A_Z^2 & A_{Z'}^2 & A_{Z''}^2
\end{pmatrix}
,&
V(Z,Z',Z^*)\triangleq\begin{pmatrix}
I & I & I \\
A_Z & A_{Z'} & A_{Z^*} \\
A_Z^2 & A_{Z'}^2 & A_{Z^*}^2
\end{pmatrix}\\
V(Z,Z'',Z^*)\triangleq\begin{pmatrix}
I & I & I \\
A_Z & A_{Z''} & A_{Z^*} \\
A_Z^2 & A_{Z''}^2 & A_{Z^*}^2
\end{pmatrix}
,&
V(Z',Z'',Z^*)\triangleq\begin{pmatrix}
I & I & I \\
A_{Z'} & A_{Z''} & A_{Z^*} \\
A_{Z'}^2 & A_{Z''}^2 & A_{Z^*}^2
\end{pmatrix}.
\end{eqnarray*}

Using elementary block row operations, we have that $V(Z,Z',Z'')$ is invertible if and only if
\begin{eqnarray*}
\begin{pmatrix}
I & I & I \\
0 & I & (A_{Z'}-A_Z)^{-1}(A_{Z''}-A_Z) \\
0 & 0 & (A_{Z}^2-A_{Z''}^2)-(A_{Z}^2-A_{Z'}^2)(A_{Z}-A_{Z'})^{-1}(A_{Z}-A_{Z''})
\end{pmatrix}
\end{eqnarray*}
is invertible. Clearly, this matrix is invertible if and only if \begin{eqnarray*}
L_1\triangleq(A_{Z}^2-A_{Z''}^2)(A_{Z}-A_{Z''})^{-1}-(A_{Z}^2-A_{Z'}^2)(A_{Z}-A_{Z'})^{-1}
\end{eqnarray*}
is invertible. Similarly, $V(Z,Z',Z^*)$, $V(Z,Z'',Z^*)$, and $V(Z',Z'',Z^*)$ are invertible if and only~if
\begin{align*}
L_2&\triangleq(A_{Z}^2-A_{Z^*}^2)(A_{Z}-A_{Z^*})^{-1}-(A_{Z}^2-A_{Z'}^2)(A_{Z}-A_{Z'})^{-1}\\
L_3&\triangleq(A_{Z}^2-A_{Z^*}^2)(A_{Z}-A_{Z^*})^{-1}-(A_{Z}^2-A_{Z''}^2)(A_{Z}-A_{Z''})^{-1}\\
L_4&\triangleq(A_{Z'}^2-A_{Z^*}^2)(A_{Z'}-A_{Z^*})^{-1}-(A_{Z'}^2-A_{Z''}^2)(A_{Z'}-A_{Z''})^{-1}
\end{align*}
are invertible. To show that $L_1,L_2,L_3$, and $L_4$ are invertible, we show that the image of each of them contains $z_i,z_i',z_i''$ for all $i\in\{0,\ldots,\ell/3-1\}$. Notice that by Table~\ref{table:2},
\begin{eqnarray}
\nonumber (z_i+z_i'+z_i'')L_1&=&(z_i+z_i'+z_i'')(A_{Z}^2-A_{Z''}^2)(A_{Z}-A_{Z''})^{-1}-\\\nonumber &~&(z_i+z_i'+z_i'')(A_{Z}^2-A_{Z'}^2)(A_{Z}-A_{Z'})^{-1}\\\nonumber
&=&(\lambda_{Z}^2-\lambda_{Z''}^2\gamma_2)(z_i+z_i'+z_i'')(A_{Z}-A_{Z''})^{-1}- \\\nonumber
&~& (\lambda_Z^2-\lambda_{Z'}^2\gamma_1)(z_i+z_i'+z_i'')(A_{Z}-A_{Z'})^{-1}\\\label{equation:L1set}
z_i L_1&=& z_i(A_{Z}^2-A_{Z''}^2)(A_{Z}-A_{Z''})^{-1}-z_i(A_{Z}^2-A_{Z'}^2)(A_{Z}-A_{Z'})^{-1}\\\nonumber
&=&\left(\lambda_Z^2(z_i+z_i'+z_i'')-\lambda_Z^2 \gamma_2 z_i'-\lambda_Z^2 \gamma_1 z_i''-\lambda_{Z''}^2 \gamma_1 z_i\right)(A_{Z}-A_{Z''})^{-1}-\\\nonumber
&~&\left(\lambda_Z^2(z_i+z_i'+z_i'')-\lambda_Z^2 \gamma_2 z_i'-\lambda_Z^2 \gamma_1 z_i''-\lambda_{Z'}^2\gamma_2z_i\right)(A_{Z}-A_{Z'})^{-1}\\\nonumber
z_i'L_1&=&  z_i'(A_{Z}^2-A_{Z''}^2)(A_{Z}-A_{Z''})^{-1}-z_i'(A_{Z}^2-A_{Z'}^2)(A_{Z}-A_{Z'})^{-1}=\\\nonumber
&=&(\lambda_{Z}^2\gamma_2-\lambda_{Z''}^2)z_i'(A_{Z}-A_{Z''})^{-1}-\\\nonumber
&~&\left(\lambda_{Z'}^2(\gamma_2-\gamma_1)z_i+(\lambda_Z^2 \gamma_2-\lambda_{Z'}^2 \gamma_1)z_i'+\lambda_{Z'}^2(1-\gamma_1)z_i''\right)(A_{Z}-A_{Z'})^{-1},
\end{eqnarray}
which can be written as
\begin{align}\label{equation:L1Matrix}
\nonumber\begin{pmatrix}
z_i+z_i'+z_i''\\z_i\\z_i'
\end{pmatrix}L_1&=\lambda_Z^2
\begin{pmatrix}
1-h^4\gamma_2&1-h^4\gamma_2&1-h^4\gamma_2\\
1-h^4\gamma_1&1-\gamma_2&1-\gamma_1\\
0 & \gamma_2-h^4&0
\end{pmatrix}\cdot \begin{pmatrix}
z_i\\z_i'\\z_i''
\end{pmatrix}(A_Z-A_{Z''})^{-1}\\
&+\lambda_Z^2\begin{pmatrix}
1-h^2\gamma_1&1-h^2\gamma_2&1-h^2\gamma_1\\
1-h^2\gamma_2&1-\gamma_2&1-h^2\gamma_2\\
h^2(\gamma_2-\gamma_1)&\gamma_2-h^2\gamma_1&h^2(1-\gamma_1)
\end{pmatrix}\cdot \begin{pmatrix}
z_i\\z_i'\\z_i''
\end{pmatrix}(A_Z-A_{Z'})^{-1}\\
&\triangleq\lambda_Z^2\left(U_{1,1}(h)\cdot \begin{pmatrix}
z_i\\z_i'\\z_i''
\end{pmatrix}(A_Z-A_{Z''})^{-1}+U_{1,2}(h)\cdot \begin{pmatrix}
z_i\\z_i'\\z_i''
\end{pmatrix}(A_Z-A_{Z'})^{-1}\right)\nonumber.
\end{align}

Notice that the values of
\begin{align*}\label{equation:inverseValues1}
\begin{pmatrix}
z_i\\z_i'\\z_i''
\end{pmatrix}(A_Z-A_{Z'})^{-1},~\begin{pmatrix}
z_i\\z_i'\\z_i''
\end{pmatrix}(A_Z-A_{Z''})^{-1}
\end{align*}
may easily be computed from~\eqref{equation:AZAZ'matrix} and the equivalent equation for $(A_Z-A_{Z''})$. Inspecting~\eqref{equation:AZAZ'matrix}, we observe that 

\begin{align*}
\lambda_Z^{-1}\cdot\hat{\Phi}(h)^{-1}\cdot\begin{pmatrix}
	0&0&1\\1&1&1\\1&0&0
	\end{pmatrix}\cdot\begin{pmatrix}
	z_i\\z_i'\\z_i''	\end{pmatrix}&=
	\begin{pmatrix}
	z_i\\z_i'\\z_i''
	\end{pmatrix}\cdot(A_Z-A_{Z'})^{-1},
\end{align*}
and hence we may write
\begin{align*}
\begin{pmatrix}
z_i\\z_i'\\z_i''
\end{pmatrix}(A_Z-A_{Z'})^{-1}&=\lambda_Z^{-1} C(h)^{-1}\cdot \begin{pmatrix}
z_i\\z_i'\\z_i''
\end{pmatrix}.
\end{align*}
for some matrix $C(h)$ which depends only on $h$. Similarly,
\begin{align*}
\begin{pmatrix}
z_i\\z_i'\\z_i''
\end{pmatrix}(A_Z-A_{Z''})^{-1}&=\lambda_Z^{-1} D(h)^{-1}\cdot \begin{pmatrix}
z_i\\z_i'\\z_i''
\end{pmatrix} 
\end{align*}
for some matrix $D(h)$ which depend only on $h$. Therefore,~\eqref{equation:L1Matrix} can be rewritten as
\begin{align*}
\begin{pmatrix}
z_i+z_i'+z_i''\\z_i\\z_i'
\end{pmatrix}L_1=\lambda_Z \left(U_{1,1}(h)D(h)^{-1}+U_{1,2}(h)C(h)^{-1}\right)\cdot \begin{pmatrix}
z_i\\z_i'\\z_i''
\end{pmatrix} \triangleq \lambda_Z \Upsilon_1(h)\cdot \begin{pmatrix}
z_i\\z_i'\\z_i''
\end{pmatrix},
\end{align*}
where $\Upsilon_1(h)$ is a $3\times 3$ matrix whose entries are functions of $h$. Note that since any entry in $U_{1,1}(h)$, $U_{1,2}(h)$, $C(h)$ and $D(h)$ is a polynomial of degree at most 6 in the variable $h$, it follows that
any entry in $\Upsilon_1(h)$ is a \textit{rational function} (that is, a division of polynomials) in $h$, in which the degree of the enumerator and denominator polynomials is some small constant\footnote{This may be easily seen as a result of the well-known formula $M^{-1}=\text{adj}{M}/\det M$, where $\text{adj}M$ is the adjoint (or adjucate) matrix of $M$.}.
Similarly, it can be shown that there exist $\Upsilon_2(h),\Upsilon_3(h),\Upsilon_4(h)$, which correspond to $L_2,L_3,L_4$, whose entries are divisions of polynomials in $h$ of small constant degree. To show that Condition~\ref{condition:3x3} is satisfied, it suffices to show that $\det \Upsilon_i(h)\ne 0$ for all $i,~1\le i\le 4$. A proper field constant $h$, for which non of these determinants vanish and $h^6,h^{12},h^{18}\ne 1$, can be found by denoting
\begin{align*}
\det\Upsilon_i(h)=\frac{Q_{i,1}(h)}{Q_{i,2}(h)},
\end{align*}
and considering the polynomial
\begin{equation*}
Q(h)\triangleq(h^6-1)(h^{12}-1)(h^{18}-1)\prod_{i=1}^{4}\left(Q_{i,1}(h)\cdot Q_{i,2}(h)\right).
\end{equation*}
Clearly, if $q> \deg Q+1$, a nonzero $h$ such that $Q(h)\ne 0$ can be found by polynomial factorization, and since $\deg Q$ is constant, the proof is complete.
\end{proof}
This theorem showed that a single matching provides an $\AS$-set of size four, satisfying the subspace property, over a field of constant size. In the next subsection it will be shown that by taking $q$ to be at least \textit{linear} in $m$, $\AS$-sets from different matchings, that satisfy the pairing condition in pairs, may be united without compromising on the subspace property.

\subsection{A Code from Two Matchings}\label{section:twoMatchings3}
In this subsection it is shown that $\AS$-sets that were constructed from different matchings may be united, as long as the pairing condition holds. 
In the remaining part of this subsection, let $\cX=(X,X',X''),\cY=(Y,Y',Y'')$ be two matchings which satisfy the pairing condition, and let the resulting $\AS$-sets be as in Definition~\ref{definition:oneMatching}:
\begin{eqnarray*}
C_\cX&\triangleq&\{(A_X,S_X),(A_{X'},S_{X'}),(A_{X''},S_{X''}),(A_{X^*},S_{X^*})\}\\
C_\cY&\triangleq&\{(A_Y,S_Y),(A_{Y'},S_{Y'}),(A_{Y''},S_{Y''}).(A_{Y^*},S_{Y^*})\}.
\end{eqnarray*}
The required values of $\lambda_X,\lambda_Y$ which are involved in the definition of these $\AS$-sets will be discussed in the sequel.

\begin{lemma}\label{lemma:simDiagMatrices}
If $C\in\{A_X,A_{X'},A_{X''},A_{X^*}\}$ and $D\in\{A_Y,A_{Y'},A_{Y''},A_{Y^*}\}$, then $C$ and $D$ are simultaneously diagonalizable. Furthermore, $C^2$ and $D^2$ are simultaneously diagonalizable.
\end{lemma}

\iffull
\begin{proof}
Follow the exact outline of the proof of Lemma~\ref{lemma:simDig}.
%
\end{proof}
\fi

Recall that the definition of an $\AS$-set from a single matching involved the choice of two field constants $\lambda_Z$ and $h$. In what follows we use the same $h$ for all matchings, and choose proper distinct values for the constants which correspond to $\lambda_Z$. The next lemma is required for the construction.

\begin{lemma}\label{lemma:lambdaXY}
If $\lambda_X$ and $\lambda_Y$ are two nonzero field elements such that
$\lambda_Y^6\notin\{\lambda_X^6,\lambda_X^6h^{\pm 6},\lambda_X^6h^{\pm 12},\lambda_X^6h^{\pm 18}\}$ and
\begin{alignat*}{5}
\lambda_1 &\in&\{\lambda_X,\lambda_{X'},\lambda_{X''},\lambda_{X^*}\}&=&\{\lambda_X,\lambda_{X}h,\lambda_{X}h^2,\lambda_Xh^3\},\\
\lambda_2 &\in&\{\lambda_Y,\lambda_{Y'},\lambda_{Y''},\lambda_{Y^*}\}&=&\{\lambda_Y,\lambda_{Y}h,\lambda_{Y}h^2,\lambda_Yh^3\},
\end{alignat*}
then
\begin{alignat*}{5}
\{\lambda_1,\lambda_1\gamma_1,\lambda_1\gamma_2\}&\cap&\{\lambda_2,\lambda_2\gamma_1,\lambda_2\gamma_2\}&=&\varnothing,\\
\{\lambda_1^2,\lambda_1^2\gamma_2,\lambda_1^2\gamma_1\}&\cap&\{\lambda_2^2,\lambda_2^2\gamma_2,\lambda_2^2\gamma_1\}&=&\varnothing.
\end{alignat*}
\end{lemma}
\iffull
\begin{proof}
We first show that $\lambda_1^6\ne \lambda_2^6$. Assume for contradiction that $\lambda_1^6= \lambda_2^6$. By the definition of $\lambda_1$ and $\lambda_2$, there exists $i,j\in\{0,1,2,3\}$ such that $\lambda_1=\lambda_X\cdot h^i$ and $\lambda_2=\lambda_Y^2\cdot h^{j}$. Since  $\lambda_1^6= \lambda_2^6$, it follows that $\lambda_X^6\cdot h^{6i}=\lambda_Y^6\cdot h^{6j}$, and hence $\lambda_Y^6=\lambda_X^6\cdot h^{6(i-j)}$. Since $i,j\in\{0,1,2,3\}$ it follows that $6(i-j)\in\{0,\pm 6, \pm 12, \pm 18\}$, and therefore $\lambda_Y^6\in\{\lambda_X^6,\lambda_X^6h^{\pm 6},\lambda_X^6h^{\pm 12},\lambda_X^6h^{\pm 18}\}$, a contradiction.

Now, if $\{\lambda_1,\lambda_1\gamma_1,\lambda_1\gamma_2\}\cap\{\lambda_2,\lambda_2\gamma_1,\lambda_2\gamma_2\}\ne\varnothing$, then there exists $i,j\in\{0,1,2\}$ such that $\lambda_1 \gamma_1^i=\lambda_2 \gamma_1^j$, and hence, $\lambda_1^6=\lambda_2 ^6$, a contradiction. Similarly, if $\{\lambda_1^2,\lambda_1^2\gamma_2,\lambda_1^2\gamma_1\}\cap\{\lambda_2^2,\lambda_2^2\gamma_1,\lambda_2^2\gamma_1\}\ne\varnothing$, there exists $i,j\in\{0,1,2\}$ such that $\lambda_1^2 \gamma_1^i=\lambda_2^2 \gamma_1^j$, and hence, $\lambda_1^6=\lambda_2 ^6$ as well, a contradiction.
\end{proof}
\fi

Lemma~\ref{lemma:simDiagMatrices} and Lemma~\ref{lemma:lambdaXY} enable an easy choice of field elements, which induce distinct eigenvalues for the simultaneously diagonalizable matrices that correspond to distinct matchings. These distinct eigenvalues, together with the simultaneous diagonalizable matrices,  will assist the proof of the following lemma, from which the construction will follow.

\begin{lemma}
If the field constants $\lambda_X,\lambda_Y$ satisfy
\[
\lambda_Y^6\notin\{\lambda_X^6,\lambda_X^6h^{\pm 6},\lambda_X^6h^{\pm 12},\lambda_X^6h^{\pm 18}\},
\]
then $C_\cX\cup C_\cY$ satisfies the subspace condition.
\end{lemma}

\begin{proof}
Since each of $C_\cX$ and $C_\cY$ satisfies the subspace condition separately, we are left to show the parts of the nonsingular property and the invariance property which involve matrices and subspaces from different matchings.

To prove the invariance property, for any $C\in\{A_X,A_{X'},A_{X''},A_{X^*}\}$ and any $T\in\{S_Y,S_{Y'},S_{Y''},S_{Y^*}\}$ we must show that $TC=T$. Let $S_1,S_2,S_3\in\{S_X,S_{X'},S_{X''},S_{X^*}\}$ be the eigenspaces of $C$. It follows from Corollary~\ref{corollary:eigenspaceIntesections} that $\dim (S_i\cap T)=\ell/9$ for all $i\in[3]$. Therefore, since $S_1+S_2+S_3=\bF_q^\ell$, it follows that there exists a basis $t_1,\ldots,t_{\ell/3}$ of $T$ in which all vectors are eigenvectors of $C$. Hence, for all $i\in[\ell/3]$ we have that $t_iC\in T$, and thus $TC=T$. The inverse case, where $C\in\{A_Y,A_{Y'},A_{Y''},A_{Y^*}\}$ and $T\in \{S_X,S_{X'},S_{X''},S_{X^*}\}$, is symmetric.

To prove the nonsingular property, let $C\in\{A_X,A_{X'},A_{X''},A_{X^*}\}$ and  $D\in\{A_Y,A_{Y'},A_{Y''},A_{Y^*}\}$. According to Definition~\ref{definition:oneMatching}, the eigenvalues of $C$ are $\lambda_C,\lambda_C \gamma_1,$ and $\lambda_C\gamma_2$ for some \[\lambda_C\in\{\lambda_X,\lambda_Xh,\lambda_Xh^2,\lambda_Xh^3\},\]and similarly, the eigenvalues of $D$ are $\lambda_D,\lambda_D \gamma_1,$ and $\lambda_D\gamma_2$ for some \[\lambda_D\in\{\lambda_Y,\lambda_Yh,\lambda_Yh^2,\lambda_Yh^3\}.\]By Lemma~\ref{lemma:lambdaXY}, and since $\lambda_Y^6\notin\{\lambda_X^6,\lambda_X^6h^{\pm 6},\lambda_X^6h^{\pm 12},\lambda_X^6h^{\pm 18}\}$, it follows that
\begin{alignat*}{5}
\{\lambda_C,\lambda_C\gamma_1,\lambda_C\gamma_2\}&\cap&\{\lambda_D,\lambda_D\gamma_1,\lambda_D\gamma_2\}&=&\varnothing\\
\{\lambda_C^2,\lambda_C^2\gamma_2,\lambda_C^2\gamma_1\}&\cap&\{\lambda_D^2,\lambda_D^2\gamma_1,\lambda_D^2\gamma_1\}&=&\varnothing&,
\end{alignat*}
that is, $C$ and $D$ have no eigenvalue in common, and $C^2$ and $D^2$ have no eigenvalue in common. Since Lemma~\ref{lemma:simDiagMatrices} implies that $C$ and $D$ are simultaneously diagonalizable, and so are $C^2$ and $D^2$, it follows by Lemma~\ref{lemma:simDiagAreInv} that $C-D$ and $C^2-D^2$ are invertible.

We are left to prove that any $3\times 3$ block submatrix is invertible. Let $A_i,A_j,A_k$ be three matrices from $C_\cX\cup C_\cY$ such that $A_i$ and $A_j$ are not from the same matching, and so are $A_i$ and $A_k$.  Recall that, as in the proof of Theorem~\ref{theorem:subspaceConditionOneMatching}, the matrix
\begin{eqnarray*}
\begin{pmatrix}
I & I & I \\
A_i & A_{j} & A_{k} \\
A_i^2 & A_{j}^2 & A_{k}^2
\end{pmatrix}
\end{eqnarray*}
is invertible if and only if
\begin{eqnarray*}
L\triangleq(A_{i}^2-A_{k}^2)(A_{i}-A_{k})^{-1}-(A_{i}^2-A_{j}^2)(A_{i}-A_{j})^{-1}
\end{eqnarray*}
is invertible. Notice that by Lemma~\ref{lemma:simDiagMatrices} $A_i$ and $A_j$ are simultaneously diagonalizable, and hence they commute. In addition, so are $A_i$ and $A_k$. Therefore,
\begin{eqnarray*}
L&=&(A_{i}^2-A_{k}^2)(A_{i}-A_{k})^{-1}-(A_{i}^2-A_{j}^2)(A_{i}-A_{j})^{-1}\\
&=&A_i+A_k-A_i+A_j=A_k-A_j,
\end{eqnarray*}
and hence $L$ is invertible.
\end{proof}

We conclude in the following theorem, in which $Q$ is the polynomial which was mentioned in the proof of Theorem~\ref{theorem:subspaceConditionOneMatching}.

\begin{theorem}\label{theorem:theorem:3ParitiesNotAccessOptimal}
If $q>\max\{42m,\deg Q\}+1$, and $\{\cX_i=(X_i,X_i',X_i'')\}_{i=1}^m$ is the set of matchings from Definition~\ref{definition:matchings3}, then the $\AS$-set $\bC_2\triangleq\cup_{i=1}^m C_{\cX_i}$ satisfies the subspace condition.
\end{theorem}

\iffull
\begin{proof}
Since $q>\deg Q+1$, it follows that there exists $h\in\bF_q^*$ such that $Q(h)\ne 0$. Therefore, by Theorem~\ref{theorem:subspaceConditionOneMatching}, we have that the $\AS$-sets $C_{\cX_i}$ satisfy the subspace condition separately. Since $q>42m+1$, the field elements $\{\lambda_{X_i}\}_{i=1}^m$ may be chosen such that every $\lambda_{X_i},\lambda_{X_j}$ satisfy $\lambda_{X_j}^6\notin\{\lambda_{X_i}^6,\lambda_{X_i}^6h^{\pm 6},\lambda_{X_i}^6h^{\pm 12},\lambda_{X_i}^6h^{\pm 18}\}$, since the choice of any $\lambda_{X_i}$ excludes the choice of at most 42 other field elements, which constitute the roots of 7 polynomials of degree 6. Notice that a proper set $\{\lambda_{X_i}\}_{i=1}^m$ may be found explicitly using a simple iterative algorithm that maintains a feasible set of elements - in each iteration it arbitrarily chooses the next element $\lambda_{X_i}$ from it, and removes all elements $e$ which satisfy $e^6\in\{\lambda_{X_i}^6,\lambda_{X_i}^6h^{\pm 6},\lambda_{X_i}^6h^{\pm 12},\lambda_{X_i}^6h^{\pm 18}\}$.
\end{proof}
\fi

%

\section*{Acknowledgments}
The authors would like to thank Prof.~Joachim Rosenthal for many helpful discussions, which contributed to this paper.
\subsection*{Appendix A}\label{appendix:rankDistance}
To show that $A_Z-A_{Z^*}$ is of full rank, notice that
\begin{eqnarray}
\nonumber z_i' (A_Z-A_{Z^*})&=&(\lambda_{Z}\gamma_1-\lambda_{Z^*}\gamma_2)z_i'\\
\nonumber z_i''(A_Z-A_{Z^*})&=&(\lambda_{Z}\gamma_2-\lambda_{Z^*} \gamma_1)z_i''\\
\nonumber z_i(A_Z-A_{Z^*})&=&\lambda_Z\left(z_i+(1-\gamma_1)z_i'+(1-\gamma_2)z_i''\right)-\lambda_{Z^*} z_i\\
\label{equation:RankDistance2}&=&(\lambda_Z -\lambda_{Z^*})z_i+\lambda_Z(1-\gamma_1)z_i'+\lambda_Z(1-\gamma_2)z_i''
\end{eqnarray}
Since $\lambda_{Z}\gamma_1\ne \lambda_{Z^*} \gamma_2$ and $\lambda_{Z} \gamma_2 \ne \lambda_{Z^*} \gamma_1$, we have that $z_i',z_i''\in\image(A_Z-A_{Z^*})$. Therefore, it follows from~(\ref{equation:RankDistance2}) that $z_i\in\image(A_Z-A_{Z^*})$, since $\lambda_Z\ne \lambda_{Z^*}$.

To show that $A_Z-A_{Z''}$ is of full rank, notice that
\begin{eqnarray}
\nonumber z_i' (A_Z-A_{Z''})&=&(\lambda_{Z}\gamma_1-\lambda_{Z''})z_i'\\
\nonumber (z_i+z_i'+z_i'')(A_Z-A_{Z''})&=&(\lambda_{Z}-\lambda_{Z''}\gamma_1)(z_i+z_i'+z_i'')\\
\nonumber z_i(A_Z-A_{Z''})&=&\lambda_Z\left(z_i+(1-\gamma_1)z_i'+(1-\gamma_2)z_i''\right)-\lambda_{Z''}\gamma_2 z_i\\
\label{equation:RankDistance3}&=&\lambda_Z(z_i+z_i'+z_i'')-\lambda_Z \gamma_1 z_i'-\lambda_Z \gamma_2 z_i''-\lambda_{Z''} \gamma_2 z_i
\end{eqnarray}
Since $\lambda_{Z}\gamma_1\ne \lambda_{Z''}$ and $\lambda_{Z} \ne \lambda_{Z''} \gamma_1$, we have that $z_i',z_i+z_i'+z_i''\in\image(A_Z-A_{Z''})$. Therefore, it follows from~(\ref{equation:RankDistance3}) that $\lambda_Z \gamma_2 z_i''+\lambda_{Z''}\gamma_2 z_i\in\image(A_Z-A_{Z^*})$. To show that $z_i,z_i',z_i''\in\image(A_Z-A_{Z''})$ it suffices to show that $\lambda_Z\gamma_2\ne \lambda_{Z''}\gamma_2$, which also follows from $\lambda_Z^6\ne \lambda_{Z''}^6$.

To show that $A_{Z'}-A_{Z''}$ is of full rank notice that
\begin{eqnarray}
\nonumber z_i (A_{Z'}-A_{Z''})&=&(\lambda_{Z'}\gamma_1-\lambda_{Z''}\gamma_2)z_i\\
\nonumber (z_i+z_i'+z_i'')(A_{Z'}-A_{Z''})&=&(\lambda_{Z'}\gamma_2-\lambda_{Z''}\gamma_1)(z_i+z_i'+z_i'')\\
\nonumber z_i'(A_{Z'}-A_{Z''})&=&\lambda_{Z'}\left( (\gamma_2-\gamma_1)z_i+\gamma_2z_i'+(\gamma_2-1)z_i''\right)-\lambda_{Z''}z_i'\\
\label{equation:RankDistance4}&=&\lambda_{Z'}\gamma_2\left(z_i+z_i'+z_i''\right)-\lambda_{Z'}\gamma_1z_i-\lambda_{Z'}z_i''-\lambda_{Z''}z_i'
\end{eqnarray}
Since $\lambda_{Z'}\gamma_1\ne \lambda_{Z''}\gamma_2$ and $\lambda_{Z'}\gamma_2 \ne \lambda_{Z''} \gamma_1$, we have that $z_i,z_i+z_i'+z_i''\in\image(A_Z-A_{Z''})$. Therefore, it follows from~(\ref{equation:RankDistance4}) that $\lambda_{Z'} z_i''+\lambda_{Z''}z_i'\in\image(A_Z-A_{Z^*})$. To show that $z_i,z_i',z_i''\in\image(A_Z-A_{Z''})$ it suffices to show that $\lambda_{Z'}\ne \lambda_{Z''}$, which also follows from $\lambda_{Z'}^6\ne \lambda_{Z''}^6$.

To show that $A_{Z'}-A_{Z^*}$ is of full rank notice that
\begin{eqnarray}
\nonumber z_i (A_{Z'}-A_{Z^*})&=&(\lambda_{Z'}\gamma_1-\lambda_{Z^*})z_i\\
\nonumber z_i''(A_{Z'}-A_{Z^*})&=&(\lambda_{Z'}-\lambda_{Z^*}\gamma_1)z_i''\\
\nonumber z_i'(A_{Z'}-A_{Z^*})&=&\lambda_{Z'}\left( (\gamma_2-\gamma_1)z_i+\gamma_2z_i'+(\gamma_2-1)z_i''\right)-\lambda_{Z^*}\gamma_2z_i'\\
\label{equation:RankDistance5}&=&\lambda_{Z'}(\gamma_2-\gamma_1)z_i+(\lambda_{Z'}\gamma_2-\lambda_{Z^*}\gamma_2)z_i'+\lambda_{Z'}(\gamma_2-1)z_i''
\end{eqnarray}
Since $\lambda_{Z'}\gamma_1\ne \lambda_{Z^*}$ and $\lambda_{Z'} \ne \lambda_{Z^*} \gamma_1$, we have that $z_i,z_i''\in\image(A_{Z'}-A_{Z^*})$. Therefore, it follows from~(\ref{equation:RankDistance5}) that $z_i'\in\image(A_{Z'}-A_{Z^*})$, since $\lambda_{Z'}\gamma_2\ne\lambda_{Z^*}\gamma_2$.

To show that $A_{Z''}-A_{Z^*}$ is of full rank notice that
\begin{eqnarray}
\nonumber z_i (A_{Z''}-A_{Z^*})&=&(\lambda_{Z''}\gamma_2-\lambda_{Z^*})z_i\\
\nonumber z_i'(A_{Z''}-A_{Z^*})&=&(\lambda_{Z''}-\lambda_{Z^*}\gamma_2)z_i'\\
\nonumber z_i''(A_{Z''}-A_{Z^*})&=&\lambda_{Z''}\left((\gamma_1-\gamma_2)z_i+(\gamma_1-1)z_i'+\gamma_1z_i''\right)-\lambda_{Z^*}\gamma_1z_i''\\
\label{equation:RankDistance6}&=&\lambda_{Z''}(\gamma_1-\gamma_2)z_i+\lambda_{Z''}(\gamma_1-1)z_i'+(\lambda_{Z''}\gamma_1-\lambda_{Z^*}\gamma_1)z_i''
\end{eqnarray}
Since $\lambda_{Z''}\gamma_2\ne \lambda_{Z^*}$ and $\lambda_{Z''} \ne \lambda_{Z^*} \gamma_2$, we have that $z_i,z_i'\in\image(A_{Z''}-A_{Z^*})$. Therefore, it follows from~(\ref{equation:RankDistance6}) that $z_i'\in\image(A_{Z''}-A_{Z^*})$, since $\lambda_{Z''}\gamma_1\ne\lambda_{Z^*}\gamma_1$.
\subsection*{Appendix~B}\label{appendix:squareRankDistance}
To show that $A_Z^2-A_{Z^*}^2$ is of full rank, notice that
\begin{eqnarray}
\nonumber z_i' (A_Z^2-A_{Z^*}^2)&=&(\lambda_{Z}^2\gamma_2-\lambda_{Z^*}^2 \gamma_1)z_i'\\
\nonumber z_i''(A_Z^2-A_{Z^*}^2)&=&(\lambda_{Z}^2\gamma_1-\lambda_{Z^*}^2\gamma_2)z_i''\\
\nonumber z_i(A_Z^2-A_{Z^*}^2)&=&\lambda_Z^2(z_i+(1-\gamma_2)z_i'+(1-\gamma_1)z_i'')-\lambda_{Z^*}^2z_i\\
&=&(\lambda_Z^2-\lambda_{Z^*}^2)z_i+\lambda_Z^2(1-\gamma_2)z_i'+\lambda_Z^2(1-\gamma_1)z_i''\label{equation:SquareRankDistance2}
\end{eqnarray}
Since $\lambda_{Z}^2\gamma_2\ne \lambda_{Z^*}^2 \gamma_1$ and $\lambda_{Z}^2\gamma_1\ne \lambda_{Z^*}^2\gamma_2$, it follows that $z_i',z_i''\in \image (A_Z^2-A_{Z^*}^2)$. Therefore,~(\ref{equation:SquareRankDistance2}) implies that $z_i\in\image(A_Z^2-A_{Z^*}^2)$, since $\lambda_Z^2\ne \lambda_{Z^*}^2$.

To show that $A_Z^2-A_{Z''}^2$ is of full rank notice that
\begin{eqnarray}
\nonumber z_i' (A_Z^2-A_{Z''}^2)&=&(\lambda_{Z}^2\gamma_2-\lambda_{Z''}^2)z_i'\\
\nonumber (z_i+z_i'+z_i'')(A_Z^2-A_{Z''}^2)&=&(\lambda_{Z}^2-\lambda_{Z''}^2\gamma_2)(z_i+z_i'+z_i'')\\
\nonumber z_i(A_Z^2-A_{Z''}^2)&=&\lambda_Z^2(z_i+(1-\gamma_2)z_i'+(1-\gamma_1)z_i'')-\lambda_{Z''}^2 \gamma_1z_i\\
&=&\lambda_Z^2(z_i+z_i'+z_i'')-\lambda_Z^2 \gamma_2 z_i'-\lambda_Z^2 \gamma_1 z_i''-\lambda_{Z''}^2 \gamma_1 z_i
\label{equation:SquareRankDistance3}
\end{eqnarray}
Since $\lambda_Z^2 \gamma_2\ne \lambda_{Z''}^2$ and $\lambda_Z^2\ne \lambda_{Z''}^2\gamma_2$, it follows that $z_i',z_i+z_i'+z_i''\in\image(A_Z^2-A_{Z'}^2)$. Therefore,~(\ref{equation:SquareRankDistance3}) implies that $\lambda_Z^2 \gamma_1z_i''+\lambda_{Z''}^2\gamma_2z_i\in\image(A_Z^2-A_{Z''}^2)$. Hence, to have that $z_i,z_i',z_i''\in\image(A_Z^2-A_{Z'}^2)$, it suffices to show that $\lambda_Z^2\gamma_1\ne\lambda_{Z''}^2\gamma_2$, which is implied by $\lambda_Z^6\ne \lambda_{Z''}^6$.

To show that $A_{Z'}^2-A_{Z''}^2$ is of full rank notice that
\begin{eqnarray}
\nonumber z_i (A_{Z'}^2-A_{Z''}^2)&=&(\lambda_{Z'}^2\gamma_2-\lambda_{Z''}^2\gamma_1)z_i\\
\nonumber (z_i+z_i'+z_i'')(A_{Z'}^2-A_{Z''}^2)&=&(\lambda_{Z'}^2\gamma_1-\lambda_{Z''}^2\gamma_2)(z_i+z_i'+z_i'')\\
\nonumber z_i'(A_{Z'}^2-A_{Z''}^2)&=&\lambda_{Z'}^2\left((\gamma_1-\gamma_2)z_i+\gamma_1z_i'+(\gamma_1-1)z_i''\right)-\lambda_{Z''}^2z_i'\\
&=&\lambda_{Z'}^2\gamma_1(z_i+z_i'+z_i'')-\lambda_{Z'}^2\gamma_2z_i-\lambda_{Z'}^2z_i''-\lambda_{Z''}^2z_i'
\label{equation:SquareRankDistance4}
\end{eqnarray}
Since $\lambda_{Z'}^2 \gamma_2\ne \lambda_{Z''}^2$ and $\lambda_{Z'}^2\gamma_1\ne \lambda_{Z''}^2\gamma_2$, it follows that $z_i,z_i+z_i'+z_i''\in\image(A_{Z'}^2-A_{Z'}^2)$. Therefore,~(\ref{equation:SquareRankDistance4}) implies that $\lambda_{Z'}^2 z_i''+\lambda_{Z''}^2m_i'\in\image(A_{Z'}^2-A_{Z''}^2)$. Hence, to have that $m_i,m_i',m_i''\in\image(A_{Z'}^2-A_{Z''}^2)$, it suffices to show that $\lambda_{Z'}^2\ne\lambda_{Z''}^2$, which is implied by $\lambda_Z^6\ne \lambda_{Z''}^6$.

To show that $A_{Z'}^2-A_{Z^*}^2$ is of full rank notice that
\begin{eqnarray}
\nonumber z_i (A_{Z'}^2-A_{Z^*}^2)&=&(\lambda_{Z'}^2\gamma_2-\lambda_{Z^*}^2\gamma_1)z_i\\
\nonumber z_i''(A_{Z'}^2-A_{Z^*}^2)&=&(\lambda_{Z'}^2-\lambda_{Z^*}^2\gamma_2)m_i''\\
\nonumber z_i'(A_{Z'}^2-A_{Z^*}^2)&=&\lambda_{Z'}^2\left((\gamma_1-\gamma_2)z_i+\gamma_1z_i'+(\gamma_1-1)z_i''\right)-\lambda_{Z^*}^2\gamma_1z_i'\\
&=&\lambda_{Z'}^2(\gamma_1-\gamma_2)z_i+(\lambda_{Z'}^2\gamma_1-\lambda_{Z^*}^2\gamma_1)z_i'+\lambda_{Z'}^2(\gamma_1-1)z_i''
\label{equation:SquareRankDistance5}
\end{eqnarray}
Since $\lambda_{Z'}^2 \gamma_2\ne \lambda_{Z^*}^2\gamma_1$ and $\lambda_{Z'}^2\ne \lambda_{Z^*}^2\gamma_2$, it follows that $z_i,z_i''\in\image(A_{Z'}^2-A_{Z^*}^2)$. Therefore,~(\ref{equation:SquareRankDistance5}) implies that $m_i'\in\image(A_{Z'}^2-A_{Z^*}^2)$, since $\lambda_{Z'}^2\gamma_1\ne\lambda_{Z^*}^2\gamma_1$.

To show that $A_{Z''}^2-A_{Z^*}^2$ is of full rank notice that
\begin{eqnarray}
\nonumber z_i (A_{Z''}^2-A_{Z^*}^2)&=&(\lambda_{Z''}^2\gamma_1-\lambda_{Z^*}^2)z_i\\
\nonumber z_i'(A_{Z''}^2-A_{Z^*}^2)&=&(\lambda_{Z''}^2-\lambda_{Z^*}^2\gamma_1)z_i'\\
\nonumber z_i''(A_{Z''}^2-A_{Z^*}^2)&=&\lambda_{Z''}^2\left((\gamma_2-\gamma_1)z_i+(\gamma_2-1)z_i'+\gamma_2z_i''\right)-\lambda_{Z^*}^2\gamma_2z_i''\\
&=&\lambda_{Z''}^2(\gamma_2-\gamma_1)z_i+\lambda_{Z''}^2(\gamma_2-1)z_i'+(\lambda_{Z''}^2\gamma_2-\lambda_{Z^*}^2\gamma_2)z_i''
\label{equation:SquareRankDistance6}
\end{eqnarray}
Since $\lambda_{Z''}^2 \gamma_1\ne \lambda_{Z^*}^2$ and $\lambda_{Z''}^2\ne \lambda_{Z^*}^2\gamma_1$, it follows that $z_i,z_i'\in\image(A_{Z''}^2-A_{Z^*}^2)$. Therefore,~(\ref{equation:SquareRankDistance6}) implies that $z_i''\in\image(A_{Z'}^2-A_{Z^*}^2)$, since $\lambda_{Z''}^2\gamma_2\ne\lambda_{Z^*}^2\gamma_2$.

\subsection*{Appendix C}\label{appendixC}
\begin{lemma}\label{lemma:toThe6}
	For nonzero constants $\lambda_x,\lambda_y$ in $\bF_q$, if $\lambda_x^6\ne \lambda_y^6$ then \[\{\lambda_x,\gamma_1\lambda_x,\gamma_2\lambda_x\}\cap\{\lambda_y,\gamma_1\lambda_y,\gamma_2\lambda_y\}=\varnothing.\]
\end{lemma}
\begin{proof}
	For any $i,j\in[3]$, let $\lambda_x\gamma_1^i\in\{\lambda_x,\gamma_1\lambda_x,\gamma_2\lambda_x\}$ and $\lambda_y\gamma_1^j\in\{\lambda_y,\gamma_1\lambda_y,\gamma_2\lambda_y\}$. Since $\gamma_1$ is a root of unity of order 3, if $\lambda_x\gamma_1^i=\lambda_y\gamma_1^j$, then $(\lambda_x\gamma_1^i)^6=(\lambda_y\gamma_1^j)^6$, and hence $\lambda_x^6=\lambda_y^6$, a contradiction.
\end{proof}
\end{document}